\documentclass{article}

\usepackage{amsmath}
\usepackage{amsthm}
\usepackage{appendix}
\usepackage{bbold}
\usepackage{amssymb}

\usepackage{color}

\usepackage{hyperref}

\usepackage{tikz}
\usepackage{pgfplots}
\usepackage{pgf}
\usetikzlibrary{patterns}
\usepackage{float}
\usepackage{graphicx}

\usepackage{braket}

 \theoremstyle{plain}
  \newtheorem{thm}{Theorem}
  
  \newtheorem{lemma}[thm]{Lemma}
  \newtheorem{cor}[thm]{Corollary}
  \newtheorem{remark}[thm]{Remark}
  \newtheorem*{thm*}{Theorem}
  \newtheorem*{prop*}{Proposition}
  \newtheorem*{lemma*}{Lemma}
  \newtheorem*{cor*}{Corollary}
  \newtheorem*{remark*}{Remark}

\usepackage{hhline}
\newcommand{\specialcell}[2][c]{%
  \begin{tabular}[#1]{@{}l@{}}#2\end{tabular}}

\newcommand{\Tr}{\mbox{\rm Tr}}

\usepackage{algorithm}
\floatname{algorithm}{Protocol}
\usepackage{algcompatible}

\theoremstyle{definition}
\newtheorem{defn}[thm]{Definition}
 \newtheorem*{conj*}{Conjecture}

\usepackage{fullpage}

\usepackage{todonotes}

\usepackage{authblk}

\begin{document}

\title{Simple and tight device-independent security proofs}

\author[1]{Rotem Arnon-Friedman}
\author[1]{Renato Renner}
\author[2]{Thomas Vidick}
\affil[1]{Institute for Theoretical Physics, ETH-Z\"urich, CH-8093, Z\"urich, Switzerland}
\affil[2]{Department of Computing and Mathematical Sciences, California Institute of Technology, Pasadena, CA, USA}

\date{}

\maketitle

\begin{abstract}
	Device-independent security is the gold standard for quantum cryptography: not only is security based entirely on the laws of quantum mechanics, but it holds irrespective of any a priori assumptions on the quantum devices used in a protocol, making it particularly applicable in a quantum-wary environment. While the existence of device-independent protocols for tasks such as randomness expansion and quantum key distribution has recently been established, the underlying proofs of security remain very challenging, yield rather poor key rates, and demand very high-quality quantum devices, thus making them all but impossible to implement in practice. 

	We introduce a technique for the analysis of device-independent cryptographic protocols. We provide a flexible protocol and give a security proof that provides  quantitative bounds that are asymptotically  tight, even in the presence of general quantum adversaries.	At a high level our approach amounts to establishing a reduction to the scenario in which the untrusted device operates in an identical and independent way in each round of the protocol. This is achieved by leveraging the sequential nature of the protocol, and makes use of a newly developed tool, the ``entropy accumulation theorem'' of Dupuis \textit{et al.}~\cite{dupuis2016entropy}.
	
	As concrete applications we give simple and modular security proofs for device-independent quantum key distribution and randomness expansion protocols based on the CHSH inequality. For both tasks we establish essentially optimal asymptotic key rates and noise tolerance.
	In view of recent experimental progress, which has culminated in loophole-free Bell tests, it is likely that these protocols can be practically implemented in the near future.

\end{abstract}

\section{Introduction}

Classical cryptography relies on computational assumptions, such as the hardness of factoring, to deliver a wide range of functionalities, from secure communication to secure distributed computation and program obfuscation. The advent of quantum information in the 1980s brought forward a completely different possibility: security based only on the fundamental laws of physics. The quantum protocols for key distribution by Bennett and Brassard~\cite{bennett1984proceedings} and  Ekert~\cite{ekert1991quantum} allow mutually trustful users connected only by an authenticated classical channel, and an arbitrary quantum channel, to establish a private key whose security is guaranteed by the laws of quantum mechanics. With their private key, the users can then communicate with perfect security using, e.g., a one-time pad.

Quantum information is a double-edged sword. A typical protocol for quantum key distribution (QKD) requires the users, Alice and Bob, to manipulate quantum states: for example, in Ekert's protocol Alice has to prepare multiple entangled pairs of photons, and send one photon from each pair to Bob; both users then perform specific measurements on their respective photons in order to generate the classical key. The first proofs of security for QKD crucially relied on the fact that each user's internal operations were implemented in a specific way: the state preparation implemented by Alice, the measurements performed by Bob, all had to follow the low-level prescription given in the protocol. Initial implementations of QKD revealed how delicate these assumptions are. This is not only a question of the quality of the devices used. A wide range of side-channel attacks~\cite{fung2007phase,lydersen2010hacking,weier2011quantum,gerhardt2011full} were able to successfully exploit the very phenomena of quantum mechanics on which the security of QKD relies, 
 such as the no-cloning or uncertainty principles, to provide attacks that did not respect some of the assumptions made by the security proofs, that were difficult, if not impossible, to verify in practice (such as the assumption that Alice prepares a single pair of photons at a time, and not a more complex system with additional, undetected degrees of freedom that could leak information to an eavesdropper). 
 
The paradigm of device-independence offers an uncompromising solution to this conundrum. A cryptographic protocol is termed device-independent (DI) if its security guarantees hold irrespective of the quality, or trustworthiness, of the  physical devices used to implement the protocol~\cite{mayers1998quantum, barrett2005no} (see~\cite{ekert2014ultimate} for a perspective article). 
Security in such protocols should be based only on the statistics observed by the honest parties executing the protocol. In other words, any execution of the protocol should  contain a  ``proof'' that the generated key is secure, a proof that remains valid as long as very mild assumptions on the physical devices used are satisfied (informally, that no information is exchanged between the users' and eavesdroppers' laboratories, arguably an unavoidable requirement). 
	
	
Although the formulation of the DI paradigm appeared only much later in work of Mayers and Yao, the possibility for device-independence was arguably already present in Ekert's protocol. Ekert's intuition was to tie privacy of the users' key to the non-local effects that led to the generation of the key (measurement of a Bell pair). Ekert observed that quantum entanglement allows distant parties to generate bits that are correlated in such a strong way that it (seemingly) precludes any correlation with a third party~--- a phenomenon now known as the monogamy of correlations. 

The framework for the study of non-locality was put in place by Bell in the 1960s~\cite{bell1964einstein}. Motivated by questions in the foundations of quantum mechanics (including a proposal for an experiment that could in principle test the EPR ``paradox''~\cite{einstein1935can}), Bell introduced the notion of what is now known as a Bell inequality  (see~\cite{scarani2013device,brunner2014bell} for excellent reviews on the topic). In the context of device-independence we interpret a Bell inequality~\cite{bell1964einstein} as the specification of a small  game\footnote{For an explicit example of a game see Section~\ref{sec:non-local_games}.} that can be played by the honest parties using their respective quantum devices. What makes the game interesting is that it is designed in a way such that any classical strategy for the devices (i.e., any model for their actions that can be implemented as a convex combination of deterministic strategies) leads to a success probability $\omega_c$ in the game such that $\omega_c <1$. In contrast, there exists a quantum strategy (i.e., one in which the devices determine outcomes in the game by performing local measurements on a shared entangled state) that achieves a greater success probability, $\omega_q > \omega_c$. The use of such a game has the following major immediate consequence: if the honest parties observe that their devices are able to attain a success probability that is strictly larger than $\omega_c$, they can conclude that their devices must be non-classical --- the devices must share entanglement. 
This provides a first step in the implementation of the DI program: a statistical test that can be performed with the devices and that guarantees some element of quantumness. Early results in device-independence went further by establishing a quantitative relationship between the devices' success probability and the amount of secret randomness produced during the game~\cite{pironio2010random,acin2012randomness}, leading to a ``statistical test for information-theoretically secure randomness''~\cite{Colbeck09}, a task that is provably impossible to achieve using classical systems alone. 
	
	In the past decade an extended line of works has explored the application of the device-independence paradigm to multiple cryptographic tasks. A partial list includes QKD~\cite{barrett2005no,pironio2009device,vazirani2014fully}, randomness expansion~\cite{pironio2010random,vazirani2012certifiable,coudron2013infinite,miller2014robust} and amplification~\cite{colbeck2012free,gallego2013full,chung2014physical,brandao2016realistic,kessler2017device}, verified quantum computation~\cite{gheorghiu2015robustness,hajduvsek2015device,coladangelo2017verifier}, bit commitment~\cite{aharon2015device} and weak string erasure~\cite{kaniewski2016device}. For virtually all these tasks a proof of security ultimately amounts to bounding the knowledge that an adversary (a malicious party, or an eavesdropper) can gain about the output of the protocol. This knowledge, or uncertainty, is modeled using a notion of entropy called the smooth conditional min-entropy~\cite{Ren05}. 
	In the case of QKD, for example, the output is the raw key $K$, and proving security is essentially equivalent\footnote{From that point onward standard classical post-processing steps, e.g., error correction and privacy amplification, suffice to prove the security of the protocol; see Section~\ref{sec:diqkd} for the details.} to establishing a lower bound on the smooth conditional min-entropy $H_{\min}^{\varepsilon}(K|E)$, where $E$ is the quantum system held by Eve, which can be initially correlated to the device producing $K$ (for formal definitions see Section~\ref{sec:prelim}). 
	
	Evaluating the smooth min-entropy $H_{\min}^{\varepsilon}(K|E)$ of a large system is often difficult, especially in the DI setting where not much is known about the way $K$ is produced. 
	One assumption commonly used to simplify this task is that the bits of $K=K_1,\dotsc,K_n$ are created in an independent and identical way and hence $K$ itself is an independent and identically distributed (i.i.d.\@) random variable. That is, it is assumed that the device held by Alice and Bob makes the same measurements on the same quantum states in every round of the protocol. This means that the device is initialized with some (unknown) state which has a tensor product structure $\rho_{AB}^{\otimes n}$, and that the measurements have a tensor product structure as well. In that case, the total entropy in $K$ can be easily related to the sum of the entropies in each round separately.\footnote{Formally the bound can be calculated using the quantum asymptotic equipartition property~\cite{tomamichel2009fully} for example.} A bound on the entropy accumulated in one round can usually be derived using the expected winning probability in the game played in that round, which in turn can be easily estimated during the protocol in the i.i.d.\@ case using standard Chernoff-type bounds since the same game is just being played repeatedly with the same strategy.  
	 
	Unfortunately, even though quite convenient (and, in many cases, seemingly necessary) for the analysis, the i.i.d.\@ assumption is a very strong one in the DI scenario.
	In particular, under such an assumption the device cannot use any internal memory (i.e., its actions in one round cannot depend on the previous rounds) or even display time-dependent behavior (due to inevitable imperfections for example).

	Without this assumption about the device, however, not much is a priori known about the structure of $K$, the expected winning probability in one round of the protocol, nor the way the total entropy of $K$ is accumulated one round after the other (as the device might correlate the different rounds in an almost arbitrary way). Therefore, security proofs that estimated $H_{\min}^{\varepsilon}(K|E)$ directly for the most general case had to use far more complicated techniques and statistical analysis compared to the i.i.d.\@ case.\footnote{This led to non-optimal proofs, both readability- and parameter-wise (e.g., key rates or amount of tolerable noise). See Section~\ref{sec:related_work} for a discussion of related works.}

\subsection{Results and contributions}

We introduce a general framework, consisting of a flexible protocol and analysis, for obtaining DI proofs of security for a broad range of cryptographic tasks. Our technique takes advantage of the sequential nature of the protocol, as well as the specific way in which classical statistics are collected by users of the protocol, to establish a reduction to the i.i.d.\@ setting. A major advantage of our approach is that the reduction is virtually lossless in terms of parameters.	
Hence, our result establishes the a priori surprising fact that general quantum adversaries are no stronger than an adversary restricted to i.i.d.\@ attacks. 
As a consequence, we are able to extend tight results known for, e.g., DIQKD, under the i.i.d.\@ assumption, to the most general setting. This yields the best rates known for any protocol for a DI cryptographic task. 

To further discuss our results we state an informal version of our main theorem, that describes the entropy generation guarantees of our protocol (see Lemma~\ref{lem:main_soundness} for a formal statement, and Theorem~\ref{thm:main_generation_chsh} for the specialization of the protocol to the CHSH inequality).

\begin{thm}[Main theorem, informal]\label{thm:main_generation_informal}
Fix a choice of parameters, including an underlying non-local game, for Protocol~\ref{pro:randomness_generation}. Then there exist constants $c_1,c_2>0$ such that the following holds. 
Let $D$ be any device and $\rho_{|\Omega}$ the state generated using Protocol~\ref{pro:randomness_generation}, conditioned on the protocol not aborting. 
Then for any $\varepsilon_1,\varepsilon_2\in(0,1)$, either the protocol aborts with probability greater than $1-\varepsilon_1$ or
	\begin{equation}\label{eq:main_thm-inf}
		 H^{\varepsilon_2}_{\min} \left( \mathbf{A B} | \mathbf{X Y T F} E \right)_{\rho_{|\Omega}} > c_1 n - c_2\sqrt{n \log(1/\varepsilon_1\varepsilon_2)} \;.
	\end{equation}
\end{thm}

We remark that there are multiple implementations of devices that when used in an execution of Protocol~\ref{pro:randomness_generation} lead to a negligible probability of the protocol aborting (this is formalized in our completeness statement; see Section~\ref{sec:EA_completeness}). Importantly, devices that are within reach of current state-of-the-art technology also belong to this set of devices. Thus, Theorem~\ref{thm:main_generation_informal} gives a  non-trivial bound on the entropy produced by such devices. This was not achieved by previous works as discussed in Section~\ref{sec:related_work}.

Let us explain~\eqref{eq:main_thm-inf}. The registers $\mathbf{AB}$ contain the classical outputs generated by the device during the protocol. The registers $\mathbf{X Y T F}$ contain the classical inputs selected by the users, as well as auxiliary classical information exchanged during the protocol, that may be leaked to the adversary. $E$ is a quantum register that describes the adversary's quantum system, that may be correlated with the initial state of the devices. Thus,~\eqref{eq:main_thm-inf} gives a very precise bound on the amount of the smooth min-entropy present in the users' outputs at the end of the protocol, conditioned on all information available to the adversary. (As we discuss later, this formulation is flexible enough that it can be applied to obtain guarantees not only for the task of randomness generation, but also for quantum key distribution and other cryptographic tasks.)


We give explicit formulas for computing the constants $c_1$ and $c_2$ that appear in~\eqref{eq:main_thm-inf}, as a function of the parameters of the protocol (such as the fraction of rounds used for testing and the threshold value based on which the decision to accept or reject is made). 
Importantly, the constant $c_1$ that governs the leading-order term equals the optimal constant, i.e., the same leading constant that would be obtained under the i.i.d.\@ assumption, which by the asymptotic equipartition property is the Shannon entropy accumulated in one round of the protocol. Thus our result implies that general quantum adversaries do not force weaker rates compared to those achieved in less general scenarios. That is, it is possible to achieve rate vs.\@ noise tradeoffs which are as good as those achieved in much more restricted settings such as under the i.i.d.\@ assumption. 

To determine the constant $c_1$ the user of our result  must perform only one further crucial optimization: identify a so-called ``min-tradeoff function'',  a convex, differentiable function that lower bounds  the conditional Shannon entropy generated in a single round of the protocol, as a function of the game value. Informally, the requirement that the min-tradeoff function is differentiable and convex allows to account for lower-order fluctuations in the entropy generated that arise from finite statistics.  In Section~\ref{sec:entropy-chsh} we give a min-tradeoff function that can be used when the game that underlies the protocol is the CHSH game of Clauser et al.~\cite{clauser1969proposed}. Other use cases may require other min-tradeoff functions; indeed in Section~\ref{sec:related_work} below we survey recent works that applied our results to a variety of scenarios by computing an appropriate min-tradeoff function.

As already mentioned, beyond the first-order term in~\eqref{eq:main_thm-inf} our result also provides control over the constant~$c_2$ in front of the second-order term. Such control is a necessary condition for any application where finite values of $n$ need to be considered, such as in cryptography, and even more so quantum cryptography, where values of $n$ that can be achieved in practice remain relatively small. (See e.g. Figure~\ref{fig:qkd_rates_n_mod}, where one can see that finite-size effects can play an important role up to even moderately large values of $n\approx 10^{10}$.) 
As loophole-free Bell tests (a necessity for DI cryptography) are finally being realized~\cite{hensen2015loophole,shalm2015strong,giustina2015significant}, 
 the ability to derive essentially optimal values for $c_1$ and $c_2$ considerably decreases the gap between theory and experiments, thereby marking an important step towards practical DI protocols and their implementations.
	
We provide two concrete applications for Theorem~\ref{eq:main_thm-inf}. To begin with, we consider a DIQKD protocol based on the CHSH game, and prove its security. The achieved key rates and noise tolerance are significantly higher than in previous works. For large enough number of rounds $n$, the key rate as a function of the noise tolerance essentially coincides with the optimal result of~\cite{pironio2009device}, derived for the restricted i.i.d.\@ and asymptotic case. In particular, as in~\cite{pironio2009device}, we show that the protocol can tolerate up to the optimal error rate of~$7.1\%$ while still producing a positive key rate. (For comparison\footnote{The noise models of the two works are a bit different; the value of $1.6\%$ is the relevant one after equating the models.}, in~\cite{vazirani2014fully} the maximal noise tolerance was~$1.6\%$). 
	Moreover, the achieved key rates are comparable to those achieved in device-\emph{dependent} QKD protocols~\cite{scarani2008quantum,scarani2008security} already starting from $n=10^6$. (For further details and plots see Section~\ref{sec:qkd_curves}). 
		As a second application we consider a randomness expansion protocol based on the CHSH inequality. Here as well, we obtain an expansion rate which is essentially the same as the optimal rate achieved in~\cite{pironio2010random} in the case of \emph{classical} adversaries only, while our result holds against \emph{quantum} adversaries. This is much better than the rates obtained in previous works~\cite{vazirani2012certifiable,miller2014robust,miller2014universal}.
	
\paragraph*{Main ideas of the proof.}
As expressed earlier, the main difficulty in the analysis is to overcome the lack of any a priori independence assumptions on the quantum state shared by the users' devices, as well as a potential eavesdropper. Towards this we first leverage the sequential nature of the protocol. 
Our approach is to show that the random variables that model events observed by the users (such as the classical input/output behavior of their device in successive rounds) obey a natural Markov property. Using that property, we are able to apply a newly developed tool, the ``entropy accumulation theorem''~\cite{dupuis2016entropy} (EAT), to act as  a replacement for the chain rule for the conditional smooth min-entropy. The EAT allows us to quantify how entropy ``accumulates'' across many random variables generated through a certain iterative quantum processes as long as it fulfills a number of conditions that are tied to the Markov property (see Section~\ref{sec:eat-thm} for the exact statement). 
	As a result, we obtain a modular protocol that can be used as a ``skeleton'' for many DI cryptographic tasks; the protocol comes with fine-tuned guarantees on the entropy that is generated throughout, as a function of quantities that can be estimated based on the analysis of a single round of the protocol. 
	Next, we provide a concrete instantiation of the protocol based on the CHSH inequality. By combining the results of~\cite{pironio2009device}, derived for the i.i.d.\@ case, with our analysis of the general protocol we obtain a lower bound on the generated entropy rate when using the CHSH inequality as a basis for the protocol. 
	Finally, we apply our results to prove security of a DIQKD protocol that we propose, with essentially optimal key rate and noise tolerance. 

\subsection{Related and subsequent work}\label{sec:related_work}

	The idea of basing the security of cryptographic protocols (QKD especially) on the violation of Bell inequalities originates in the celebrated work of Ekert~\cite{ekert1991quantum}. Later, Mayers and Yao~\cite{mayers1998quantum} recognized that devices maximally violating a Bell inequality (they considered a variant of the CHSH inequality) could be fully characterized, up to local degrees of freedom, and thus need not be trusted a priori. Barrett \textit{et al.}~\cite{barrett2005no} were the first to combine both ideas together and derive a proof of security for QKD in the DI scenario. Their security proof holds even in the presence of a super-quantum adversary, limited only by the non-signalling principle. The protocol of~\cite{barrett2005no}, however, could not tolerate any amount of noise and produced just one secret bit when using the device many times (i.e. the key rate is zero). 
	
	Following these initial works a long line of research~\cite{acin2006bell,acin2006efficient,scarani2006secrecy,acin2007device,masanes2009universally,pironio2009device,hanggi2010efficient,hanggi2010device,masanes2011secure,masanes2014full} led to protocols, and proof techniques, that establish non-vanishing key rates with a positive noise tolerance in the i.i.d.\@ setting, against quantum or super-quantum adversaries (the former typically leading to better rates and noise tolerance). Most relevant for our work are the results of~\cite{pironio2009device}, where security of a DIQKD protocol was proven in the asymptotic limit, i.e., when the device is used $n\rightarrow\infty$ times, and under the i.i.d.\@ assumption described above. Their protocol is based on the CHSH inequality~\cite{clauser1969proposed}, and their analysis shows that it achieves the best possible rates under these assumptions.
	
	For the more challenging scenario presented by the non-i.i.d.\@ setting, security was first established in~\cite{vazirani2014fully}; see also~\cite{reichardt2013classical}, who give a secure protocol but with vanishing rate and no noise tolerance. A more recent proof of security by Miller and Shi~\cite{miller2014robust} is closest to our results in that it bounds the amount of entropy generated in the protocol in a round-by-round fashion, similar in spirit (but technically very different) from our use of the EAT (see Section~\ref{sec:eat-thm} for a description). The security proofs of the existing works are quite complex and achieve relatively low key rates and noise tolerance (if any). 
	
 Although it was introduced only much more recently than QKD, the first task to have received a complete proof of security in the DI setting is the task of randomness expansion. This task, first considered in~\cite{Colbeck09}, is the problem of expanding a short initial amount of seed randomness into a longer string that is information-theoretically random; aside of its practical relevance the task received attention because it is one the simplest problems that is classically impossible, yet for which quantum computing provides an information-theoretically secure solution. In the non-i.i.d.\@ setting it was shown in~\cite{pironio2010random} that a quadratic expansion was possible, but the analysis in that paper was limited to the case of classical adversaries. Security against quantum adversaries was established in~\cite{vazirani2012certifiable}, where it was shown that exponential expansion is possible. The analysis of~\cite{vazirani2012certifiable}, however, does not tolerate noise in the devices; subsequent work~\cite{miller2014robust} provided a different analysis that is able to tolerate a positive noise rate. 
	
	The maximum amount of randomness that can be generated from one system violating a specific Bell inequality by a given amount has been well-studied. In~\cite{pironio2010random} tight bounds for the CHSH game are obtained; see, e.g.,~\cite{dhara2013maximal,law2014quantum} for recent works exploring different aspects of the question. However, when using the device repeatedly, in the non-i.i.d.\@ setting, few works give explicit rates; to the best of our knowledge the only quantitative results available are from~\cite{miller2014universal} (see also~\cite{pironio2013security,fehr2013security} for an analysis in the non-i.i.d.\@ case but under the assumption that the adversary holds only classical side information), and remain relatively weak in comparison to the best one may expect from the known results under the i.i.d.\@ assumption.

	Since the initial announcement of our work in~\cite{arnon2018practical},\footnote{The publication~\cite{arnon2018practical} is an extended abstract that presents the main results reported in this submission, but has a much more limited discussion of applications, and only contains informal proof sketches.} our framework has already been applied to a variety of additional tasks, including conference key agreement~\cite{ribeiro2017fully}, randomness expansion and privatization~\cite{kessler2017device}, and randomness generation with sublinear quantum resources~\cite{bamps2017device}. Our results have been applied to the analysis of the first experimental implementations of a protocol for randomness generation in the fully DI framework~\cite{liu2017high,shen2018randomness}. In all these cases the difficulty consists in establishing a good min-tradeoff function by analyzing in detail a single round of the protocol used; our results then almost automatically imply the appropriate rate for the $n$-round protocol. More recently, the second-order terms in the EAT have been improved in~\cite{dupuis2018entropy}. 

\paragraph*{Structure of the paper.}
	The paper is organized as follows. We start with some preliminaries in Section~\ref{sec:prelim}. In Section~\ref{sec:ea_general} we show how the EAT can be used in DI protocols for a general Bell inequality. Then, in Section~\ref{sec:entropy-chsh} we explicitly calculate and plot the entropy rates for the case of the CHSH inequality. We continue in Sections~\ref{sec:diqkd} and~\ref{sec:expansion} with our  DIQKD and randomness expansion protocols, respectively. We end in Section~\ref{sec:open_questions} with some open questions.

\section{Preliminaries}\label{sec:prelim}

\subsection{General notation}\label{sec:pre_nota}

	All logarithms are in base 2. 
	Random variables (RV) are denoted by capital letters while specific values are denoted by small letters.
	We denote vectors in bold face; for example, $\mathbf{X} = X_1, \dotsc, X_n$ is a vector of RV. Sets are denoted with calligraphic fonts.
	
	The set $\{1,2,\dotsc,n\}$ is denoted by $[n]$. 
	
	Given a value $\mathbf{c}=c_1, \dotsc, c_n \in \mathcal{C}^n$, where $\mathcal{C}$ is a finite alphabet, we denote by $\mathrm{freq}_\mathbf{c}$ the probability distribution over $\mathcal{C}$ defined by $\mathrm{freq}_\mathbf{c}(\tilde{c}) = \frac{| \left\{ i | c_i = \tilde{c} \right\} |}{n}$ for $\tilde{c}\in\mathcal{C}$.

 We assume familiarity with the standard notation for quantum states and measurements; see~\cite{nielsen2002quantum} for a comprehensive introduction. We generally index pure quantum states or density matrices by the registers on which they are supported, e.g. $\rho_{AB}$ is a density matrix supported on the Hilbert space $\mathcal{H}_A\otimes\mathcal{H}_B$. If $\rho_{\mathbf{C}E}$ is a state classical on $\mathbf{C}$ we write $\Pr\left[\mathbf{c}\right]_{\rho}$ to denote the probability that $\rho$ assigns to $\mathbf{c}$. For $m\in \mathbb{N}_{+} $, $\rho_{U_m}$ denotes the completely mixed state on $m$ qubits and $\mathbb{I}$ is the identity operator. 
	
	Let $f: \mathcal{S} \rightarrow \mathbb{R}$ be a function over some set $\mathcal{S} \subset \mathbb{R}^{m}$. Then the infinity norm of the gradient of $f$ is defined as
	\begin{equation*}
		\|  \nabla f \|_\infty = \sup \left\{ \Big|\frac{\partial}{\partial x_{i}} f(\textbf{x}) \Big|: \textbf{x} \in \mathcal{S}, \, i \in \{1,\dots,m\} \right\} \,. 
	\end{equation*}
	
	For convenience all important parameters, constants, and random variables used in the paper are listed in the tables in the appendix.

\subsection{Entropies and Markov chains}

	\paragraph*{Entropies and conditional entropies.}
	$h$ is used for the binary entropy function $h(p)=-p\log (p) -(1-p)\log (1-p)$. 
	The von Neumann entropy $H(\rho)$ of a quantum state $\rho$ is given by $H(\rho)=-\Tr (\rho \log \rho)$. Given a bipartite state $\rho_{AE}\in\mathcal{H}_A\otimes\mathcal{H}_E$ the conditional von Neumann entropy is defined as $H(A|E)_{\rho_{AE}} = H(\rho_{AE}) -H(\rho_{E})$. When the state on which the entropy is evaluated is clear from the context we drop the subscript and write $H(A|E)$.
	
	\paragraph*{Min-entropy.}
	Given a state classical on $A$, $\rho_{AE}=  \sum_a p_a \ket{a}\!\!\bra{a}\otimes \rho_E^a$, the conditional min-entropy is 
	\[	
		H_{\text{min}}(A|E) = -\log p_{\text{guess}}(A|E)\;,
	\]
	where $p_{\text{guess}}(A|E)$ is the maximum probability of guessing $A$ given the quantum system $E$: 
	\[
		p_{\text{guess}}(A|E) = \max_{\{M^a_E\}_a} \, \sum_a p_a \Tr (M^a_E\rho^a_E)  \;,
	\]
	and the maximum is taken over all POVMs $\{M^a_E\}_a$ on $E$. 
	For any quantum state $\rho_{AE}$, $H(A|E) \geq H_{\text{min}}(A|E)$. 
	
	The smooth conditional min-entropy with smoothness parameter $\varepsilon$ of a state $\rho_{AE}$ is defined to be $H_{\text{min}}^\varepsilon(A|E)_{\rho_{AE}} = \max_{\sigma_{AE} \in \mathcal{B}^\varepsilon(\rho_{AE})} H_{\text{min}}(A|E)_{\sigma_{AE}}$,
	for $\mathcal{B}^\varepsilon(\rho_{AE})$ the set of sub-normalised states $\sigma_{AE}$ with $P(\rho_{AE},\sigma_{AE}) \leq \varepsilon$, where $P$ is the purified distance~\cite{tomamichel2010entropyduality}.
	
	\paragraph*{Max-entropy.}
	The quantum smooth max-entropy of a state $\rho_{AE}$ is given by 
	\[
		H^{\varepsilon}_{\max}(A|E)_{\rho_{AE}} = \log \inf_{\sigma_{AE} \in \mathcal{B}^\varepsilon(\rho_{AE})} \sup_{\tau_E} \| \sigma_{AE}^{\frac{1}{2}}\tau_E^{-\frac{1}{2}}\|_1^2 \;.
	\]
	We will also use the closely related $H^{\varepsilon}_{0}$ entropy. For classical $\mathbf{X}$ and $\mathbf{Y}$ distributed according to $\mathrm{P}_{\mathbf{X}\mathbf{Y}}$, $H_{0}(\mathbf{X}|\mathbf{Y})= \max_{\mathbf{y}} \log \left| \text{Supp}\left(\mathrm{P}_{\mathbf{X}|\mathbf{Y}=\mathbf{y}}\right)\right|$, where $\text{Supp}\left(\mathrm{P}_{\mathbf{X}|\mathbf{Y}=\mathbf{y}}\right)=\{ \mathbf{x} |\mathrm{P}_{\mathbf{X}|\mathbf{Y}=\mathbf{y}}\left(\mathbf{x}\right) > 0\}$. Its smooth version is given by 
	\[
		H^{\varepsilon}_{0}(\mathbf{X}|\mathbf{Y})=\min_{\Omega} \max_{\mathbf{y}} \log \left| \text{Supp}\left(\mathrm{P}_{\mathbf{X}|\Omega,\mathbf{Y}=\mathbf{y}}\right)\right| \;,
	\]
	where the minimum ranges over all events $\Omega$ with probability at least $1-\varepsilon$.

	\paragraph*{Markov chains.}
	A tripartite quantum state $\rho_{ABC}$ is said to fulfil the Markov chain condition $A\leftrightarrow B \leftrightarrow C$ if $I(A:C|B) = 0$, where $I(A:C|B)= H(AB) + H(BC) - H(B) - H(ABC)$ is the conditional mutual information. $I(A:C|B) = 0$ if and only if given $B$, $A$ and $C$ are independent.\footnote{There are also other equivalent ways of defining Markov chains for quantum states~\cite{hayden2004structure}, but for our purposes this definition suffices.}

\subsection{Non-local games}
\label{sec:non-local_games}

	We consider general two-player non-local games $G$.
	In a game $G$, the two players, Alice and Bob, share a bipartite quantum state. Given a question for Alice and a question for Bob, they can choose how to measure their parts of the state, and then use the measurements outcomes to supply an answer each. They win if their answers fulfil a pre-defined requirement, called the winning criterion.    
	
	More formally, a game $G$ is defined via sets of questions and answers for Alice and Bob, $\mathcal{X},\mathcal{Y}$ and $\mathcal{A},\mathcal{B}$, a distribution $\pi$ over $\mathcal{X}\times\mathcal{Y}$ (we will generally assume this is a product distribution), and a winning criterion $w:\mathcal{X}\times\mathcal{Y}\times\mathcal{A}\times\mathcal{B} \rightarrow \{0,1\}$.\footnote{A general Bell inequality would allow for an $\mathbb{R}$-valued $w$; we will not need this here.}
	
A strategy for the players in a game $G$ is specified by, first, a bipartite state $\rho_{Q_AQ_B}$, where Alice holds register $Q_A$ and Bob register $Q_B$, and second, local measurements that each player performs on his or her register in order to determine the answer to the given question.
	We use $\omega\in[0,1]$ to denote the winning probability of a strategy in the game~$G$.
	
	We sometimes use the equivalent language of Bell inequalities. The Bell functional associated to a nonlocal game is the linear function from $\mathbb{R}^{\mathcal{X}\times\mathcal{Y}\times\mathcal{A}\times\mathcal{B}}$ to $\mathbb{R}$ that maps a tuple $p$ to $\sum_{x,y,a,b}  \pi(x,y)w(x,y,a,b)p(x,y,a,b)$. In this language, the quantum value of the game is also called the  largest violation of the Bell inequality, i.e., the largest value attained by the Bell functional when evaluated on tuples $p$ that correspond to conditional distributions that can be realized by performing local measurements on an entangled state. 
	
	\paragraph{The CHSH game.}
	We use a variant of the CHSH game previously used in~\cite{pironio2009device,vazirani2014fully} in the context of DIQKD. In this game Alice has two possible inputs $\mathcal{X}=\{0,1\}$ and Bob three possible inputs $\mathcal{Y}=\{0,1,2\}$. The output sets are $\mathcal{A}=\mathcal{B}=\{0,1\}$. The input distribution $\pi_{\text{CHSH}}$ is uniform on $\mathcal{X}\times\mathcal{Y}$. The winning condition is the following:\footnote{The value of $w_{\text{CHSH}}$ for the inputs $(x,y)=(1,2)$ is left undefined, as it is never used.}
	\[
		w_{\text{CHSH}} = \begin{cases}
			1 & x,y\in\{0,1\} \text{ and } a\oplus b = x \cdot y \\
			1 & (x,y)=(0,2) \text{ and } a=b \\
			0 & \text{otherwise.}
		\end{cases}
	\]
	The optimal quantum strategy for this game is the same as in the standard CHSH game~\cite{clauser1969proposed}, except that if Bob's input is a $2$ he applies the same measurement as Alice's measurement on input~$0$. Since the underlying state is maximally entangled this ensures that their outputs will always match when $(x,y)=(0,2)$. 
	
	Conditioned on Bob's input not being $2$, the game played is the CHSH game. The optimal quantum strategy in the CHSH game achieves winning probability $\omega = \frac{2+\sqrt{2}}{4}\approx 0.85$, while the optimal classical strategy achieves a winning probability of $0.75$.
	
	Instead of describing the quantum advantage in the CHSH game in terms of the winning probability one can also work with the correlation coefficients defined by: $E_{xy} =\Pr[a=b|x,y]-\Pr[a\neq b|xy]$ for any pair of inputs $(x,y)$. The CHSH value is then given by $\beta = E_{00} + E_{01} + E_{10} - E_{11}$. 
	The relation between the winning probability in the CHSH game and the CHSH value is given by $\omega=1/2+\beta/8$. The largest values that these quantities can take in the classical case are $\beta=2$ and $\omega = \frac{3}{4}$, and the optimal quantum are $\beta = 2\sqrt{2}$ and $\omega = \frac{2+\sqrt{2}}{4}$.

\subsection{Untrusted device}\label{sec:untrusted_devices}

	In a DI protocol the honest parties interact with an \emph{untrusted device}. We now explain what is meant by this term and what are the assumptions regarding such a device. For simplicity we consider the case of two honest parties, Alice and Bob, but this can be extended to more parties in the obvious way.
	
	A device $D$ is modelled by a tripartite apparatus (including both state and measurements devices), distributed between Alice, Bob, and the adversary Eve. We think of the device as being prepared by Eve, and hence we call it untrusted. This allows Eve, in particular, to keep a purification of Alice and Bob's quantum state in a quantum register in her possession.\footnote{We emphasise that Eve is not required to measure her quantum state at any particular point. During the run of the considered protocol, Eve can eavesdrop on all the classical communication between the honest parties, and can later choose to measure her quantum register depending on this information.} 
	Although the device is untrusted we always assume that the following requirements hold (some of these requirements can be verified). 
	
	\paragraph*{The device can be used to run the considered protocol.} That is, Alice and Bob can interact with $D$ according to the relevant protocol (for an example of a protocol, see Protocol~\ref{pro:randomness_generation} below). Alice and Bob's components of $D$ implement the protocol by making sequential measurements on quantum states. In each round of the protocol, we say that the device is implementing some strategy for the game $G$ being played. The device may have memory, and thus apply a different strategy each time the game is played, depending on the previous rounds. Therefore, the measurement operators may change in each round, and the state on which the measurements are performed may be the post-measurement state from the previous round, a new state, or any combination of these two. 
	
	We sometimes use the terminology \emph{honest device} or \emph{honest implementation}. A device is said to be honest if it implements the protocol by using a certain pre-specified strategy. In that case, the actions of the device are known and fixed (noise can still be present). 
	
	\paragraph*{Communication (signalling) between the components of the device.} The communication between Alice, Bob, and Eve's components is restricted in the following way:
	\begin{enumerate}
		\item Alice and Bob's components of $D$ cannot signal to Eve's component. \label{it:com_eve}
		\item Alice and Bob can decide when to allow communication (if any) between their components. This ensures that the underlying quantum state of Alice and Bob's components of the device is (at least) bipartite and that the measurements made in the two components, in each round, are in tensor product with one another. \label{it:com_alice_bob}
		\item Alice and Bob can decide when to receive communication (if any) from Eve's component.\label{it:from_eve} 
	\end{enumerate}
	
	The requirement given in Item~\ref{it:com_eve} is necessary for DI cryptography; without it the device could directly send to Eve all the raw data it generated.
	
	Item~\ref{it:com_alice_bob} implies that Alice and Bob's component must be (at least) bipartite. This is necessary to assure that the violation of the considered Bell inequality is meaningful and implies security.

	
	Items~\ref{it:com_alice_bob} and~\ref{it:from_eve} give Alice, Bob, and Eve's components the possibility to communicate in certain stages of the protocol. This is not a restrictive nor necessary assumption. This possibility to communicate is added since it   is advantageous to actual implementations of certain protocols. To be specific, we consider the following scenario. \emph{In-between} different rounds of the protocol, Alice and Bob's components of the device are allowed to communicate freely. During the execution of a single round, however, no communication is allowed. In particular, when the game is being played, there is no communication between the components once the honest parties' inputs are chosen and until the outputs are supplied by the device.\footnote{To be more precise and concrete, in Protocol~\ref{pro:randomness_generation} for example, communication is allowed in every round $i$ right after Step~\ref{prostep:measurement} is done, and until the beginning of round $i+1$, i.e., before $T_{i+1}$ is chosen.} Furthermore, in-between rounds Eve may send information to the device, but not receive any from it.
	In actual implementations this implies that entanglement can be distributed ``on the fly'' for each round of the protocol, instead of maintaining large quantum memories. 
	
	\paragraph*{Other assumptions.} Apart from the above description of the untrusted device, we assume the following other standard assumptions used in DI cryptography:
	\begin{enumerate}
		\item The honest parties' physical locations are secure (unwanted information cannot leak outside to Eve or between their devices).
		\item The honest parties have a trusted random number generator.
		\item The honest parties have trusted classical post-processing units to make the necessary (classical) calculations during the protocol.
		\item There is an  authenticated, but public, classical channel connecting the honest parties (if necessary).
		\item Quantum physics is correct.
	\end{enumerate}

\subsection{Security definitions}\label{sec:security_defs}

	\paragraph{DIQKD.}
	 A DIQKD protocol (see Section~\ref{sec:diqkd} for a description of an explicit protocol) consists of an interaction between two trusted parties, Alice and Bob, and an untrusted device as defined in Section~\ref{sec:untrusted_devices}. 
	At the end of the protocol each party outputs a key, $\tilde{K}_A$ for Alice and $\tilde{K}_B$ for Bob.  
	The goal of the adversary, Eve, is to gain as much information as possible about Alice and Bob's keys without being detected (i.e., in the case where the protocol is not being aborted). 
	
	Correctness, secrecy, and overall security of a protocol are defined as follows (see also~\cite{portmann2014cryptographic,beaudry2015assumptions}):
	
	\begin{defn} [Correctness]
		A DIQKD protocol is said to be $\varepsilon_{corr}$-correct, when implemented using a device $D$, if Alice and Bob's keys,  $\tilde{K}_A$ and $\tilde{K}_B$ respectively, are identical with probability at least $1-\varepsilon_{corr}$. That is, $\Pr ( \tilde{K}_A\neq \tilde{K}_B ) \leq \varepsilon_{corr}$.
	\end{defn}
	
	\begin{defn}[Secrecy]
		A DIQKD protocol is said to be $\varepsilon_{sec}$-secret, when implemented using a device $D$, if for a key of length $l$, $\left(1-\Pr[\text{abort}]\right) \| \rho_{\tilde{K}_A E} - \rho_{U_l} \otimes \rho_{E} \|_1 \leq \varepsilon_{sec}$, where $E$ is a quantum register that may initially be correlated with $D$.
	\end{defn}
	$\varepsilon_{sec}$ in the above definition can be understood as the probability that some non-trivial information leaks to the adversary~\cite{portmann2014cryptographic}. 
	
	If a protocol is $\varepsilon_{corr}$-correct and $\varepsilon_{sec}$-secret (for a given $D$), then it is $\varepsilon_{\mathrm{QKD}}^s$-correct-and-secret for any $\varepsilon_{\mathrm{QKD}}^s\geq \varepsilon_{corr}+\varepsilon_{sec}$.
	
	\begin{defn}[Security]\label{def:security_QKD}
		A DIQKD protocol is said to be $(\varepsilon_{\mathrm{QKD}}^s,\varepsilon_{\mathrm{QKD}}^c,l)$-secure if:
		\begin{enumerate}
			\item (Soundness) For \emph{any} implementation of the device $D$ it is $\varepsilon_{\mathrm{QKD}}^s$-correct-and-secret.
			\item (Completeness) There exists an honest implementation of the device $D$ such that the protocol aborts with probability at most $\varepsilon_{\mathrm{QKD}}^c$. 
		\end{enumerate}
	\end{defn}
	
	The protocols that we consider below take into account possible noise in the honest implementation. That is, even when there is no adversary at all, the actual implementation of the devices might not be perfect. Thus, the \emph{completeness} of the protocol implies its \emph{robustness} to the desired amount of noise.
	
	Lastly, a remark regarding the composability of this security definition is in order.  A security definition is said to be composable~\cite{canetti2001universally,ben2004general,portmann2014cryptographic} if it implies that the protocol can be used arbitrarily and composed with other protocols (proven secure by themselves), without compromising security. 
	Obviously, if Alice and Bob wish to use the keys they produced in the DIQKD protocol in some other cryptographic protocol (i.e., they compose the two protocols), it is necessary for them to use protocols which were proven to have composable security. 
	
	For the case of (device-\emph{dependent}) QKD, Definition~\ref{def:security_QKD} was rigorously proven to be composable~\cite{portmann2014cryptographic}. This suggests that the same security definition should also be the relevant one in the DI context and, indeed, as far as we are aware, it is the definition that has been used in all prior works on DI cryptography. Nevertheless, the claim that Definition~\ref{def:security_QKD} is composable for DI protocols as well has never been rigorously proven, and the result of~\cite{barrett2013memory} suggests that this is not the case when the same devices are reused in the composition. We still use this definition as it seems like the most promising security definition to date. This implies that, as in all other works, \emph{after the end of the protocol} the device cannot be used again in general~\cite{barrett2013memory}.
	
	\paragraph{Randomness expansion.} In the task of randomness expansion there is a single user interacting sequentially with an untrusted device. At the start of the interaction the user is presented with a source $R\in\{0,1\}^r$ of uniformly random bits. The user then interacts sequentially with the device in a deterministic way (the only sources of randomness being the initial string $R$ and any randomness which may be present in the devices' outputs). At the end of the protocol the user returns a string $Z\in\{0,1\}^m$ of $m$ bits that is statistically close to uniform, conditioned on $R$ as well as any side information of the adversary. (See Section~\ref{sec:expansion} for a concrete example of a randomness expansion protocol.) More formally, we require the following.

	\begin{defn}[Security of randomness expansion]
		A protocol is called an $(\varepsilon^c_{RE},\varepsilon_{RE}^s)$-secure $r\to m$ randomness expansion protocol\footnote{All parameters $\varepsilon^c_{RE},\varepsilon_{RE}^s,r$ and $m$ will in general be function of a parameter $n$ that also parametrises the protocol and the number of rounds of interactions between the user and the device.} if, provided as input $r$ uniformly random bits:
	\begin{enumerate}
		\item (Soundness) For any implementation of the device $D$ the protocol either aborts or returns a classical string $Z\in\{0,1\}^m$ and we have
		\[
			\left( 1- \Pr[\text{abort}]\right)\| \rho_{ZRE} - \rho_{U_m} \otimes \rho_{RE} \|_1 \leq \varepsilon_{RE}^s \;,
		\]
		where $E$ is a quantum register that may initially be correlated with $D$.
		\item (Completeness)  There exists an honest implementation of the device such that the protocol aborts with probability at most $\varepsilon_{RE}^c$.
	\end{enumerate}
	\end{defn}
	
	As in the case of DIQKD, this security definition was not proven to be composable in general.
\subsection{The entropy accumulation theorem}
\label{sec:eat-thm}

	The main tool used in this work is the EAT~\cite[Theorem 4.4]{dupuis2016entropy}. Below we give the necessary details in a notation appropriate for our work (although less general than the original EAT). 
	
	We work with channels with the following properties:
	\begin{defn}[EAT channels]\label{def:eat_channels}
		EAT channels $\mathcal{N}_i:R_{i-1}\rightarrow R_i A_i B_i I_i C_i$, for $i\in[n]$, are CPTP maps such that for all $i\in [n]$:
		\begin{enumerate}
			\item $A_i,B_i,I_i$ and $C_i$ are finite-dimensional classical systems (RV). $A_i$ and $B_i$ or of dimension $d_{A_i}$ and $d_{B_i}$ respectively. $R_i$ are arbitrary quantum registers. 
			\item For any input state $\sigma_{R_{i-1}R'}$, where $R'$ is a register isomorphic to $R_{i-1}$, the output state $\sigma_{R_i A_i B_i I_i C_i R'} = \left( \mathcal{N}_i \otimes \mathbb{I}_{R'} \right) \left(\sigma_{R_{i-1}R'} \right)$ has the property that the classical value $C_i$ can be measured from the marginal $\sigma_{A_i B_i I_i}$ without changing the state. 
			\item For any initial state $\rho_{R_0 E}^0$, the final state $\rho_{\mathbf{ABIC}E} = \left( \Tr_{R_n}\circ \mathcal{N}_n \circ \dots \circ \mathcal{N}_1\right) \otimes \mathbb{I}_E \; \rho_{R_0E}^0$ fulfils the Markov chain condition $A_{1\dotsc i-1}B_{1\dotsc i-1} \leftrightarrow I_{1\dotsc i-1} E \leftrightarrow I_i$ for each $i\in[n]$. 
		\end{enumerate}
	\end{defn}
	
	\begin{defn}[Tradeoff functions]\label{def:min_tradeoff_func}
		Let $\mathcal{N}_1,\ldots,\mathcal{N}_N$ be a family of EAT channels. Let $\mathcal{C}$ denote the common alphabet of $C_1,\ldots,C_n$. 	
		A differentiable and convex function $f_{\min}$ from the set of probability distributions $p$ over $\mathcal{C}$ to the real numbers is called a \emph{min-tradeoff function} for $\{\mathcal{N}_i\}$ if it satisfies\footnote{The infimum and supremum over the empty set are defined as plus and minus infinity, respectively.} 
		\[
			f_{\min}(p) \leq \inf_{\sigma_{R_{i-1}R'}:\mathcal{N}_i(\sigma)_{C_i}=p} H\left( A_i B_i | I_i R' \right)_{\mathcal{N}_i(\sigma)} \;
		\]
		for all $i\in [n]$, where the infimum is taken over all input states of $\mathcal{N}_i$ for which the marginal on $C_i$ of the output state is the probability distribution $p$.   
		
		Similarly, a differentiable and concave function $f_{\max}$ from the set of probability distributions $p$ over $\mathcal{C}$ to the real numbers is called a \emph{max-tradeoff function} for $\{\mathcal{N}_i\}$ if it satisfies 
		\[
			f_{\max}(p) \geq \sup_{\sigma_{R_{i-1}R'}:\mathcal{N}_i(\sigma)_{C_i}=p} H\left( A_i B_i | I_i R' \right)_{\mathcal{N}_i(\sigma)} \;
		\]
		for all $i\in [n]$, where the supremum is taken over all input states of $\mathcal{N}_i$ for which the marginal on $C_i$ of the output state is the probability distribution $p$.  
	\end{defn}

	\begin{thm}[EAT~\cite{dupuis2016entropy}]\label{thm:eat}
		Let $\mathcal{N}_i:R_{i-1}\rightarrow R_i A_i B_i I_i C_i$ for $i\in [n]$ be EAT channels as in Definition~\ref{def:eat_channels}, $\rho_{\mathbf{ABIC}E} = \left( \Tr_{R_n}\circ \mathcal{N}_n \circ \dots \circ \mathcal{N}_1\right) \otimes \mathbb{I}_E \; \rho_{R_0E}$ be the final state, 
		$\Omega$ an event  defined over $\mathcal{C}^n$, $p_\Omega$ the probability of $\Omega$ in $\rho$, 
		and $\rho_{|\Omega}$ the final state conditioned on $\Omega$. Let $\varepsilon_{\text{s}} \in (0,1)$.
		 
		For $f_{\min}$ a min-tradeoff function for $\{\mathcal{N}_i\}$, as in Definition~\ref{def:min_tradeoff_func}, and any $t\in \mathbb{R}$ such that $f_{\min}\left( \mathrm{freq}_\mathbf{c} \right) \geq t$ for any $\mathbf{c}\in\mathcal{C}^n$ for which $\Pr\left[\mathbf{c}\right]_{\rho_{|\Omega}}> 0$,
		\[
			H_{\min}^{\varepsilon_{\text{s}}} \left( \mathbf{AB}|\mathbf{I}E \right)_{\rho_{|\Omega}} > n t - v\sqrt{n} \;,
		\]
		where $v = 2\left(\log(1+2 d_{A_iB_i} ) + \lceil \|  \triangledown f_{\min} \|_\infty \rceil \right)\sqrt{1-2\log (\varepsilon_{\text{s}} \cdot p_\Omega)}$ and $d_{A_iB_i}$ denotes the dimension of $A_iB_i$.
		
		Similarly, for $f_{\max}$ a max-tradeoff function for  $\{\mathcal{N}_i\}$ as in Definition~\ref{def:min_tradeoff_func} and any $t\in \mathbb{R}$ such that $f_{\max}\left(\mathrm{freq}_\mathbf{c} \right) \leq t$ for any $\mathbf{c}\in\mathcal{C}^n$ for which $\Pr\left[\mathbf{c}\right]_{\rho_{|\Omega}}> 0$,
		\[
			H_{\max}^{\varepsilon_{\text{s}}} \left( \mathbf{AB}|\mathbf{I}E \right)_{\rho_{|\Omega}} < n t + v\sqrt{n} \;,
		\]
		where $v = 2\left(\log(1+2 d_{A_iB_i} ) + \lceil \|  \triangledown f_{\max} \|_\infty \rceil \right)\sqrt{1-2\log (\varepsilon_{\text{s}} \cdot p_\Omega)}$.
	\end{thm}

	To gain a bit of intuition on how Theorem~\ref{thm:eat} is going to be used note the following. The event $\Omega$ will usually be the event of the considered protocol not aborting (or a closely related event). The relevant state for which the smooth min- or max-entropy is going to be evaluated is $\rho_{|\Omega}$. To use the theorem, it should be possible to define \emph{some} EAT channels $\{\mathcal{N}_i\}$ that produce the final state $\rho$ from the initial state $\rho_{R_0}$ by applying the channels sequentially; these channels are not necessarily the channels used in the actual protocol to produce $\rho$. The tradeoff functions can be seen as a bound on the entropy accumulated in \emph{one} round $i$, and, if such a bound $t$ exists, then Theorem~\ref{thm:eat} asserts that the total amount of entropy, accumulated in all rounds $i=1$ to $n$ together, is roughly $n$ times $t$. It is in this sense that the theorem essentially allows us to perform a reduction to the  i.i.d.\@ setting.

\section{Device-independent entropy accumulation protocol}\label{sec:ea_general}

The main task in proving security of DIQKD and other protocols is to prove a bound on the (smooth) min-entropy of the raw data held by Alice and Bob, conditioned on all the information available to the adversary Eve. The goal of this section is to show how the EAT (Theorem~\ref{thm:eat}) can be used in a general DI setting to achieve such a bound.  

For this we consider the entropy accumulation protocol, described as Protocol~\ref{pro:randomness_generation} below. Although we call it a ``protocol'',  one should see it more as a mathematical tool which allows us to use the EAT rather than an actual protocol to be implemented.\footnote{In particular, in a setting with two distinct parties, Alice and Bob, communication is required to actually implement it. We ignore this here as it is not relevant for the analysis.} 
To be more specific, the EAT channels (as in Definition~\ref{def:eat_channels}) will be defined via the steps made in the entropy accumulation protocol. 
The relevance of the protocol stems from the fact that the final state at the end of the protocol, on which a smooth min-entropy bound can be proven using the EAT, is the same state as (or can easily be related to) the final state in the actual protocol to be executed (depending on the specific application).    

\subsection{The protocol}

Protocol~\ref{pro:randomness_generation} is used to generate raw data for Alice and Bob by using an untrusted device $D$. 
It is based on an arbitrary non-local game $G$ as defined in Section~\ref{sec:non-local_games}, together with a definition of test and generation inputs for Alice and Bob. The test inputs, $\mathcal{X}_t \subset \mathcal{X}$ and $\mathcal{Y}_t \subset \mathcal{Y}$, are used by the parties during the test rounds ($T_i=1$ below) from which the Bell violation is estimated, while the generation inputs, $\mathcal{X}_g \subset \mathcal{X}$ and $\mathcal{Y}_g \subset \mathcal{Y}$, are used in the other rounds (the sets are not necessarily disjoint). We also assume that $\mathcal{X}_g\subset \mathcal{X}_t$, as it is important that, given a value in $\mathcal{X}_g$, the device is not able to infer the value of $T_i$. Ideally, one should use a game $G$ for which Alice and Bob's outputs are perfectly correlated (or anti-correlated) with sufficiently high probability when the parties use the generation inputs.\footnote{In a DIQKD protocol (or other tasks with two separated honest parties) this requirement is used to ensure a good key rate, as the output bits in the generation rounds will be the main contributors to the final key. For tasks such as randomness expansion, where there is only one honest party, it is not necessary to generate matching outputs.}

\begin{algorithm}
\caption{Entropy accumulation protocol}
\label{pro:randomness_generation}
\begin{algorithmic}[1]
	\STATEx \textbf{Arguments:} 
		\STATEx\hspace{\algorithmicindent} $G$ -- two-player non-local game
		\STATEx\hspace{\algorithmicindent} $\mathcal{X}_g\subset \mathcal{X}_t \subset \mathcal{X}$ -- generation and test inputs for Alice
		\STATEx\hspace{\algorithmicindent} $\mathcal{Y}_g,\mathcal{Y}_t \subset \mathcal{Y}$ -- generation and test inputs for Bob
		\STATEx\hspace{\algorithmicindent} $D$ -- untrusted device of (at least) two components that can play $G$ repeatedly
		\STATEx\hspace{\algorithmicindent} $n \in \mathbb{N}_+$ -- number of rounds
		\STATEx\hspace{\algorithmicindent} $\gamma \in (0,1]$ -- expected fraction of test rounds 
		\STATEx\hspace{\algorithmicindent} $\omega_{\mathrm{exp}}$ -- expected winning probability in $G$ for an honest (perhaps noisy) implementation    
		\STATEx\hspace{\algorithmicindent} $\delta_{\mathrm{est}} \in (0,1)$ -- width of the statistical confidence interval for the estimation test
		
	\STATEx
	
	\STATE For every round $i\in[n]$ do Steps~\ref{prostep:choose_Ti}-\ref{prostep:calculate_Ci_EA}:
		\STATE\hspace{\algorithmicindent}Alice chooses $T_i\in\{0,1\}$ at random such that $\Pr(T_i=1)=\gamma$, and sends her choice of $T_i$ to Bob over a public authenticated classical channel. 
		\STATE\hspace{\algorithmicindent}If $T_i=0$ Alice and Bob choose inputs $X_i\in\mathcal{X}_g$ and $Y_i\in \mathcal{Y}_g$ respectively. If $T_i=1$ they choose inputs $X_i\in\mathcal{X}_t$ and $Y_i\in \mathcal{Y}_t$. 
		\STATE\hspace{\algorithmicindent}Alice and Bob use $D$ with $X_i,Y_i$ and record their outputs as $A_i$ and $B_i$ respectively. \label{prostep:measurement}
		\STATE\hspace{\algorithmicindent}(Optional symmetrisation step:) Alice and Bob choose together a (random) value $F_i$, and respectively update their outputs $A_i,B_i$ depending on $F_i$. \label{prostep:symmetry} 
		\STATE\hspace{\algorithmicindent}If $T_i=0$ then Bob updates $B_i$ to $B_i = \perp$, and they set $C_i=\perp$. If $T_i=1$ they set $C_i =w\left(A_i,B_i,X_i,Y_i\right)$.\label{prostep:calculate_Ci_EA}
	\STATE Alice and Bob abort if $\sum_i C_i < \left(\omega_{\mathrm{exp}}\gamma - \delta_{\mathrm{est}}\right) \cdot n\;$. \label{prostep:abort_general_EA}
\end{algorithmic}
\end{algorithm}

We now define the EAT channels using the rounds of the protocol (where one round includes Steps~\ref{prostep:choose_Ti}-\ref{prostep:calculate_Ci_EA} in Protocol~\ref{pro:randomness_generation}).  For this, the following notation is used. For every $i\in\{0\}\cup[n]$, the (unknown) quantum state of the device $D$ shared by Alice and Bob after round $i$ of the protocol is denoted by $\rho_{Q_AQ_B}^i$. We denote the register holding this state by $R_i$. In particular, $R_0\equiv Q_AQ_B$ at the start of the protocol. At Step~\ref{prostep:measurement} in Protocol~\ref{pro:randomness_generation}, the quantum state of the devices is changed from $\rho_{Q_AQ_B}^{i-1}$ in $R_{i-1}$ to $\rho_{Q_AQ_B}^{i}$ in $R_i$ by the use of the device.\footnote{To be a bit more precise, the quantum state is changed in two steps. First, the relevant measurement of Step~\ref{prostep:measurement} is done (where it is assumed that the measurements of the different components are in tensor product). Then, after $A_i$ and $B_i$ are recorded, the different components of the device are allowed to communicate. Thus, some further changes can be made to the post-measurement state even based on the memory of all components together.}
Our EAT channels are then $\mathcal{N}_i:R_{i-1}\rightarrow R_i A_i B_i X_i Y_i T_i C_i$ defined by the CPTP map describing the $i$-th round of Protocol~\ref{pro:randomness_generation}, as implemented by the untrusted device $D$ (see Figure~\ref{fig:maps}). We prove in Lemma~\ref{lem:main_soundness} below that they indeed satisfy the conditions given in Definition~\ref{def:eat_channels}. 

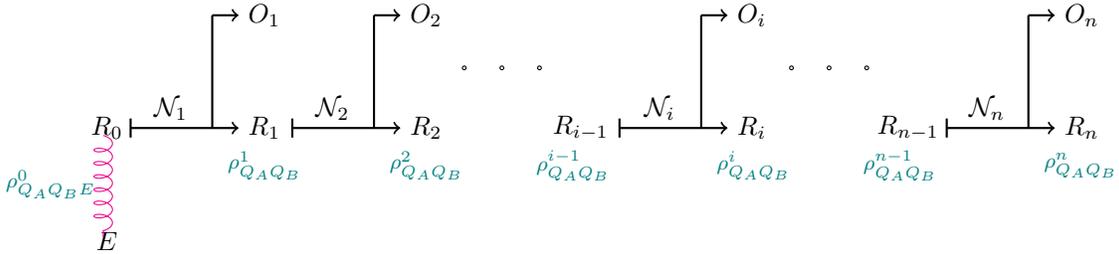
\begin{figure}
	\centering
	\begin{tikzpicture}
		
		\draw (0,1.5) node {$R_0$};
		\draw (0,0) node {$E$};
		\draw[decorate, magenta, decoration={coil,amplitude=3.5pt, segment length=5pt}] (-0.05,1.4) -- (-0.05,0.05);
		\draw[teal] (-0.75,0.75) node[font=\footnotesize] {$\rho^0_{Q_AQ_BE}$};
		
		\begin{scope}[shift={(2,0)}]
			\draw[thick,|->] (-1.7,1.5) -- (-0.25,1.5) node[anchor=west] {$R_1$};
			\draw[thick] (-0.6,1.5) -- (-0.6,3);
			\draw[thick, ->] (-0.6,3) -- (-0.25,3) node[anchor=west] {$O_1$};
			\draw (-1.15,1.75) node {$\mathcal{N}_1$};
			
			\draw[teal] (0.1,1) node[font=\footnotesize] {$\rho^1_{Q_AQ_B}$};
		\end{scope}
		
		\begin{scope}[shift={(4.15,0)}]
			\draw[thick,|->] (-1.7,1.5) -- (-0.25,1.5) node[anchor=west] {$R_2$};
			\draw[thick] (-0.6,1.5) -- (-0.6,3);
			\draw[thick, ->] (-0.6,3) -- (-0.25,3) node[anchor=west] {$O_2$};
			\draw (-1.15,1.75) node {$\mathcal{N}_2$};
			
			\draw[teal] (0.1,1) node[font=\footnotesize] {$\rho^2_{Q_AQ_B}$};
		\end{scope}
		
		\begin{scope}[shift={(4.75,0)}]
			\draw (0,2.3) circle (0.25mm);
			\draw (0.5,2.3) circle (0.25mm);
			\draw (1,2.3) circle (0.25mm);
		\end{scope}
		
		\begin{scope}[shift={(8.5,0)}]
			\draw (-2.2,1.5) node {$R_{i-1}$};
			\draw[thick,|->] (-1.7,1.5) -- (-0.25,1.5) node[anchor=west] {$R_i$};
			\draw[thick] (-0.6,1.5) -- (-0.6,3);
			\draw[thick, ->] (-0.6,3) -- (-0.25,3) node[anchor=west] {$O_i$};
			\draw (-1.15,1.75) node {$\mathcal{N}_i$};
			
			\draw[teal] (-2.3,1) node[font=\footnotesize] {$\rho^{i-1}_{Q_AQ_B}$};
			\draw[teal] (0.1,1) node[font=\footnotesize] {$\rho^i_{Q_AQ_B}$};
		\end{scope}
		
		\begin{scope}[shift={(9.1,0)}]
			\draw (0,2.3) circle (0.25mm);
			\draw (0.5,2.3) circle (0.25mm);
			\draw (1,2.3) circle (0.25mm);
		\end{scope}
		
		\begin{scope}[shift={(12.85,0)}]
			\draw (-2.2,1.5) node {$R_{n-1}$};
			\draw[thick,|->] (-1.7,1.5) -- (-0.25,1.5) node[anchor=west] {$R_n$};
			\draw[thick] (-0.6,1.5) -- (-0.6,3);
			\draw[thick, ->] (-0.6,3) -- (-0.25,3) node[anchor=west] {$O_n$};
			\draw (-1.15,1.75) node {$\mathcal{N}_n$};
			
			\draw[teal] (-2.3,1) node[font=\footnotesize] {$\rho^{n-1}_{Q_AQ_B}$};
			\draw[teal] (0.1,1) node[font=\footnotesize] {$\rho^n_{Q_AQ_B}$};.
		\end{scope}

	\end{tikzpicture}
\caption{The EAT channels $\mathcal{N}_i:R_{i-1}\rightarrow R_i A_i B_i X_i Y_i T_i C_i$. In the figure, $O_i = A_i B_i X_i Y_i T_i C_i$. The initial quantum state shared by Alice, Bob, and Eve is $\rho^0_{Q_AQ_BE}$ and the sequence of maps $\mathcal{N}_i$ creates the state $\rho^n_{Q_AQ_BE\mathbf{O}}$.} \label{fig:maps}
\end{figure}

In the following we are interested in the state of Alice, Bob, and Eve after the $n$-th round of the protocol, both before and after Alice and Bob decide whether to abort or not in Step~\ref{prostep:abort_general_EA}. The state \emph{before} Step~\ref{prostep:abort_general_EA} is denoted by
\begin{equation}\label{eq:final_state_before_abort}
	\rho_{\mathbf{ABXYTC}E} = \left( \Tr_{R_n}\circ \mathcal{N}_n \circ \dots \circ \mathcal{N}_1\right) \otimes \mathbb{I}_E \; \rho_{Q_AQ_BE}^0 \;.
\end{equation}

In Step~\ref{prostep:abort_general_EA} Alice and Bob decide whether they should abort the protocol or not according to the estimated Bell violation in the test rounds. Let $\Omega$ denote the event that they do not abort\footnote{Note that $C_j \in\{0,1,\perp\}$; the quantity $\sum_jC_j$ should be understood as $\sum_{j|C_j=1}1$.}, i.e.,
\begin{equation}\label{eq:good_event_def}
	\Omega =  \Big\{\sum_j C_j \geq \left(\omega_{\mathrm{exp}}\gamma - \delta_{\mathrm{est}}\right) \cdot n\;\Big\}\;.
\end{equation}

The final state, \emph{conditioned on not aborting}, is denoted by $\rho_{\mathbf{ABXYTC}E|\Omega}$ or just $\rho_{|\Omega}$ to ease notation.  Below we bound the entropy which is accumulated in this state during the rounds of the protocol.  

\subsection{Completeness}
\label{sec:EA_completeness}

Suppose that Alice and Bob execute Protocol~\ref{pro:randomness_generation} with a device $D$ which performs i.i.d.\@ measurements on a tensor product state $\rho_{Q_AQ_B}^{\otimes n}$ such that the winning probability achieved in game $G$ by the device $D$ executed on a single state $\rho_{Q_AQ_B}$ is $\omega_{\mathrm{exp}}$. We call any such implementation an \emph{honest implementation}. The following lemma bounds the probability of Protocol~\ref{pro:randomness_generation} aborting in an honest implementation.

\begin{lemma}\label{lem:ea_completeness}
	Protocol~\ref{pro:randomness_generation} is complete with completeness error $\varepsilon^c_{EA} \leq \exp(- 2 n \delta_{\mathrm{est}}^2 ) $.
	That is, the probability that the protocol aborts for an honest implementation of the devices $D$ is at most $\varepsilon^c_{EA}$.
\end{lemma}

\begin{proof}
	Alice and Bob abort in Step~\ref{prostep:abort_general_EA} when the sum of the $C_i$ is not sufficiently high (this happens when the estimated Bell violation is too low or when not enough test rounds were chosen). In the honest implementation $C_i$ are i.i.d.\@ RVs with $\mathbb{E}\left[ C_i\right] = \omega_{\mathrm{exp}}\gamma$.  Therefore, we can use Hoeffding's inequality:
	\begin{equation}\label{eq:completeness_error_EA}
		\varepsilon^c_{EA} =  \Pr \left[ \sum_j C_j \geq \left(\omega_{\mathrm{exp}}\gamma - \delta_{\mathrm{est}}\right) \cdot n\ \right] \leq \exp(- 2 n \delta_{\mathrm{est}}^2 )  \;. \qedhere
	\end{equation}
\end{proof}

\subsection{Soundness}
\label{sec:EA_soundness}

The EAT, Theorem~\ref{thm:eat}, almost immediately provides a general lower bound on the amount of entropy generated by Protocol~\ref{pro:randomness_generation}. We state the result as Lemma~\ref{lem:main_soundness} below; in Section~\ref{sec:entropy-chsh} we will obtain a more refined bound based on an instantiation of the protocol with the game $G$ taken to be the CHSH game. 

\begin{lemma}\label{lem:main_soundness}
	Let $D$ be any device, and for $i\in [n]$ let $I_i = X_iY_iT_iF_i$ and $\mathcal{N}_i:R_{i-1}\rightarrow R_i A_i B_i I_i C_i$ the CPTP map implemented by the $i$-th round of Protocol~\ref{pro:randomness_generation}. Let $\rho$ be the state generated by the protocol (as defined in Equation~\eqref{eq:final_state_before_abort}), $\Omega$ the event that the protocol does not abort (as defined in Equation~\eqref{eq:good_event_def}), and $\rho_{|\Omega}$ the state conditioned on $\Omega$. Let $f_{min}$ be a real-valued differentiable function defined on the set of probability distributions $p$ over the alphabet $\{\perp,0,1\}$ of $C_i$ such that
	\begin{equation}\label{eq:eat_f_min_bound}
		\forall i\in[n]\qquad	f_{\min}(p) \leq \inf_{\sigma_{R_{i-1}R'}:\mathcal{N}_i(\sigma)_{C_i}=p} H\left( A_i B_i | X_i Y_i T_i F_i R' \right)_{\mathcal{N}_i(\sigma)} \;,
	\end{equation}
	where the infimum over an empty set is defined as infinity. 	
	Then, for any $\varepsilon_{\mathrm{EA}},\varepsilon_{\text{s}}\in (0,1)$, either the protocol aborts with probability $1-\Pr(\Omega)\geq 1-\varepsilon_{\mathrm{EA}}$ or,
	\begin{equation}\label{eq:main_lb}
		 H^{\varepsilon_{\text{s}}}_{\min} \left( \mathbf{A B} | \mathbf{X Y T F} E \right)_{\rho_{|\Omega}} > nt- v\sqrt{n},
	\end{equation}
	  	where 
	\begin{align*}
		t &= \min_{p:\, p(1) \geq\omega_{\mathrm{exp}}\gamma-\delta_{\mathrm{est}}} \, f_{\min}(p) \;,\\
		v &= 2\left(\log(1+2 d_{A_iB_i} ) + \lceil \|  \triangledown f_{\min} \|_\infty \rceil \right)\sqrt{1-2\log (\varepsilon_{\text{s}} \cdot\varepsilon_{\mathrm{EA}})}\;
	\end{align*}
	and $d_{A_iB_i}$ denotes the dimension of $A_iB_i$.
\end{lemma}

\begin{proof}
In order to apply the EAT we first verify that  the conditions stated in Definition~\ref{def:eat_channels} are fulfilled.  Using that $C_i$ is a function of $A_i,B_i,X_i,$ and $Y_i$ the first two conditions in Definition~\ref{def:eat_channels} clearly hold. Moreover, the  Markov chain condition
	\[
		\forall i\in[n], \quad A_{1\dotsc i} B_{1\dotsc i} \leftrightarrow X_{1\dotsc i}Y_{1\dotsc i}T_{1\dotsc i}F_{1\dotsc i}E \leftrightarrow X_{i+1}Y_{i+1}T_{i+1}F_{i+1} \
	\]
	holds as well since the values of $X_{i+1}$, $Y_{i+1}$, $T_{i+1}$, and $F_{i+1}$ are chosen independently of everything else at each round. To conclude, note the event $\Omega$ of the protocol not aborting implies that the fraction of successful game rounds $\mathrm{freq}_\mathbf{c}(1)$ is at least $\omega_{\mathrm{exp}}\gamma- \delta_{\mathrm{est}}$ for any $\mathbf{c}$ for which $\Pr\left[\mathbf{c}\right]_{\rho_{|\Omega}}> 0$. \qedhere
\end{proof}

The main work remaining for a successful use of Protocol~\ref{pro:randomness_generation} for entropy generation consists in obtaining a good lower bound in Equation~\eqref{eq:main_lb}, i.e., devising an appropriate min-tradeoff function $f_{\min}$ satisfying Equation~\eqref{eq:eat_f_min_bound}. In order to understand the task to be accomplished note that $\mathcal{N}_i$ defines $X_i, Y_i, T_i$, and $F_i$, so although the infimum in Equation~\eqref{eq:eat_f_min_bound} is taken over all states $\sigma$ the distributions of $X_i, Y_i,T_i$, and $F_i$ are fixed. Moreover, the infimum is only taken over states with $\mathcal{N}_i(\sigma)_{C_i}=p$, a condition which fixes the Bell violation achieved by $\sigma$ under the bipartite measurement performed by the device.
This is precisely the sense in which the EAT can be understood as providing a reduction to the i.i.d.\@ case. 

Lower bounds of the form of Equation~\eqref{eq:eat_f_min_bound} of different quality can be obtained depending on the specific Bell inequality employed in the protocol. A general method consists in using the chain rule to write  
\begin{align}
	H\left( A_i B_i | X_i Y_i T_i F_iR' \right)_{\mathcal{N}_i(\sigma)} &= H\left( A_i | X_i Y_i T_i F_iR' \right)_{\mathcal{N}_i(\sigma)} + H\left( B_i | X_i Y_i T_i F_iR' A_i \right)_{\mathcal{N}_i(\sigma)} \label{eq:one_round_chain_rule} \\
	&\geq H_{\min}\left( A_i | X_i Y_i T_i F_iR' \right)_{\mathcal{N}_i(\sigma)} \nonumber 
\end{align}
Note that here the random variable $F_i$ depends on the (optional) symmetrisation step, and was introduced precisely to enable an easier lower bound on the quantities above; we will show how it can be used in the specific case of the CHSH game in the next section.

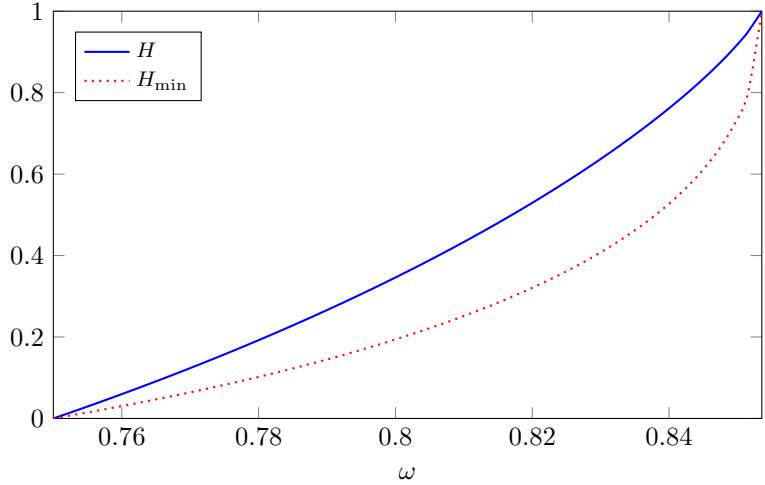
\begin{figure}
\centering
\begin{tikzpicture}
	\begin{axis}[
		height=7cm,
		width=11cm,
		xlabel=$\omega$,
		xmin=0.75,
		xmax=0.853553,
		ymax=1,
		ymin=0,
	     xtick={0.76,0.78,0.80,0.82,0.84},
          ytick={0,0.2,0.4,0.6,0.8,1},
		legend style={at={(0.12,0.95)},anchor=north,legend cell align=left,font=\footnotesize} 
	]
	

	\addplot[blue,thick,smooth] coordinates {
	(0.75, 0.) (0.752071, 0.0120347) (0.754142, 0.0242374) (0.756213, 0.0366113) (0.758284, 0.0491598) (0.760355, 0.0618861) (0.762426, 0.074794) (0.764497, 0.0878872) (0.766569, 0.10117) (0.76864, 0.114646) (0.770711, 0.128319) (0.772782, 0.142196) (0.774853, 0.156279) (0.776924, 0.170575) (0.778995, 0.185089) (0.781066, 0.199826) (0.783137, 0.214793) (0.785208, 0.229996) (0.787279, 0.245441) (0.78935, 0.261137) (0.791421, 0.277091) (0.793492, 0.293311) (0.795563, 0.309806) (0.797635, 0.326586) (0.799706, 0.343661) (0.801777, 0.361042) (0.803848, 0.378741) (0.805919, 0.396771) (0.80799, 0.415147) (0.810061, 0.433884) (0.812132, 0.452998) (0.814203, 0.47251) (0.816274, 0.49244) (0.818345, 0.51281) (0.820416, 0.533648) (0.822487, 0.554982) (0.824558, 0.576846) (0.82663, 0.599279) (0.828701, 0.622324) (0.830772, 0.646033) (0.832843, 0.670469) (0.834914, 0.695705) (0.836985, 0.721832) (0.839056, 0.748965) (0.841127, 0.777251) (0.843198, 0.806888) (0.845269, 0.838156) (0.84734, 0.871481) (0.849411, 0.907587) (0.851482, 0.948007) (0.853553, 1.)
	};
	\addlegendentry{$H$}
	
	\addplot[red,thick,smooth,dotted] coordinates {
	(0.75, 0.) (0.752071, 0.00603838) (0.754142, 0.0122045) (0.756213, 0.0185027) (0.758284, 0.0249374) (0.760355, 0.0315132) (0.762426, 0.0382352) (0.764497, 0.0451088) (0.766569, 0.0521395) (0.76864, 0.0593332) (0.770711, 0.0666965) (0.772782, 0.0742361) (0.774853, 0.0819592) (0.776924, 0.0898735) (0.778995, 0.0979875) (0.781066, 0.10631) (0.783137, 0.114851) (0.785208, 0.12362) (0.787279, 0.132628) (0.78935, 0.141889) (0.791421, 0.151414) (0.793492, 0.161218) (0.795563, 0.171316) (0.797635, 0.181726) (0.799706, 0.192467) (0.801777, 0.203558) (0.803848, 0.215022) (0.805919, 0.226885) (0.80799, 0.239174) (0.810061, 0.251921) (0.812132, 0.265161) (0.814203, 0.278934) (0.816274, 0.293285) (0.818345, 0.308265) (0.820416, 0.323936) (0.822487, 0.340367) (0.824558, 0.357638) (0.82663, 0.375848) (0.828701, 0.395113) (0.830772, 0.415575) (0.832843, 0.437409) (0.834914, 0.460839) (0.836985, 0.486151) (0.839056, 0.513728) (0.841127, 0.544096) (0.843198, 0.578018) (0.845269, 0.616671) (0.84734, 0.662045) (0.849411, 0.718043) (0.851482, 0.794796) (0.853553, 1.)
	};
	\addlegendentry{$H_{\min}$}
			
	\end{axis}  
\end{tikzpicture}

\caption{The lower bounds on $H\left( A_i | X_i Y_i T_i R'_{i-1} \right)$ and $H_{\min}\left( A_i | X_i Y_i T_i R'_{i-1} \right)$ as a function of the Bell violation for the CHSH inequality. The bound on $H\left( A_i | X_i Y_i T_i R'_{i-1} \right)$ is given in Equation~\eqref{eq:one_box_entropy_final} while the bound on $H_{\min}\left( A_i | X_i Y_i T_i R'_{i-1} \right)$ can be taken from~\cite{masanes2011secure}; both bounds are asymptotically tight. For non-optimal Bell violation the min-entropy is significantly lower than the entropy.}\label{fig:h_vs_min}
\end{figure}

A bound using the min-entropy $H_{\min}$, instead of $H$ itself, is not tight in general, and one can expect to lose quite a lot by performing the relaxation above (see for example Figure~\ref{fig:h_vs_min}). The advantage, however, is that a lower bound on $H_{\min}\left( A_i | X_i Y_i T_i F_iR' \right)_{\mathcal{N}_i(\sigma)}$ can be found using general techniques based on the semidefinite programming (SDP) hierarchies of~\cite{navascues2008convergent}. 
For a slightly better bound one should not drop the second term in Equation~\eqref{eq:one_round_chain_rule}. A bound on $H_{\min}\left( A_i B_i | X_i Y_i T_i F_i R'  \right)_{\mathcal{N}_i(\sigma)}$ (usually called ``global randomness'') can also be found using the SDP hierarchies (see, e.g.,~\cite{pironio2010random}).  For further details and references see~\cite[Section IV-C]{brunner2014bell}.

\section{A bound for the CHSH game}\label{sec:entropy-chsh}

In this section we devise a specific min-tradeoff function $f_{min}$ which, through an application of Lemma~\ref{lem:main_soundness}, leads to a concrete bound on the entropy generated by Protocol~\ref{pro:randomness_generation} when the game $G$ is the CHSH game (described in Section~\ref{sec:non-local_games}). 

We use Protocol~\ref{pro:randomness_generation} with the following choices: $\mathcal{X}_g=\{0\}, \; \mathcal{X}_t=\{0,1\}, \;\mathcal{Y}_g=\{2\}$, and $\mathcal{Y}_t=\{0,1\}$.

In order to fully specify the protocol it suffices to describe the symmetrisation step.  In this step, Alice and Bob choose together a uniform bit $F_i$, and they both flip their output bits if and only if $F_i=1$.
This symmetrisation is helpful in the proof of the main theorem below. The downside is that it costs a lot of randomness to implement, which can be problematic for some applications such as randomness expansion. At the end of the section we show that the step is in fact not necessary in any real implementation of the protocol.

The proof of Theorem~\ref{thm:main_generation_chsh} shows that the rate of entropy generation is governed by the following functions, where $h$ is the binary entropy and~$\gamma,p(1)\in(0,1]$ such that $ p(1)/\gamma\geq 3/4$:\footnote{We define the functions $g$ and  $f_{min}$ only in the regime in which the protocol does not abort, i.e., $ p(1)/\gamma\geq 3/4$. Any extension of $g$ to the regime $p(1)/\gamma\in[0,3/4]$ that keeps the function differential can be used for mathematical completeness.}
\begin{align}
	&g(p) =  \begin{cases} 
			 1 - h\left( \frac{1}{2} + \frac{1}{2}\sqrt{16 \; \frac{p(1)}{\gamma} \left(\frac{p(1)}{\gamma}-1\right) +3}  \right)&  \frac{p(1)}{\gamma}\in\left[\frac{3}{4},\frac{2+\sqrt{2}}{4}\right] \\
			1 & \frac{p(1)}{\gamma}\in\left[\frac{2+\sqrt{2}}{4},1\right]\;,
			\end{cases}\notag\\
	&f_{\min}\left(p,p_t\right) = \begin{cases}
	g\left(p\right)&  p(1) \leq p_t(1) \;  \\
	\frac{\mathrm{d}}{\mathrm{d}p(1)} g(p)\big|_{p_t}  \cdot p(1)+ \Big( g(p_t) -	\frac{\mathrm{d}}{\mathrm{d}p(1)} g(p)\big|_{p_t} \cdot p_t(1) \Big)& p(1)> p_t(1)\;,
	\end{cases} \nonumber \\
	&\eta(p,p_t,\varepsilon_{\text{s}},\varepsilon_{\text{e}}) =  f_{\min}\left(p, p_t\right) - \frac{1}{\sqrt{n}}2\left( \log 13 + \Big\lceil\frac{\mathrm{d}}{\mathrm{d}p(1)} g(p)\big|_{p_t} \Big\rceil \right)\sqrt{1-2 \log (\varepsilon_{\text{s}} \cdot \varepsilon_{\text{e}})}\;, \nonumber\\
	&\eta_{\mathrm{opt}}(\varepsilon_{\text{s}}, \varepsilon_{\text{e}}) = \max_{\frac{3}{4} < \frac{p_t(1)}{\gamma} < \frac{2+\sqrt{2}}{4}} \; \eta(\omega_{\mathrm{exp}}\gamma - \delta_{\mathrm{est}},p_t,\varepsilon_{\text{s}},\varepsilon_{\text{e}})\;. \label{eq:eta_opt}
\end{align}

\begin{thm}[Main theorem]\label{thm:main_generation_chsh}
Let $D$ be any device, $\rho$ the state (as defined in Equation~\eqref{eq:final_state_before_abort}) generated using Protocol~\ref{pro:randomness_generation} when the game $G$ is the CHSH game, $\Omega$ (as defined in Equation~\eqref{eq:good_event_def}) the event that the protocol does not abort, and $\rho_{|\Omega}$ the state conditioned on $\Omega$.
Then, for any $\varepsilon_{\mathrm{EA}},\varepsilon_{\text{s}}\in (0,1)$, either the protocol aborts with probability greater than $1-\varepsilon_{\mathrm{EA}}$ or
	\begin{equation}\label{eq:main_thm}
		 H^{\varepsilon_{\text{s}}}_{\min} \left( \mathbf{A B} | \mathbf{X Y T F} E \right)_{\rho_{|\Omega}} > n\cdot \eta_{\mathrm{opt}}(\varepsilon_{\text{s}},\varepsilon_{\mathrm{EA}}) \;,
	\end{equation}
	where  $\eta_{\mathrm{opt}}$ is defined in Equation~\eqref{eq:eta_opt}. 
\end{thm}

The rate $\eta_{\mathrm{opt}}$ as a function of the expected Bell violation $\omega_{\mathrm{exp}}$ is plotted in Figure~\ref{fig:eta_rates} for $\gamma=1$ and several choices of values for $\varepsilon_{\mathrm{EA}}$, $\delta_{\mathrm{est}}$, and $n$. 
For comparison, we also plot  in Figure~\ref{fig:eta_rates} the asymptotic rate ($n\rightarrow\infty$) under the assumption that the state of the device is an (unknown) i.i.d.\@ state $\rho_{Q_AQ_BE}^{\otimes n}$. In this case, the quantum asymptotic equipartition property~\cite[Theorems~1 and~9]{tomamichel2009fully} implies that the optimal rate is the Shannon entropy accumulated in one round of the protocol (as given in Equation~\eqref{eq:one_box_entropy_final}). This rate, appearing as the dashed line in Figure~\ref{fig:eta_rates}, is an upper bound on the entropy that can be accumulated. One can see that as the number of rounds in the protocol increases our rate $\eta_{\mathrm{opt}}$ approaches this optimal rate.

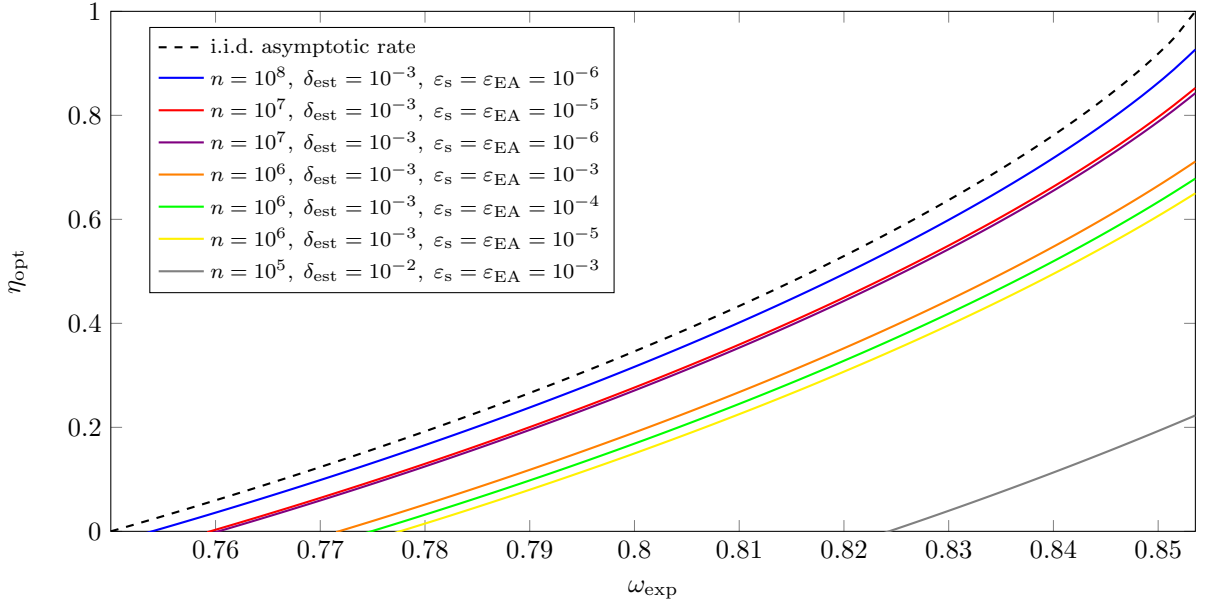
\begin{figure}
\centering
\begin{tikzpicture}
	\begin{axis}[
		height=8.5cm,
		width=16cm,
		xlabel=$\omega_{\mathrm{exp}}$,
		ylabel=$\eta_{\mathrm{opt}}$,
		xmin=0.75,
		xmax=0.853553,
		ymax=1,
		ymin=0,
	     xtick={0.76,0.77,0.78,0.79,0.80,0.81,0.82,0.83,0.84,0.85},
          ytick={0,0.2,0.4,0.6,0.8,1},
		legend style={at={(0.25,0.97)},anchor=north,legend cell align=left,font=\footnotesize} 
	]
	

	\addplot[black,thick,smooth,dashed] coordinates {
	(0.75, 0.) (0.752071, 0.0120347) (0.754142, 0.0242374) (0.756213, 0.0366113) (0.758284, 0.0491598) (0.760355, 0.0618861) (0.762426, 0.074794) (0.764497, 0.0878872) (0.766569, 0.10117) (0.76864, 0.114646) (0.770711, 0.128319) (0.772782, 0.142196) (0.774853, 0.156279) (0.776924, 0.170575) (0.778995, 0.185089) (0.781066, 0.199826) (0.783137, 0.214793) (0.785208, 0.229996) (0.787279, 0.245441) (0.78935, 0.261137) (0.791421, 0.277091) (0.793492, 0.293311) (0.795563, 0.309806) (0.797635, 0.326586) (0.799706, 0.343661) (0.801777, 0.361042) (0.803848, 0.378741) (0.805919, 0.396771) (0.80799, 0.415147) (0.810061, 0.433884) (0.812132, 0.452998) (0.814203, 0.47251) (0.816274, 0.49244) (0.818345, 0.51281) (0.820416, 0.533648) (0.822487, 0.554982) (0.824558, 0.576846) (0.82663, 0.599279) (0.828701, 0.622324) (0.830772, 0.646033) (0.832843, 0.670469) (0.834914, 0.695705) (0.836985, 0.721832) (0.839056, 0.748965) (0.841127, 0.777251) (0.843198, 0.806888) (0.845269, 0.838156) (0.84734, 0.871481) (0.849411, 0.907587) (0.851482, 0.948007) (0.853553, 1)	
	};
	\addlegendentry{i.i.d.\@ asymptotic rate}
	
	\addplot[blue,thick,smooth] coordinates {
	(0.753051, -0.00519263) (0.755102, 0.00675606) (0.757153, 0.0188591) (0.759204, 0.0311305) (0.761255, 0.0435732) (0.763306, 0.0561908) (0.765357, 0.0689867) (0.767409, 0.0819646) (0.76946, 0.0951283) (0.771511, 0.108482) (0.773562, 0.12203) (0.775613, 0.135776) (0.777664, 0.149726) (0.779715, 0.163884) (0.781766, 0.178256) (0.783817, 0.192847) (0.785868, 0.207663) (0.787919, 0.222711) (0.78997, 0.237996) (0.792021, 0.253526) (0.794072, 0.269309) (0.796123, 0.285352) (0.798175, 0.301665) (0.800226, 0.318255) (0.802277, 0.335134) (0.804328, 0.352312) (0.806379, 0.369801) (0.80843, 0.387613) (0.810481, 0.405761) (0.812532, 0.424261) (0.814583, 0.44313) (0.816634, 0.462384) (0.818685, 0.482044) (0.820736, 0.502133) (0.822787, 0.522674) (0.824838, 0.543696) (0.82689, 0.56523) (0.828941, 0.587313) (0.830992, 0.609986) (0.833043, 0.633297) (0.835094, 0.657305) (0.837145, 0.682076) (0.839196, 0.707695) (0.841247, 0.734266) (0.843298, 0.761921) (0.845349, 0.790836) (0.8474, 0.821253) (0.849451, 0.853531) (0.851502, 0.888251) (0.853553, 0.926527)
	};
	\addlegendentry{$n=10^8,\; \delta_{\mathrm{est}}=10^{-3},\;\varepsilon_{\text{s}}=\varepsilon_{\mathrm{EA}}=10^{-6}$}

	\addplot[red,thick,smooth] coordinates {
	(0.759204, -0.00167058) (0.761255, 0.01049) (0.763306, 0.0228199) (0.765357, 0.0353224) (0.767409, 0.0480009) (0.76946, 0.0608589) (0.771511, 0.0739003) (0.773562, 0.0871289) (0.775613, 0.100549) (0.777664, 0.114164) (0.779715, 0.12798) (0.781766, 0.142001) (0.783817, 0.156232) (0.785868, 0.170679) (0.787919, 0.185347) (0.78997, 0.200242) (0.792021, 0.21537) (0.794072, 0.230739) (0.796123, 0.246355) (0.798175, 0.262226) (0.800226, 0.278361) (0.802277, 0.294768) (0.804328, 0.311457) (0.806379, 0.328437) (0.80843, 0.345721) (0.810481, 0.363319) (0.812532, 0.381245) (0.814583, 0.399513) (0.816634, 0.418137) (0.818685, 0.437136) (0.820736, 0.456528) (0.822787, 0.476333) (0.824838, 0.496574) (0.82689, 0.517277) (0.828941, 0.538472) (0.830992, 0.56019) (0.833043, 0.582471) (0.835094, 0.605358) (0.837145, 0.628901) (0.839196, 0.653163) (0.841247, 0.678215) (0.843298, 0.704149) (0.845349, 0.731076) (0.8474, 0.759141) (0.849451, 0.788541) (0.851502, 0.81955) (0.853553, 0.852586)
	};
	\addlegendentry{$n=10^7,\; \delta_{\mathrm{est}}=10^{-3},\;\varepsilon_{\text{s}}=\varepsilon_{\mathrm{EA}}=10^{-5}$}
	
	\addplot[violet,thick,smooth] coordinates {
	(0.759204, -0.00635168) (0.761255, 0.00576902) (0.763306, 0.0180583) (0.765357, 0.0305193) (0.767409, 0.0431556) (0.76946, 0.0559705) (0.771511, 0.0689678) (0.773562, 0.0821515) (0.775613, 0.0955254) (0.777664, 0.109094) (0.779715, 0.122862) (0.781766, 0.136833) (0.783817, 0.151014) (0.785868, 0.165409) (0.787919, 0.180023) (0.78997, 0.194863) (0.792021, 0.209936) (0.794072, 0.225246) (0.796123, 0.240803) (0.798175, 0.256613) (0.800226, 0.272684) (0.802277, 0.289025) (0.804328, 0.305646) (0.806379, 0.322556) (0.80843, 0.339766) (0.810481, 0.357288) (0.812532, 0.375135) (0.814583, 0.39332) (0.816634, 0.411858) (0.818685, 0.430766) (0.820736, 0.450063) (0.822787, 0.469767) (0.824838, 0.489902) (0.82689, 0.510493) (0.828941, 0.531567) (0.830992, 0.553158) (0.833043, 0.575301) (0.835094, 0.598039) (0.837145, 0.621421) (0.839196, 0.645505) (0.841247, 0.670362) (0.843298, 0.696076) (0.845349, 0.722754) (0.8474, 0.750533) (0.849451, 0.779594) (0.851502, 0.810189) (0.853553, 0.842692)
	};
	\addlegendentry{$n=10^7,\; \delta_{\mathrm{est}}=10^{-3},\;\varepsilon_{\text{s}}=\varepsilon_{\mathrm{EA}}=10^{-6}$}

	\addplot[orange,thick,smooth] coordinates {
	(0.771511, -0.0015508) (0.773562, 0.0110039) (0.775613, 0.0237357) (0.777664, 0.0366482) (0.779715, 0.0497451) (0.781766, 0.0630305) (0.783817, 0.0765085) (0.785868, 0.0901835) (0.787919, 0.10406) (0.78997, 0.118143) (0.792021, 0.132438) (0.794072, 0.14695) (0.796123, 0.161685) (0.798175, 0.176649) (0.800226, 0.191849) (0.802277, 0.207291) (0.804328, 0.222982) (0.806379, 0.238932) (0.80843, 0.255147) (0.810481, 0.271637) (0.812532, 0.288412) (0.814583, 0.305482) (0.816634, 0.322858) (0.818685, 0.340553) (0.820736, 0.35858) (0.822787, 0.376952) (0.824838, 0.395687) (0.82689, 0.414801) (0.828941, 0.434314) (0.830992, 0.454247) (0.833043, 0.474624) (0.835094, 0.495471) (0.837145, 0.516818) (0.839196, 0.538701) (0.841247, 0.561158) (0.843298, 0.584235) (0.845349, 0.607986) (0.8474, 0.632476) (0.849451, 0.657782) (0.851502, 0.683999) (0.853553, 0.711249)
	};
	\addlegendentry{$n=10^6,\; \delta_{\mathrm{est}}=10^{-3},\;\varepsilon_{\text{s}}=\varepsilon_{\mathrm{EA}}=10^{-3}$}
	
	\addplot[green,thick,smooth] coordinates {
	(0.773562, -0.00798184) (0.775613, 0.00458303) (0.777664, 0.0173252) (0.779715, 0.0302482) (0.781766, 0.0433558) (0.783817, 0.0566522) (0.785868, 0.0701415) (0.787919, 0.083828) (0.78997, 0.0977165) (0.792021, 0.111812) (0.794072, 0.126119) (0.796123, 0.140644) (0.798175, 0.155392) (0.800226, 0.170369) (0.802277, 0.185583) (0.804328, 0.201039) (0.806379, 0.216745) (0.80843, 0.23271) (0.810481, 0.24894) (0.812532, 0.265447) (0.814583, 0.282238) (0.816634, 0.299326) (0.818685, 0.31672) (0.820736, 0.334434) (0.822787, 0.35248) (0.824838, 0.370873) (0.82689, 0.389629) (0.828941, 0.408766) (0.830992, 0.428302) (0.833043, 0.44826) (0.835094, 0.468663) (0.837145, 0.489538) (0.839196, 0.510916) (0.841247, 0.53283) (0.843298, 0.555322) (0.845349, 0.578436) (0.8474, 0.602229) (0.849451, 0.626763) (0.851502, 0.652119) (0.853553, 0.678392)
	};
	\addlegendentry{$n=10^6,\; \delta_{\mathrm{est}}=10^{-3},\;\varepsilon_{\text{s}}=\varepsilon_{\mathrm{EA}}=10^{-4}$}
	
	\addplot[yellow,thick,smooth] coordinates {
	(0.775613, -0.0122004) (0.777664, 0.00039393) (0.779715, 0.0131662) (0.781766, 0.0261199) (0.783817, 0.039259) (0.785868, 0.0525874) (0.787919, 0.0661095) (0.78997, 0.0798296) (0.792021, 0.0937525) (0.794072, 0.107883) (0.796123, 0.122226) (0.798175, 0.136788) (0.800226, 0.151574) (0.802277, 0.166591) (0.804328, 0.181845) (0.806379, 0.197342) (0.80843, 0.213091) (0.810481, 0.2291) (0.812532, 0.245376) (0.814583, 0.26193) (0.816634, 0.27877) (0.818685, 0.295908) (0.820736, 0.313355) (0.822787, 0.331124) (0.824838, 0.349227) (0.82689, 0.36768) (0.828941, 0.386499) (0.830992, 0.405701) (0.833043, 0.425307) (0.835094, 0.445338) (0.837145, 0.465818) (0.839196, 0.486775) (0.841247, 0.50824) (0.843298, 0.530249) (0.845349, 0.552841) (0.8474, 0.576066) (0.849451, 0.599978) (0.851502, 0.624644) (0.853553, 0.650146)
	};
	\addlegendentry{$n=10^6,\; \delta_{\mathrm{est}}=10^{-3},\;\varepsilon_{\text{s}}=\varepsilon_{\mathrm{EA}}=10^{-5}$}
	
	\addplot[gray,thick,smooth] coordinates {
	(0.822787, -0.0102655) (0.824838, 0.00363856) (0.82689, 0.0177499) (0.828941, 0.0320736) (0.830992, 0.0466152) (0.833043, 0.0613805) (0.835094, 0.0763757) (0.837145, 0.0916072) (0.839196, 0.107082) (0.841247, 0.122808) (0.843298, 0.138792) (0.845349, 0.155044) (0.8474, 0.171572) (0.849451, 0.188385) (0.851502, 0.205496) (0.853553, 0.222914)
	};
	\addlegendentry{$n=10^5,\; \delta_{\mathrm{est}}=10^{-2},\;\varepsilon_{\text{s}}=\varepsilon_{\mathrm{EA}}=10^{-3}$}
	
	\end{axis}  
\end{tikzpicture}

\caption{$\eta_{\mathrm{opt}}(\omega_{\mathrm{exp}})$ for $\gamma=1$ and several choices of $\delta_{\mathrm{est}},\;n,\;\varepsilon_{\mathrm{EA}},$ and the smoothing parameter $\varepsilon_{\text{s}}$. Note that for the errors of the protocols to be meaningful the number of rounds $n$ should be at least of order $\delta_{\mathrm{est}}^{-2}$.  $\varepsilon_{\mathrm{EA}}$ and $\varepsilon_{\text{s}}$ affect the soundness error in the protocols of the following sections. The dashed line shows the optimal asymptotic ($n\rightarrow\infty$) rate under the assumption that the devices are such that Alice, Bob, and Eve share an (unknown) i.i.d.\@ state.}
\label{fig:eta_rates}
\end{figure}

\begin{proof}[Proof of Theorem~\ref{thm:main_generation_chsh}]
Based on Lemma~\ref{lem:main_soundness}, it will suffice to define a min-tradeoff function $f_{\min}$ such that Equation~\eqref{eq:eat_f_min_bound} is satisfied. Using the chain rule,
	\[
		H\left( A_i B_i | X_i Y_i T_i F_i R' \right)_{\mathcal{N}_i(\sigma)} \geq H\left( A_i | X_i Y_i T_i F_i R' \right)_{\mathcal{N}_i(\sigma)}.
	\]

Furthermore,  
\begin{equation}\label{eq:entropy_input_split}
\begin{split}
	H\left( A_i | X_i Y_i T_i F_i R' \right)_{\mathcal{N}_i(\sigma)} &= \Pr\left[ X_i=0\right] \cdot H\left( A_i | Y_i T_i F_i R' , X_i=0 \right)_{\mathcal{N}_i(\sigma)} \\
	&+\Pr\left[ X_i=1\right] \cdot H\left( A_i | Y_i T_i F_i R' , X_i=1 \right)_{\mathcal{N}_i(\sigma)} \;.
\end{split}
\end{equation}
In the following we find a bound on $H\left( A_i | Y_i T_i F_i R' , X_i=0 \right)_{\mathcal{N}_i(\sigma)}$. Using exactly the same steps the same bound can be derived on $H\left( A_i | Y_i T_i F_i R' , X_i=1 \right)_{\mathcal{N}_i(\sigma)}$. 

Due to the bipartite requirement on the untrusted device $D$ used to implement the protocol and since Alice's actions (and her device's) are independent of Bob's choice of $Y_i$ and $T_i$ for the case $X_i=0$ we have\footnote{We assume that the value of $T_i$ is exchanged over a classical authenticated channel to which the device $D$ does not have access. In particular, Alice's part of the device is independent from the value of $T_i$.}
\[
	H\left( A_i | Y_i T_i F_i R' , X_i=0  \right)_{\mathcal{N}_i(\sigma)} = H\left(A_i | F_i R' , X_i=0  \right)_{\mathcal{N}_i(\sigma)} \;.
\]	
Using the definition of the conditional entropy one can rewrite $H\left( A_i | F_i R' , X_i=0  \right)_{\mathcal{N}_i(\sigma)}$ as follows:
	\begin{align}
		H\left( A_i | F_i R' , X_i=0\right) &= H\left( A_i F_i R' |  X_i=0 \right) - H\left( F_i R' | X_i=0 \right) \notag \\
		&= H\left( A_i | X_i=0\right) + H\left(F_i R' | A_i, X_i=0 \right) - H\left(F_i R' | X_i=0\right)\notag  \\
		&= H\left( A_i |  X_i=0 \right)  - \chi\left( A_i : F_iR' | X_i=0 \right) \notag\\
		&= 1  - \chi\left( A_i : F_iR' | X_i=0 \right) \;, \label{eq:h-chain-1}
	\end{align}
	where $\chi\left( A_i : F_iR' | X_i=0\right)=H\left( F_iR' |  X_i=0 \right)-H\left( F_iR' | A_i, X_i=0 \right)$ and the last equality follows from the symmetrisation step, Step~\ref{prostep:symmetry}.
	
	For states leading to a CHSH violation of $\beta\in[2,2\sqrt{2}]$ (for inputs restricted to $\{0,1\}\times\{0,1\}$) a tight bound on $\chi\left( A_i : F_iR' | X_i=0 \right)$ was derived in~\cite[Section 2.3]{pironio2009device}:
	\begin{equation}\label{eq:chi-bound}
		 \chi\left( A_i : F_i R' | X_i=0\right) \leq h\left( \frac{1}{2} + \frac{1}{2}\sqrt{\frac{\beta^2}{4} -1} \right) \;.
	\end{equation}
	Since $\omega = \frac{1}{8}\beta + \frac{1}{2}$, for $\omega\in\left[\frac{3}{4},\frac{2+\sqrt{2}}{4}\right]$ (i.e. a violation in the quantum regime) we get 
	\[
		 \chi\left( A_i : F_iR' | X_i=0 \right) \leq h\left( \frac{1}{2} + \frac{1}{2}\sqrt{16\omega \left(\omega-1\right) +3}  \right) \;.
	\]
	Combining this bound with Equations~\eqref{eq:entropy_input_split} and~\eqref{eq:h-chain-1} we conclude that for a state with winning probability~$\omega$,
	\begin{equation}\label{eq:entropy_bound_for_min_tradeoff}
		H\left( A_i B_i | X_i Y_i T_i F_i R' \right) \geq 1 - h\left( \frac{1}{2} + \frac{1}{2}\sqrt{16\omega \left(\omega-1\right) +3}  \right) \;.
	\end{equation}

	Consider a probability distribution $p=\mathrm{freq}_\mathbf{c}$ resulting from the observed data. If $p(0)+p(1)\neq \gamma$ then the set of states fulfilling $\mathcal{N}_i(\sigma)_{C_i}=p$ is empty and the condition on the min-tradeoff function given in Definition~\ref{def:min_tradeoff_func} becomes trivial. Hence, for the construction of the min-tradeoff function we can restrict our attention to $p$ with $p(0)+p(1)= \gamma$. For such $p$ we can write $\omega=\frac{p(1)}{p(0)+p(1)}=\frac{p(1)}{\gamma}$. All together we have for all $p$ in the considered regime,
	\begin{equation}\label{eq:one_box_entropy_final}
		\inf_{\sigma_{R_{i-1}R'}:\mathcal{N}_i(\sigma)_{C_i}=p} H\left( A_i B_i | X_i Y_i T_i F_i R' \right)_{\mathcal{N}_i(\sigma)} \geq 1 - h\left( \frac{1}{2} + \frac{1}{2}\sqrt{16\;\frac{p(1)}{\gamma} \left(\frac{p(1)}{\gamma}-1\right) +3}  \right) \;.
	\end{equation}
	
	Define a function $g$ over $p$ for which $p(2)/\gamma\in[3/54,1]$ by  
	\begin{equation}\label{eq:def-g}
	g(p) \,=\,  \begin{cases}
			1 - h\left( \frac{1}{2} + \frac{1}{2}\sqrt{16\; \frac{p(1)}{\gamma} \left(\frac{p(1)}{\gamma}-1\right) +3}  \right) &  \frac{p(1)}{\gamma}\in\left[\frac{3}{4},\frac{2+\sqrt{2}}{4}\right] \\
			1& \frac{p(1)}{\gamma}\in\left[\frac{2+\sqrt{2}}{4},1\right] \;.
			\end{cases}
	\end{equation}

	From Equation~\eqref{eq:one_box_entropy_final} it follows that any choice of $f_{\min}(p)$ that is differentiable and satisfies $f_{\min}(p) \leq g(p)$ for all $p$ will satisfy Equation~\eqref{eq:eat_f_min_bound}.

	For $\frac{p(1)}{\gamma}=\frac{2+\sqrt{2}}{4}$ the derivative of $g$ is infinite. For the final bound of the EAT to be meaningful $f_{\min}$ should be chosen such that $\| \triangledown f_{\min} \|_{\infty}$ is finite. To remedy this problem we choose $f_{\min}$ by ``cutting'' the function $g$ and ``gluing'' it to a linear function at some point $p_t$ (which is later optimised), while keeping the function differentiable. By doing this we ensure that the gradient of $f_{\min}$ is bounded, at the cost of losing a bit of entropy for $p$ with $p(1)> p_t(1)$. Towards this, denote 
	\begin{equation}\label{eq:derivative} 
		a(p_t)= \Big\lceil\frac{\mathrm{d}}{\mathrm{d}p(1)} g(p)\big|_{p_t} \Big\rceil \qquad\text{and} \qquad b(p_t) =  g(p_t)-a(p_t)\cdot p_t(1). 
	\end{equation}

	We then make the following choice\footnote{Note that $f_{min}$ is nonpositive for $ \frac{p(1)}{\gamma}\leq 3/4$, but this regime is not relevant as it would lead to the protocol aborting; the extension of $f_{min}$ to that range of values is only for mathematical convenience.} for the min-tradeoff function $f_{\min}$ (see Figure~\ref{fig:f_min}):
	\begin{equation}\label{eq:f_min_choice}
		f_{\min}\left(p,p_t\right) = \begin{cases}
			g\left(p\right) &  p(1) \leq p_t(1) \\
			a(p_t) \cdot p(1) + b(p_t) & p(1)> p_t(1)\;
		\end{cases}
	\end{equation}
	From the definition of $a$ and $b$ in Equation~\eqref{eq:derivative} this function is differentiable and fulfils the condition given in Equation~\eqref{eq:eat_f_min_bound}. Furthermore, by definition for any choice of $p_t$ it holds that $\| \triangledown f_{\min}(\cdot,p_t) \|_{\infty} \leq a(p_t)$. 
	
	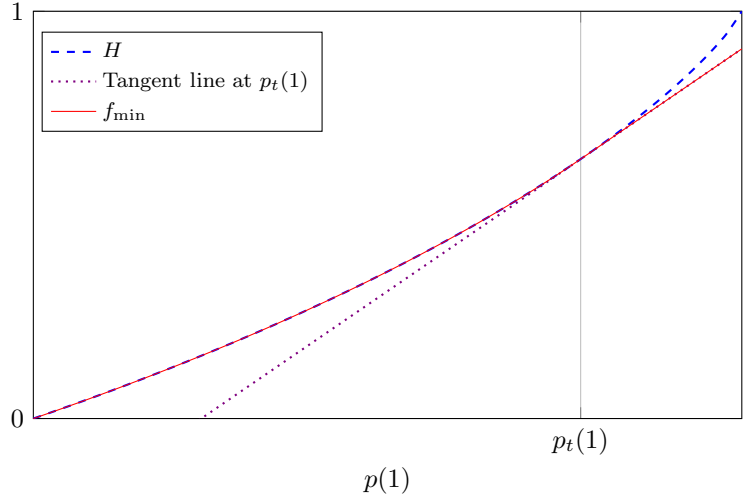
\begin{figure}
\centering
\begin{tikzpicture}
	\begin{axis}[
		height=7cm,
		width=11cm,
		grid=major,
		xlabel=$p(1)$,
		xmin=0.75,
		xmax=0.853553,
		ymax=1,
		ymin=0,
	     xtick={0.83},
	     xticklabels={$p_t(1)$},
          ytick={0,1},
		legend style={at={(0.21,0.95)},anchor=north,legend cell align=left,font=\footnotesize} 
	]
	

	\addplot[blue,thick,smooth,dashed] coordinates {
	(0.75, 0.) (0.752071, 0.0120347) (0.754142, 0.0242374) (0.756213, 0.0366113) (0.758284, 0.0491598) (0.760355, 0.0618861) (0.762426, 0.074794) (0.764497, 0.0878872) (0.766569, 0.10117) (0.76864, 0.114646) (0.770711, 0.128319) (0.772782, 0.142196) (0.774853, 0.156279) (0.776924, 0.170575) (0.778995, 0.185089) (0.781066, 0.199826) (0.783137, 0.214793) (0.785208, 0.229996) (0.787279, 0.245441) (0.78935, 0.261137) (0.791421, 0.277091) (0.793492, 0.293311) (0.795563, 0.309806) (0.797635, 0.326586) (0.799706, 0.343661) (0.801777, 0.361042) (0.803848, 0.378741) (0.805919, 0.396771) (0.80799, 0.415147) (0.810061, 0.433884) (0.812132, 0.452998) (0.814203, 0.47251) (0.816274, 0.49244) (0.818345, 0.51281) (0.820416, 0.533648) (0.822487, 0.554982) (0.824558, 0.576846) (0.82663, 0.599279) (0.828701, 0.622324) (0.830772, 0.646033) (0.832843, 0.670469) (0.834914, 0.695705) (0.836985, 0.721832) (0.839056, 0.748965) (0.841127, 0.777251) (0.843198, 0.806888) (0.845269, 0.838156) (0.84734, 0.871481) (0.849411, 0.907587) (0.851482, 0.948007) (0.853553, 1.)	
	};
	\addlegendentry{$H$}
	
	\addplot[violet,thick,smooth,dotted] coordinates {
	(0.75, -0.282055) (0.752071, -0.258259) (0.754142, -0.234463) (0.756213, -0.210667) (0.758284, -0.186871) (0.760355, -0.163075) (0.762426, -0.139279) (0.764497, -0.115484) (0.766569, -0.0916877) (0.76864, -0.0678918) (0.770711, -0.0440959) (0.772782, -0.0203) (0.774853, 0.00349582) (0.776924, 0.0272917) (0.778995, 0.0510876) (0.781066, 0.0748834) (0.783137, 0.0986793) (0.785208, 0.122475) (0.787279, 0.146271) (0.78935, 0.170067) (0.791421, 0.193863) (0.793492, 0.217659) (0.795563, 0.241455) (0.797635, 0.26525) (0.799706, 0.289046) (0.801777, 0.312842) (0.803848, 0.336638) (0.805919, 0.360434) (0.80799, 0.38423) (0.810061, 0.408026) (0.812132, 0.431821) (0.814203, 0.455617) (0.816274, 0.479413) (0.818345, 0.503209) (0.820416, 0.527005) (0.822487, 0.550801) (0.824558, 0.574597) (0.82663, 0.598393) (0.828701, 0.622188) (0.830772, 0.645984) (0.832843, 0.66978) (0.834914, 0.693576) (0.836985, 0.717372) (0.839056, 0.741168) (0.841127, 0.764964) (0.843198, 0.788759) (0.845269, 0.812555) (0.84734, 0.836351) (0.849411, 0.860147) (0.851482, 0.883943) (0.853553, 0.907739)
	};
	\addlegendentry{Tangent line at $p_t(1)$}
	
	\addplot[red,smooth] coordinates {
	(0.75, 0.) (0.752071, 0.0120347) (0.754142, 0.0242374) (0.756213, 0.0366113) (0.758284, 0.0491598) (0.760355, 0.0618861) (0.762426, 0.074794) (0.764497, 0.0878872) (0.766569, 0.10117) (0.76864, 0.114646) (0.770711, 0.128319) (0.772782, 0.142196) (0.774853, 0.156279) (0.776924, 0.170575) (0.778995, 0.185089) (0.781066, 0.199826) (0.783137, 0.214793) (0.785208, 0.229996) (0.787279, 0.245441) (0.78935, 0.261137) (0.791421, 0.277091) (0.793492, 0.293311) (0.795563, 0.309806) (0.797635, 0.326586) (0.799706, 0.343661) (0.801777, 0.361042) (0.803848, 0.378741) (0.805919, 0.396771) (0.80799, 0.415147) (0.810061, 0.433884) (0.812132, 0.452998) (0.814203, 0.47251) (0.816274, 0.49244) (0.818345, 0.51281) (0.820416, 0.533648) (0.822487, 0.554982) (0.824558, 0.576846) (0.82663, 0.599279) (0.828701, 0.622324) (0.830772, 0.645984) (0.832843, 0.66978) (0.834914, 0.693576) (0.836985, 0.717372) (0.839056, 0.741168) (0.841127, 0.764964) (0.843198, 0.788759) (0.845269, 0.812555) (0.84734, 0.836351) (0.849411, 0.860147) (0.851482, 0.883943) (0.853553, 0.907739)
	};
	\addlegendentry{$f_{\min}$}
			
	\end{axis}  
\end{tikzpicture}

\caption{The construction of the min-tradeoff function $f_{\min}$ as in Equation~\eqref{eq:f_min_choice}. The plot shows the values of the min-tradeoff function on a slice $p(0)+p(1)= \text{constant}$. }
		\label{fig:f_min}
\end{figure}

	Applying Lemma~\ref{lem:main_soundness}, we conclude that  for any $\frac{3}{4}<\frac{p_t(1)}{\gamma}<\frac{2+\sqrt{2}}{4} $, either the protocol aborts with probability greater than $1-\varepsilon_{\mathrm{EA}}$, or
	\begin{equation}\label{eq:EAT_statement}
		H^{\varepsilon_{\text{s}}}_{\min} \left( \mathbf{A B} | \mathbf{X Y T F} E \right)_{\rho_{|\Omega}} > n f_{\min}\left( \omega_{\mathrm{exp}}\gamma - \delta_{\mathrm{est}},p_t \right) - \sqrt{n} \zeta(p_t)  \;,
	\end{equation}
	for $\zeta(p_t,\varepsilon_{\text{s}},\varepsilon_{\mathrm{AE}}) = 2\left( \log 13 + a(p_t) \right)\sqrt{1-2 \log (\varepsilon_{\text{s}} \cdot\varepsilon_{\mathrm{AE}})}$ (as $d_{A_iB_i} = 6$). To obtain the optimal rate we optimise over $p_t$ . Denote $\eta(p,p_t,\varepsilon_{\text{s}},\varepsilon_{\mathrm{AE}}) =  f_{\min}\left(p, p_t\right) - \frac{1}{\sqrt{n}}\zeta(p_t,\varepsilon_{\text{s}},\varepsilon_{\mathrm{AE}})$ and let
	\[
		\eta_{\mathrm{opt}} (\varepsilon_{\text{s}},\varepsilon_{\mathrm{AE}})= \max_{\frac{3}{4} < \frac{p_t(1)}{\gamma} < \frac{2+\sqrt{2}}{4}} \eta(\omega_{\mathrm{exp}}\gamma - \delta_{\mathrm{est}},p_t,\varepsilon_{\text{s}},\varepsilon_{\mathrm{AE}}) \;.
	\]	
	Plugging this into Equation~\eqref{eq:EAT_statement} the theorem follows.
\end{proof}

We end this section by showing how the particular implementation of the symmetrisation step, Step~\ref{prostep:symmetry}, of Protocol~\ref{pro:randomness_generation} made here for the CHSH game can be ignored in any implementation of the protocol. For this,  rewrite Equation~\eqref{eq:main_thm} more formally as\footnote{Previously for ease of notation we wrote $\mathbf{AB}$ for the flipped outputs; here we denote the same bits as $ g_{\mathbf{F}}\left(\mathbf{AB}\right)$ to make the flipping operation explicit.} 
\begin{equation}\label{eq:smooth_with_symmetry}
	 H^{\varepsilon_{\text{s}}}_{\min} \left( g_{\mathbf{F}}\left(\mathbf{AB}\right) | \mathbf{X Y T F} E \right)_{\rho_{|\Omega}} > n\cdot \eta_{\mathrm{opt}} \;,
\end{equation}
where $g_{\mathbf{F}}$ is the function that flips the bits according to $\mathbf{F}$. Since for any fixed value of $F$, $g_F$ is a deterministic function it follows from~\cite[Lemma 1]{scarani2008security} that for any $\varepsilon_{\text{s}}\geq 0$,
\begin{equation}\label{eq:remove_symmetry}
	H^{\varepsilon_{\text{s}}}_{\min} \left( \mathbf{AB}  | \mathbf{X Y T} E \right)_{\rho_{|\Omega}} \geq H^{\varepsilon_{\text{s}}}_{\min} \left( g_{\mathbf{F}}\left(\mathbf{AB}\right) | \mathbf{X Y T F} E \right)_{\rho_{|\Omega}} \;.
\end{equation}
Combining Equations~\eqref{eq:smooth_with_symmetry} and~\eqref{eq:remove_symmetry} proves the following corollary. 

\begin{cor}\label{cor:main_generation}
Under the same assumptions as Theorem~\ref{thm:main_generation_chsh}, but for an implementation of Protocol~\ref{pro:randomness_generation} in which the symmetrisation step, Step~\ref{prostep:symmetry}, is omitted, for any $\varepsilon_{\mathrm{EA}},\varepsilon_{\text{s}}\in (0,1)$, either the protocol aborts with probability greater than $1-\varepsilon_{\mathrm{EA}}$ or
	\begin{equation}\label{eq:main_thmc}
		 H^{\varepsilon_{\text{s}}}_{\min} \left( \mathbf{A B} | \mathbf{X Y T} E \right)_{\rho_{|\Omega}} > n\cdot \eta_{\mathrm{opt}}(\varepsilon_{\text{s}},\varepsilon_{\mathrm{EA}}) \;, 
	\end{equation}
	where  $\eta_{\mathrm{opt}}$ is defined in Equation~\eqref{eq:eta_opt}. 
\end{cor}

In Appendix~\ref{sec:better_rate} we use a small modification of the entropy accumulation protocol and the above proof to get a similar bound on the entropy rate which has a better dependency on the probability of a test $\gamma$. This is of relevance for some applications such as DIQKD. 
The calculations presented in Appendix~\ref{sec:better_rate} are slightly more technical than the proof given above, but do not require any substantially different observations.

\section{DIQKD}\label{sec:diqkd}

\subsection{The protocol}
\label{sec:diqkd-protocol}

Our protocol for DIQKD is described as Protocol~\ref{pro:diqkd_chsh} below. An honest implementation is described in Section~\ref{sec:honest_qkd_imp}. 

In the first part of the protocol Alice and Bob use their devices to produce the raw data, similarly to what is done in the entropy accumulation protocol, Protocol~\ref{pro:randomness_generation} (with the game $G$ equal to the CHSH game, as in Section~\ref{sec:entropy-chsh}). The main difference is that Bob's outputs always contains Bob's $i$-th measurement outcome (instead of being set to $\perp$ in all rounds for which $T_i=0$); to make the distinction explicit we denote Bob's outputs in Protocol~\ref{pro:diqkd_chsh} with a tilde, $\tilde{\mathbf{B}}$. 

In the second part of the protocol Alice and Bob apply classical post-processing steps to produce their final keys. We choose classical post-processing steps that optimise the key rate, but which may not be optimal in other aspects, e.g., computation time. The protocol and the analysis can easily be adapted for other choices of classical post-processing. 

We now describe the three post-processing steps, error correction, parameter estimation and privacy amplification, in detail.\footnote{We remark that in Step~\ref{prostep:choosing_est_test} of Protocol~\ref{pro:diqkd_chsh} Alice and Bob choose $T_i$ together (or exchange its value between them) in every round of the protocol and choose their inputs accordingly. This is in contrast to choosing Alice and Bob's input from a product distribution and then adding a sifting step, as usually done in QKD protocols. 
It follows from our proof technique that making $T_i$ public as we do does not compromise the security of the protocol.}

\paragraph{Error correction.} 
Alice and Bob use an error correction protocol $\mathrm{EC}$ to obtain identical raw keys $K_A$ and $K_B$ from their bits $\mathbf{A},\tilde{\mathbf{B}}$. In our analysis we use a protocol, based on universal hashing, which minimises the amount of leakage to the adversary~\cite{brassard1993secret,renner2005simple}  (see also Section 3.3.2 in~\cite{beaudry2015assumptions} for details). To implement this protocol Alice chooses a hash function and sends the chosen function and the hashed value of her bits to Bob. We denote this classical communication by $O$. Bob uses $O$, together with his prior knowledge $\tilde{\mathbf{B}}\mathbf{XYT}$, to compute a guess $\hat{\mathbf{A}}$ for Alice's bits $\mathbf{A}$.
If $\mathrm{EC}$ raises a ``fail'' flag Alice and Bob abort; in an honest implementation this happens with probability at most $\varepsilon_{\mathrm{EC}}^c$. 
The probability that Alice and Bob do not abort but hold different raw keys $K_A=\mathbf{A}$ and $K_B=\hat{\mathbf{A}}\neq K_A$ is at most $\varepsilon_{\mathrm{EC}}$.

Due to the communication from Alice to Bob $\mathrm{leak_{EC}}$ bits of information are leaked to the adversary. The following guarantee follows for the described protocol~\cite{renner2005simple}:
\begin{equation}\label{eq:ec_leakage}
	\mathrm{leak_{EC}} \leq H_{0}^{\varepsilon'_{\mathrm{EC}}}\left(\mathbf{A}|\tilde{\mathbf{B}}\mathbf{XYT}\right) + \log\left(\frac{1}{\varepsilon_{\mathrm{EC}}}\right)  \;,
\end{equation}
for $\varepsilon_{\mathrm{EC}}^c = \varepsilon'_{\mathrm{EC}} + \varepsilon_{\mathrm{EC}}$ and where $H_{0}^{\varepsilon'_{\mathrm{EC}}}(\mathbf{A}|\tilde{\mathbf{B}}\mathbf{XYT})$ is evaluated on the state in an \emph{honest} implementation of the protocol. For example, for quantum channels with an i.i.d.\@ noise model $H_{0}^{\varepsilon'_{\mathrm{EC}}}\left(\mathbf{A}|\tilde{\mathbf{B}}\mathbf{XYT}\right)$ can be bounded by above using the asymptotic equipartition property~\cite{tomamichel2009fully} (see Equation~\eqref{eq:explicit_leak_ec} below for the explicit bound in that case).  
If a larger fraction of errors occur when running the actual DIQKD protocol (for instance due to adversarial interference) the error correction might not succeed, as Bob will not have a sufficient amount of information to obtain a good guess of Alice's bits. If so this will be detected with probability at least $ 1 - \varepsilon_{\mathrm{EC}}$ and the protocol will abort.

\paragraph{Parameter estimation.}
After the error correction step, Bob has all of the relevant information to perform parameter estimation from his data alone, without any further communication with Alice. Using $\tilde{\mathbf{B}}$ and $K_B$, Bob sets $C_i = w_{\text{CHSH}}\left(\hat{A}_i,\tilde{B}_i,X_i,Y_i\right)=w_{\text{CHSH}}\left({K_B}_i,\tilde{B}_i,X_i,Y_i\right)$ for the test rounds and $C_i = \perp$ otherwise. He aborts if the fraction of successful game rounds is too low, that is, if $\sum_j C_j < \left(\omega_{\mathrm{exp}}\gamma - \delta_{\mathrm{est}}\right) \cdot n\;$.

As Bob does the estimation using his guess of Alice's bits, the probability of aborting in this step in an honest implementation, $\varepsilon_{\mathrm{PE}}^c$, is bounded by
\begin{align}
	\varepsilon_{\mathrm{PE}}^c&\leq \Pr\Big( \sum_j C_j < \left(\omega_{\mathrm{exp}} \gamma - \delta_{\mathrm{est}}\right) \cdot \sum_j T_j \Big| K_A=K_B \Big) + \Pr \big( K_A \neq K_B \text{ and $\mathrm{EC}$ does not abort}\big) \nonumber \\
	&\leq \varepsilon_{\mathrm{EA}}^c + \varepsilon_{\mathrm{EC}} \;, \label{eq:pe_completeness}
\end{align}
since conditioned on the error correction protocol succeeding the probability of aborting in the honest case is exactly as in the entropy accumulation protocol (Protocol~\ref{pro:randomness_generation}).

\paragraph{Privacy amplification.} 
Finally, Alice and Bob use a (quantum-proof) privacy amplification protocol $\mathrm{PA}$ (which takes some random seed $S$ as input) to create their final keys $\tilde{K}_A$ and $\tilde{K}_B$ of length $\ell$, which are close to ideal keys, i.e., uniformly random and independent of the adversary's knowledge. 

For simplicity we use universal hashing~\cite{renner2005universally} as the privacy amplification protocol in the analysis below. Any other quantum-proof strong extractor, e.g., Trevisan's extractor~\cite{de2012trevisan}, can be used for this task and the analysis can be easily adapted. 

The secrecy of the final key depends only on the privacy amplification protocol used and the value of $H^{\varepsilon_{\text{s}}}_{\min} ( \mathbf{A} | \mathbf{X Y T} O E )$, evaluated on the state at the end of the protocol, conditioned on not aborting. For universal hashing, for every $\varepsilon_{\mathrm{PA}},\varepsilon_{\text{s}}\in(0,1)$ a secure key of maximal length 
\begin{equation}\label{eq:universal_hashing_length}
	\ell = H^{\varepsilon_{\text{s}}}_{\min} ( \mathbf{A} | \mathbf{X Y T} O E ) -2\log\frac{1}{\varepsilon_{\mathrm{PA}}} 
\end{equation}
is produced with probability\footnote{$\varepsilon_{\mathrm{PA}}$ is the error probability of the extractor when it is applied on a normalised state satisfying the relevant min-entropy bound. For universal hashing, when only a bound on the smooth min-entropy is given the smoothing parameter $\varepsilon_{\text{s}}$ should be added to the error $\varepsilon_{\mathrm{PA}}$. When working with other extractors one should adapt the parameters accordingly; see~\cite[Section 4.3]{arnon2015quantum}.} at least $1-\varepsilon_{\mathrm{PA}} - \varepsilon_{\text{s}}$.

\begin{algorithm}[t]
\caption{CHSH-based DIQKD protocol}
\label{pro:diqkd_chsh}
\begin{algorithmic}[1]
	\STATEx \textbf{Arguments:} 
		\STATEx\hspace{\algorithmicindent} $D$ -- untrusted device of two components that can play CHSH repeatedly
		\STATEx\hspace{\algorithmicindent} $n \in \mathbb{N}_+$ -- number of rounds
		\STATEx\hspace{\algorithmicindent} $\gamma \in (0,1]$ -- expected fraction of test rounds 

		\STATEx\hspace{\algorithmicindent} $\omega_{\mathrm{exp}}$ -- expected winning probability in an honest (perhaps noisy) implementation    
		\STATEx\hspace{\algorithmicindent} $\delta_{\mathrm{est}} \in (0,1)$ -- width of the statistical confidence interval for the estimation test
		
		\STATEx\hspace{\algorithmicindent} $\mathrm{EC}$ -- error correction protocol which leaks $\mathrm{leak_{EC}}$ bits and has error probability $\varepsilon_{\mathrm{EC}}$
		 \STATEx\hspace{\algorithmicindent} $\mathrm{PA}$ -- privacy amplification protocol with error probability $\varepsilon_{\mathrm{PA}}$
		
	\STATEx
	
	\STATE For every round $i\in[n]$ do Steps~\ref{prostep:choose_Ti}-\ref{prostep:use_device_qkd}:
		\STATE\hspace{\algorithmicindent} Alice and Bob choose a random $T_i\in\{0,1\}$ such that $\Pr(T_i=1)=\gamma$.  \label{prostep:choosing_est_test}
		\STATE\hspace{\algorithmicindent} If $T_i=0$ Alice and Bob choose $(X_i,Y_i)=(0,2)$ and otherwise $X_i,Y_i\in \{0,1\}$ uniformly at random.  
		\STATE\hspace{\algorithmicindent} Alice and Bob use $D$ with $X_i,Y_i$ and record their outputs as $A_i$ and $\tilde{B}_i$ respectively. \label{prostep:use_device_qkd}
	
	\STATEx
	
	\STATE \textbf{Error correction:} Alice and Bob apply the error correction protocol $\mathrm{EC}$, communicating $O$ in the process. If $\mathrm{EC}$ aborts they abort the protocol. Otherwise, they obtain raw keys denoted by $K_A$ and $K_B$. \label{prostep:ec}
	\STATE \textbf{Parameter estimation:} Using $\tilde{\mathbf{B}}$ and $K_B$, Bob sets $C_i = w_{\text{CHSH}}\left({K_B}_i,\tilde{B}_i,X_i,Y_i\right)$ for the test rounds and $C_i = \perp$ otherwise. He aborts if $\sum_j C_j < \left(\omega_{\mathrm{exp}}\gamma - \delta_{\mathrm{est}}\right) \cdot n;$. \label{prostep:abort_chsh_qkd}
	\STATE \textbf{Privacy amplification:} Alice and Bob apply the privacy amplification protocol $\mathrm{PA}$ on $K_A$ and $K_B$ to create their final keys $\tilde{K}_A$ and $\tilde{K}_B$ of length $\ell$ as defined in Equation~\eqref{eq:key_length_def}. \label{prostep:pa}	
\end{algorithmic}
\end{algorithm}

The main theorem of this section is the following security result for Protocol~\ref{pro:diqkd_chsh}:

\begin{thm}\label{thm:QKD_security}
	The DIQKD protocol given in Protocol~\ref{pro:diqkd_chsh} is $(\varepsilon_{\mathrm{QKD}}^s,\varepsilon_{\mathrm{QKD}}^c,\ell)$-secure according to Definition~\ref{def:security_QKD}, with $\varepsilon^s_{\mathrm{QKD}} \leq  \varepsilon_{\mathrm{EC}} + \varepsilon_{\mathrm{PA}} + \varepsilon_{\text{s}} + \varepsilon_{\mathrm{EA}}$, $\varepsilon^c_{\mathrm{QKD}} \leq \varepsilon^c_{EC} + \varepsilon_{\mathrm{EA}}^c + \varepsilon_{\mathrm{EC}}$, and
	\begin{equation}\label{eq:key_length_def}
	\begin{split}
		\ell = \; n \cdot \eta_{\mathrm{opt}}\left(\varepsilon_{\text{s}}/4,\varepsilon_{\mathrm{EA}} + \varepsilon_{\mathrm{EC}}\right) - \mathrm{leak_{EC}} - 3 \log\left(1-\sqrt{1-(\varepsilon_{\text{s}}/4)^2}\right)  \\
		- \gamma n - \sqrt{n} 2\log7\sqrt{1-2\log \left( \varepsilon_{\text{s}}/4 \cdot \left(\varepsilon_{\mathrm{EA}} + \varepsilon_{\mathrm{EC}}\right) \right)}   -2\log\left(\varepsilon^{-1}_{\mathrm{PA}}\right) \;,
	\end{split}
	\end{equation}
	where $\eta_{\mathrm{opt}}$ is specified in Equation~\eqref{eq:eta_opt}. 
\end{thm}

In Section~\ref{sec:rates_comparison_qkd} we plot the resulting key rates, $\ell/n$, for different choices of parameters.

Theorem~\ref{thm:QKD_security} follows by combining Lemmas~\ref{lem:QKD_complete} and~\ref{lem:QKD_sound} that we prove in the following sections. 

\subsection{The honest implementation}\label{sec:honest_qkd_imp}

The honest (but possibly noisy) implementation of the protocol is one where the device $D$ performs in \emph{every round~$i$} of the protocol the measurements $\mathcal{M}_{x_i}^{a_i} \otimes \mathcal{M}_{y_i}^{b_i}$ on Alice and Bob's state $\rho_{Q_AQ_B}$. The state and measurements are such that the winning probability achieved in the CHSH game in a single round is $\omega_{\mathrm{exp}}$.\footnote{Note that in our notation, the noise that affects the winning probability in the CHSH game is already included in $\omega_{\mathrm{exp}}$.} For the measurements $(X_i,Y_i)=(0,2)$ we denote the quantum bit error rate, i.e., the probability that $A_i\neq B_i$ while using these measurements, by $Q$. 
Thus, in the honest case we assume the device $D$ behaves in an i.i.d.\@ manner (and in particular an i.i.d.\@ noise model for the quantum channels used in the protocol): it is initialised in an i.i.d.\@  bipartite state, $\rho_{Q_AQ_B}^{\otimes n}$, on which it makes i.i.d.\@ measurements.

As an example, one possible realisation of such an implementation is the following. Alice and Bob share the two-qubit Werner state $\rho_{Q_AQ_B} = (1-\nu) \ket{\phi^+}\bra{\phi^+} + \nu\mathbb{I}/4$ for $\ket{\Phi^+} = 1/\sqrt{2}\left(\ket{00}+\ket{11}\right)$ and $ \nu \in [0,1]$. The state $\rho_{Q_AQ_B}$ arises, e.g., from the state $\ket{\Phi^+}$ after going through a depolarisation channel. 
For every $i\in[n]$, Alice's measurements $X_i=0$ and $X_i=1$ correspond to $\sigma_z$ and $\sigma_x$ respectively and Bob's measurements $Y_i=0$, $Y_i=1$, and $Y_i=2$ to $\frac{\sigma_z + \sigma_x}{\sqrt{2}}$, $\frac{\sigma_z - \sigma_x}{\sqrt{2}}$ and $\sigma_z$ respectively. The winning probability in the CHSH game (restricted to $X_i,Y_i\in\{0,1\}$) using these measurements on $\rho_{Q_AQ_B}$ is $\omega_{\mathrm{exp}}=\frac{2+\sqrt{2}(1-\nu)}{4}$ and $Q=\frac{\nu}{2}$. 

\subsection{Completeness}

The completeness of the protocol follows from the honest i.i.d.\@ implementation described in Section~\ref{sec:honest_qkd_imp} and the completeness of the entropy accumulation protocol as shown in Section~\ref{sec:EA_completeness}.

\begin{lemma}\label{lem:QKD_complete}
	Protocol~\ref{pro:diqkd_chsh} is complete with completeness error $\varepsilon^c_{\mathrm{QKD}} \leq \varepsilon^c_{EC} + \varepsilon_{\mathrm{EA}}^c + \varepsilon_{\mathrm{EC}} $. That is, the probability that the protocol aborts for an honest implementation of the device $D$ is at most $\varepsilon^c_{\mathrm{QKD}}$.
\end{lemma}

\begin{proof}
There are two steps at which	Protocol~\ref{pro:diqkd_chsh} may abort: after the error correction (Step~\ref{prostep:ec}) or in the Bell violation estimation (Step~\ref{prostep:abort_chsh_qkd}). By the union bound, the total probability of aborting is at most the probability of aborting in each of these steps. Using this and Equation~\eqref{eq:pe_completeness} we get:
	\[
		\varepsilon^c_{\mathrm{QKD}} \leq \varepsilon^c_{EC} + \varepsilon^c_{PE}  \leq \varepsilon^c_{EC} + \varepsilon_{\mathrm{EA}}^c + \varepsilon_{\mathrm{EC}} \;. \qedhere
	\]
\end{proof}

\subsection{Soundness}

To establish soundness, first note that by definition, as long as the protocol does not abort it produces a key of length $\ell$. Therefore it remains to verify correctness, which depends on the error correction step, and security, which is based on the privacy amplification step. To prove security we start with Lemma~\ref{lem:smooth_bound_qkd}, in which we assume that the error correction step is successful. We then use it to prove soundness in Lemma~\ref{lem:QKD_sound}.

Let $\tilde{\Omega}$ denote the event of Protocol~\ref{pro:diqkd_chsh} not aborting \emph{and} the $\mathrm{EC}$ protocol being successful, and let $\tilde{\rho}_{\mathbf{A}\tilde{\mathbf{B}}\mathbf{XYT}OE|\tilde{\Omega}}$ be the state at the end of the protocol, conditioned on this event.

Success of the privacy amplification step relies on the min-entropy $H^{\varepsilon_{\text{s}}}_{\min} ( \mathbf{A} | \mathbf{X Y T} O E )_{\tilde{\rho}_{|\tilde{\Omega}}}$ being sufficiently large.  The following lemma connects this quantity to 
 $H^{\frac{\varepsilon_{\text{s}}}{4}}_{\min} ( \mathbf{AB} | \mathbf{X Y T}  E )_{\rho_{|\Omega}}$, on which a lower bound is provided by Corollary~\ref{cor:main_generation}. 

\begin{lemma}\label{lem:smooth_bound_qkd}
	For any device $D$, let $\tilde{\rho}$ be the state generated in Protocol~\ref{pro:diqkd_chsh} right before the privacy amplification step, Step~\ref{prostep:pa}. Let $\tilde{\rho}_{|\tilde{\Omega}}$ be the state conditioned on not aborting the protocol and success of the $\mathrm{EC}$ protocol. Then, for any $\varepsilon_{\mathrm{EA}},\varepsilon_{\mathrm{EC}},\varepsilon_{\text{s}}\in (0,1)$, either the protocol aborts with probability greater than $1-\varepsilon_{\mathrm{EA}} - \varepsilon_{\mathrm{EC}}$ or
	\begin{equation}\label{eq:final_min_entropy_bound_qkd}
	\begin{split}
		H^{\varepsilon_{\text{s}}}_{\min} \left( \mathbf{A} | \mathbf{X Y T} O E \right)_{\tilde{\rho}_{|\tilde{\Omega}}} \geq n \cdot \eta_{\mathrm{opt}}\left(\varepsilon_{\text{s}}/4,\varepsilon_{\mathrm{EA}} + \varepsilon_{\mathrm{EC}}\right) - \mathrm{leak_{EC}}  -  3 \log\left(1-\sqrt{1-(\varepsilon_{\text{s}}/4)^2}\right)  \\
		- \gamma n - \sqrt{n} 2\log7\sqrt{1-2\log \left( \varepsilon_{\text{s}}/4 \cdot \left(\varepsilon_{\mathrm{EA}} + \varepsilon_{\mathrm{EC}}\right) \right)} \;.
	\end{split}
	\end{equation}
\end{lemma}

\begin{proof}

	Consider the following events:
	\begin{enumerate}
		\item $\Omega$: the event of not aborting in the entropy accumulation protocol, Protocol~\ref{pro:randomness_generation}. This happens when the Bell violation, calculated using Alice and Bob's outputs (and inputs), is sufficiently high. 
		\item $\hat{\Omega}$: Suppose Alice and Bob run  Protocol~\ref{pro:randomness_generation}, and then execute the $\mathrm{EC}$ protocol. The event $\hat{\Omega}$ is defined by $\Omega$ \emph{and} $K_B = \mathbf{A}$.	
		\item $\tilde{\Omega}$: the event of not aborting the DIQKD protocol, Protocol~\ref{pro:diqkd_chsh}, \emph{and} $K_B = \mathbf{A}$. 
	\end{enumerate}
	The state $\rho_{|\hat{\Omega}}$ then denotes the state at the end of Protocol~\ref{pro:randomness_generation} conditioned on $\hat{\Omega}$. 

	As we are only interested in the case where the $\mathrm{EC}$ protocol outputs the correct guess of Alice's bits, that is $K_B = \mathbf{A}$ (which happens with probability $1-\varepsilon_{\mathrm{EC}}$), we have 
	$\tilde{\rho}_{\mathbf{A X Y T}  E|\tilde{\Omega}} = \rho_{\mathbf{A X Y T}  E|\hat{\Omega}}$ (note $\tilde{\mathbf{B}}$ and $\mathbf{B}$ were traced out from $\tilde{\rho}$ and $\rho$ respectively). Hence, 
	\begin{equation}\label{eq:marginals_equiv}
		H^{\varepsilon_{\text{s}}}_{\min} \left( \mathbf{A} | \mathbf{X Y T}  E \right)_{\tilde{\rho}_{|\tilde{\Omega}}} = H^{\varepsilon_{\text{s}}}_{\min} \left( \mathbf{A} | \mathbf{X Y T}  E \right)_{\rho_{|\hat{\Omega}}} \;.
	\end{equation}

	Using the chain rule given in~\cite[Lemma 6.8]{tomamichel2015quantum} together with Equation~\eqref{eq:marginals_equiv} we get  that 
	\begin{align}
		H^{\varepsilon_{\text{s}}}_{\min} \left( \mathbf{A} | \mathbf{X Y T} O E \right)_{\tilde{\rho}_{|\tilde{\Omega}}} &\geq H^{\varepsilon_{\text{s}}}_{\min} \left( \mathbf{A} | \mathbf{X Y T}  E \right)_{\tilde{\rho}_{|\tilde{\Omega}}} - \mathrm{leak_{EC}} \nonumber \\
		&= H^{\varepsilon_{\text{s}}}_{\min} \left( \mathbf{A} | \mathbf{X Y T}  E \right)_{\rho_{|\hat{\Omega}}} - \mathrm{leak_{EC}} \;. \label{eq:entropy_chain_1}
	\end{align}
	
	To apply Corollary~\ref{cor:main_generation} it remains to relate $H^{\varepsilon_{\text{s}}}_{\min} \left( \mathbf{A} | \mathbf{X Y T}  E \right)_{\rho_{|\hat{\Omega}}}$ to $H^{\varepsilon'_{\text{s}}}_{\min} \left( \mathbf{AB} | \mathbf{X Y T}  E \right)_{\rho_{|\hat{\Omega}}}$ for some~$\varepsilon'_{\text{s}}$. For this we first write
	\begin{align*}
		H^{\varepsilon_{\text{s}}}_{\min} \left( \mathbf{A} | \mathbf{X Y T}  E \right)_{\rho_{|\hat{\Omega}}} &\geq H^{\frac{\varepsilon_{\text{s}}}{4}}_{\min} \left( \mathbf{AB} | \mathbf{X Y T}  E \right)_{\rho_{|\hat{\Omega}}} - H^{\frac{\varepsilon_{\text{s}}}{4}}_{\max} \left( \mathbf{B} | \mathbf{ A X Y T}  E \right)_{\rho_{|\hat{\Omega}}}  - 3 \log\left(1-\sqrt{1-(\varepsilon_{\text{s}}/4)^2}\right) \\
		&\geq H^{\frac{\varepsilon_{\text{s}}}{4}}_{\min} \left( \mathbf{AB} | \mathbf{X Y T}  E \right)_{\rho_{|\hat{\Omega}}} - H^{\frac{\varepsilon_{\text{s}}}{4}}_{\max} \left( \mathbf{B} | \mathbf{T} E  \right)_{\rho_{|\hat{\Omega}}}  - 3 \log\left(1-\sqrt{1-(\varepsilon_{\text{s}}/4)^2}\right) \;,
	\end{align*}
	where the first inequality is due to the chain rule~\cite[Equation (6.57)]{tomamichel2015quantum} and the second is due to strong sub-additivity of the smooth max-entropy.

	One can now apply the EAT to upper bound $H^{\frac{\varepsilon_{\text{s}}}{4}}_{\max} \left( \mathbf{B} | \mathbf{ T}  E \right)_{\rho_{|\hat{\Omega}}}$ in the following way. We use Theorem~\ref{thm:eat} with the replacements $\mathbf{AB} \rightarrow \mathbf{B}, \; \mathbf{I}\rightarrow \mathbf{T}, \; E \rightarrow E$. The Markov conditions $B_{1,\dotsc, i-1} \leftrightarrow T_{1,\dotsc, i-1} E \leftrightarrow T_i$ then trivially hold and the condition on the max-tradeoff function reads
	\[
		f_{\max}(p) \geq \sup_{\sigma_{R_{i-1}R'}:\mathcal{N}_i(\sigma)} H\left(  B_i | T_i R' \right)_{\mathcal{N}_i(\sigma)} \;.
	\]
	By the definition of the EAT channels $\{\mathcal{N}_i\}$, $B_i \neq \perp$ only for $T_i = 1$, which happens with probability~$\gamma$. Hence, for any state $\sigma_{R_{i-1}R'}$ we have, 
	\begin{align*}
		 H\left(  B_i | T_i R' \right)_{\mathcal{N}_i(\sigma)} \leq  H\left(  B_i | T_i \right)_{\mathcal{N}_i(\sigma)} \leq \gamma 
	\end{align*}
	and the max-tradeoff function is simply $f_{\max}(p) = \gamma$ for any $p$ (and thus $\|  \triangledown f_{\max} \|_\infty =0$). Applying\footnote{Here a slightly more general version of the EAT than the one given in this paper is needed, in which the event $\Omega$ can be defined via $A,B,X,Y$ and not only $C$; see~\cite{dupuis2016entropy} for the details.} Theorem~\ref{thm:eat} with this choice of $f_{\max}$ we get
	\begin{equation}\label{eq:bound_max_entropy}
		H^{\frac{\varepsilon_{\text{s}}}{4}}_{\max} \left( \mathbf{B} | \mathbf{ T}  E \right)_{\rho_{|\hat{\Omega}}} < \gamma n + \sqrt{n} 2\log7\sqrt{1-2\log \left( \varepsilon_{\text{s}}/4 \cdot \left(\varepsilon_{\mathrm{EA}} + \varepsilon_{\mathrm{EC}}\right) \right)} \;.
	\end{equation}

	Combining Equation~\eqref{eq:entropy_chain_1} with the above inequalities we get that
	\begin{equation*}
	\begin{split}
		H^{\varepsilon_{\text{s}}}_{\min} \left( \mathbf{A} | \mathbf{X Y T} O E \right)_{\tilde{\rho}_{|\tilde{\Omega}}} \geq H^{\frac{\varepsilon_{\text{s}}}{4}}_{\min} \left( \mathbf{AB} | \mathbf{X Y T}  E \right)_{\rho_{|\hat{\Omega}}} - \mathrm{leak_{EC}}  - 3 \log\left(1-\sqrt{1-(\varepsilon_{\text{s}}/4)^2}\right) \\
		  - \gamma n - \sqrt{n} 2\log7\sqrt{1-2\log \left( \varepsilon_{\text{s}}/4 \cdot \left(\varepsilon_{\mathrm{EA}} + \varepsilon_{\mathrm{EC}}\right) \right)}   \;.
	\end{split}
	\end{equation*}
	
	Finally, note that by applying the EAT on $\rho_{|\hat{\Omega}}$, as in Corollary~\ref{cor:main_generation}, we have that either $1-\Pr(\hat{\Omega})\geq 1-\varepsilon_{\mathrm{EA}} - \varepsilon_{\mathrm{EC}}$, or 
	\begin{equation*}
		H^{\frac{\varepsilon_{\text{s}}}{4}}_{\min} ( \mathbf{AB} | \mathbf{X Y T}  E )_{\rho_{|\hat{\Omega}}} > n \cdot \eta_{\mathrm{opt}}\left(\varepsilon_{\text{s}}/4,\varepsilon_{\mathrm{EA}} + \varepsilon_{\mathrm{EC}}\right) \;.
	\end{equation*}
	
	The last two equations together give us the desired bound on $H^{\varepsilon_{\text{s}}}_{\min} \left( \mathbf{A} | \mathbf{X Y T} O E \right)_{\tilde{\rho}_{|\tilde{\Omega}}}$: either the protocol aborts with probability greater than $1-\varepsilon_{\mathrm{EA}} - \varepsilon_{\mathrm{EC}}$ or 
	\begin{equation*}
	\begin{split}
		H^{\varepsilon_{\text{s}}}_{\min} \left( \mathbf{A} | \mathbf{X Y T} O E \right)_{\tilde{\rho}_{|\tilde{\Omega}}} \geq n \cdot \eta_{\mathrm{opt}}\left(\varepsilon_{\text{s}}/4,\varepsilon_{\mathrm{EA}} + \varepsilon_{\mathrm{EC}}\right) - \mathrm{leak_{EC}}  -  3 \log\left(1-\sqrt{1-(\varepsilon_{\text{s}}/4)^2}\right)  \\
		- \gamma n - \sqrt{n} 2\log7\sqrt{1-2\log \left( \varepsilon_{\text{s}}/4 \cdot \left(\varepsilon_{\mathrm{EA}} + \varepsilon_{\mathrm{EC}}\right) \right)} \;. \qedhere
	\end{split}
	\end{equation*}
\end{proof}

Using Lemma~\ref{lem:smooth_bound_qkd}, we prove that  Protocol~\ref{pro:diqkd_chsh} is sound.

\begin{lemma}\label{lem:QKD_sound}
	For any device $D$ let $\tilde{\rho}$ be the state generated using Protocol~\ref{pro:diqkd_chsh}. Then either the protocol aborts with probability greater than $ 1 - \varepsilon_{EA} - \varepsilon_{EC}$ or it is $(\varepsilon_{\mathrm{EC}} + \varepsilon_{\mathrm{PA}} + \varepsilon_{\text{s}})$-correct-and-secret while producing keys of length $\ell$, as defined in Equation~\eqref{eq:key_length_def}.
\end{lemma}

\begin{proof}
	Denote all the classical public communication during the protocol by $J=\mathbf{X Y T} O S$ where $S$ is the seed used in the privacy amplification protocol $\mathrm{PA}$. 
	Let $\overset{\approx}{\Omega}$ denote the event of not aborting Protocol~\ref{pro:diqkd_chsh} and $\tilde{\rho}_{\tilde{K}_A \tilde{K}_B JE|\overset{\approx}{\Omega}}$ be the final state of Alice, Bob, and Eve at the end of the protocol, \emph{conditioned on not aborting}.
	
	We consider two cases. First assume that the $\mathrm{EC}$ protocol was not successful (but did not abort). Then Alice and Bob's final keys might not be identical. This happens with probability at most~$\varepsilon_{\mathrm{EC}}$. 
	
	Otherwise, assume the $\mathrm{EC}$ protocol was successful, i.e., $K_B = \mathbf{A}$. In that case, Alice and Bob's keys must be identical also after the final privacy amplification step. That is, conditioned on $K_B = \mathbf{A}$, $\tilde{K}_A=\tilde{K}_B$.
		
	We continue to show that in this case the key is also secret. The secrecy depends only on the privacy amplification step, and for universal hashing a secure key is produced as long as Equation~\eqref{eq:universal_hashing_length} holds. 
	Hence, a uniform and independent key of length $\ell$ as in Equation~\eqref{eq:key_length_def} is produced by the privacy amplification step unless the smooth min-entropy is not high enough (i.e., the bound in Equation~\eqref{eq:final_min_entropy_bound_qkd} does not hold)
	or the privacy amplification protocol was not successful, which happens with probability at most $\varepsilon_{\mathrm{PA}} + \varepsilon_{\text{s}}$.
		
	According to Lemma~\ref{lem:smooth_bound_qkd}, either the protocol aborts with probability greater than $1-\varepsilon_{\mathrm{EA}}-\varepsilon_{\mathrm{EC}}$, or the entropy is sufficiently high for us to have
	\begin{equation*}
		\| \tilde{\rho}_{\tilde{K}_A JE|\overset{\approx}{\Omega}} - \rho_{U_l} \otimes \tilde{\rho}_{JE} \|_1 \leq \varepsilon_{\mathrm{PA}} + \varepsilon_{\text{s}}  \;.
	\end{equation*} 
	
	Combining both cases above we get that Protocol~\ref{pro:diqkd_chsh} is sound (that is, it produces identical and secret keys of length $\ell$ for Alice and Bob) with soundness error at most $ \varepsilon_{\mathrm{EC}} + \varepsilon_{\mathrm{PA}} + \varepsilon_{\text{s}}$. \qedhere
\end{proof}

\subsection{Key rate analysis}\label{sec:rates_comparison_qkd}

Theorem~\ref{thm:QKD_security} establishes a relation between the length $\ell$ of the secure key produced by our protocol and the different error terms. As this relation, given in Equation~\eqref{eq:key_length_def}, is somewhat hard to visualise, we analyse the key rate $r=\ell/n$ for some specific choices of parameters and compare it to the key rates achieved in device-dependent QKD with finite resources~\cite{scarani2008quantum,scarani2008security} and DIQKD with infinite resources and a restricted set of attacks~\cite{pironio2009device}. 

The key rate depends on the amount of leakage of information due to the error correction step, which in turn depends on the honest implementation of the protocol. We use the honest i.i.d.\@ implementation described in Section~\ref{sec:honest_qkd_imp} and assume that in the honest case the state of each round is the two-qubit Werner state $\rho_{Q_AQ_B} = (1-\nu) \ket{\phi^+}\bra{\phi^+} + \nu\mathbb{I}/4$ (and the measurements are as described in Section~\ref{sec:honest_qkd_imp}). 
The quantum bit error rate is then $Q=\frac{\nu}{2}$ and the expected winning probability is $\omega_{\mathrm{exp}}=\frac{2+\sqrt{2}(1-2Q)}{4}$. 

We emphasise that this is an assumption regarding the \emph{honest} implementation and it does not in any way restrict the actions of the adversary (and, in particular, the types of imperfections in the device). Furthermore, the analysis done below can be adapted to any other honest implementation of interest.

\subsubsection{Leakage due to error correction}\label{sec:leakage_ec_calc}

To compare the rates we first need to explicitly upper bound the leakage of information due to the error correction protocol, $\mathrm{leak_{EC}}\;$. As mentioned before, this can be done by evaluating $H_{0}^{\varepsilon'_{\mathrm{EC}}}(\mathbf{A}|\tilde{\mathbf{B}}\mathbf{XYT})$ on Alice and Bob's state in an honest i.i.d.\@ implementation of the protocol, described in Section~\ref{sec:honest_qkd_imp}. 

For this we first use the following relation between $H_{0}^{\varepsilon}$ and $H_{\max}^{\varepsilon'}$~\cite[Lemma 18]{tomamichel2011leftover}:
\[
	H_{0}^{\varepsilon'_{\mathrm{EC}}}(\mathbf{A}|\tilde{\mathbf{B}}\mathbf{XYT}) \leq H_{\max}^{\frac{\varepsilon'_{\mathrm{EC}}}{2}}\left(\mathbf{A}|\tilde{\mathbf{B}}\mathbf{XYT}\right) + \log \left( 8/\varepsilon'^2_{\mathrm{EC}} + 2/\left(2-\varepsilon'_{\mathrm{EC}}\right)\right) \;.
\]

The non-asymptotic version of the asymptotic equipartition property~\cite[Theorem 9]{tomamichel2009fully} (see also~\cite[Result 5]{tomamichel2012framework}) tells us that 
\[
	H_{\max}^{\frac{\varepsilon'_{\mathrm{EC}}}{2}}\left(\mathbf{A}|\tilde{\mathbf{B}}\mathbf{XYT}\right) \leq n H(A_i|\tilde{B}_iX_i Y_i T_i) + \sqrt{n} \delta (\varepsilon'_{\mathrm{EC}}, \tau) \;,
\]
for $\tau =2 \sqrt{2^{H_{\max}(A_i|\tilde{B}_iX_i Y_i T_i)}}+1$ and $\delta (\varepsilon'_{\mathrm{EC}}, \tau) = 4\log \tau \sqrt{2 \log \left(8/ {\varepsilon'}^2_{\mathrm{EC}} \right)}$. 

For the honest implementation of Protocol~\ref{pro:diqkd_chsh}, $H_{\max}(A_i|\tilde{B}_iX_i Y_i T_i)=1$  and 
\begin{align*}
	H(A_i|\tilde{B}_iX_i Y_i T_i) = &\Pr(T_i=0) \cdot H(A_i|\tilde{B}_i X_i Y_i, T_i=0) + \\
	&\Pr(T_i=1) \cdot H(A_i|\tilde{B}_iX_i Y_i ,T_i=1) \\
	=&\left( 1-\gamma \right) \cdot H(A_i|\tilde{B}_i X_i Y_i, T_i=0) + \\
	&\gamma \cdot H(A_i|\tilde{B}_iX_i Y_i ,T_i=1) \\
	=& \left( 1-\gamma \right)  h(Q) + \gamma h(\omega_{\mathrm{exp}}) \;,
\end{align*}
where the first equality follows from the definition of conditional entropy and the second from the way $T_i$ is chosen in Protocol~\ref{pro:diqkd_chsh}. The last equality holds since for generation rounds the error rate (i.e., the probability that $A_i$ and $\tilde{B}_i$ differ) in the honest case is $Q$ and for test rounds given $\tilde{B}_i,X_i$ and $Y_i$ Bob can guess $A_i$ correctly with probability $\omega_{\mathrm{exp}}$. 

We thus have 
\[
\begin{split}
	H_{0}^{\varepsilon'_{\mathrm{EC}}}\left(\mathbf{A}|\tilde{\mathbf{B}}\mathbf{XYT}\right) \leq n \left[\left( 1-\gamma \right)  h(Q) + \gamma h(\omega_{\mathrm{exp}}) \right] + \sqrt{n} 4\log \left(2\sqrt{2} +1\right) \sqrt{2 \log \left(8/ {\varepsilon'}^2_{\mathrm{EC}} \right)}  \\
	+ \log \left( 8/\varepsilon'^2_{\mathrm{EC}} + 2/\left(2-\varepsilon'_{\mathrm{EC}}\right)\right) \;.
	\end{split}
\]
Plugging this into Equation~\eqref{eq:ec_leakage} we get 
\begin{equation}\label{eq:explicit_leak_ec}
	\begin{split}
		\mathrm{leak_{EC}} \leq n \left[\left( 1-\gamma \right)  h(Q) + \gamma h(\omega_{\mathrm{exp}}) \right] + \sqrt{n} 4\log \left(2\sqrt{2} +1\right) \sqrt{2 \log \left(8/ {\varepsilon'}^2_{\mathrm{EC}} \right)}\\
		+ \log \left( 8/\varepsilon'^2_{\mathrm{EC}} + 2/\left(2-\varepsilon'_{\mathrm{EC}}\right)\right) + \log\left(\frac{1}{\varepsilon_{\mathrm{EC}}}\right) \;.
	\end{split}
\end{equation}

\subsubsection{Key rate curves}\label{sec:qkd_curves}

In Appendix~\ref{sec:better_rate} a slightly modified protocol is considered in which, instead of fixing the number of rounds in the protocol, only the expected number of rounds is fixed.

The completeness and soundness proofs follow the same lines as the proofs above, as detailed in Appendix~\ref{sec:better_rate}. The analysis presented in the appendix leads to improved key rates for the modified protocol, and are the ones presented here.\footnote{The key rate curves for a fixed number of rounds $n$ have the same shape as the curves presented here but require more signals to achieve the same rates (the difference is roughly two orders of magnitude).} 

\begin{figure}
\centering
\begin{tikzpicture}
	\begin{axis}[
		height=8.5cm,
		width=15cm,
		xlabel=$Q(\%)$,
		ylabel=$r$,
		xmin=0,
		xmax=0.145,
		ymax=1,
		ymin=0,
	     xtick={0.02,0.06,0.1,0.142},
	     xticklabels={$1$, $3$, $5$, $7.1$},
          ytick={0,0.1,0.2,0.3,0.4,0.5,0.6,0.7,0.8,0.9,1},
		legend style={at={(0.88,0.97)},anchor=north,legend cell align=left,font=\footnotesize} 
	]
	

	\addplot[blue,thick,smooth] coordinates {
		(1.*10^-10, 0.997348) (0.00584, 0.91723) (0.01168, 0.854049) (0.01752, 0.797357) (0.02336, 0.744869) (0.0292, 0.695502) (0.03504, 0.648617) (0.04088, 0.603795) (0.04672, 0.560736) (0.05256, 0.519218) (0.0584, 0.479068) (0.06424, 0.44015) (0.07008, 0.402351) (0.07592, 0.365578) (0.08176, 0.329752) (0.0876, 0.294806) (0.09344, 0.260681) (0.09928, 0.227327) (0.10512, 0.194699) (0.11096, 0.162758) (0.1168, 0.131467) (0.12264, 0.100796) (0.12848, 0.0707151) (0.13432, 0.0411983) (0.14016, 0.0122219) (0.146, -0.0162357)	
	};
	\addlegendentry{$\bar{n}=10^{15}$}
	
	\addplot[red,thick,smooth,dotted] coordinates {
	(1.*10^-10, 0.922753) (0.00584, 0.85326) (0.01168, 0.794151) (0.01752, 0.740126) (0.02336, 0.689641) (0.0292, 0.641886) (0.03504, 0.596356) (0.04088, 0.552704) (0.04672, 0.510679) (0.05256, 0.470087) (0.0584, 0.430779) (0.06424, 0.392632) (0.07008, 0.355544) (0.07592, 0.319433) (0.08176, 0.284226) (0.0876, 0.249862) (0.09344, 0.216287) (0.09928, 0.183453) (0.10512, 0.15132) (0.11096, 0.11985) (0.1168, 0.0890098) (0.12264, 0.0587699) (0.12848, 0.0291029) (0.13432, -0.0000156698) 
	};
	\addlegendentry{$\bar{n}=10^{10}$}

	\addplot[orange,thick,smooth,dashed] coordinates {
	(1.*10^-10, 0.694276) (0.00584, 0.636597) (0.01168, 0.586384) (0.01752, 0.539418) (0.02336, 0.494819) (0.0292, 0.452125) (0.03504, 0.411042) (0.04088, 0.371362) (0.04672, 0.332928) (0.05256, 0.295616) (0.0584, 0.259328) (0.06424, 0.22398) (0.07008, 0.189504) (0.07592, 0.15584) (0.08176, 0.122936) (0.0876, 0.090747) (0.09344, 0.0592343) (0.09928, 0.028362) (0.10512, -0.00190154) 
	};
	\addlegendentry{$\bar{n}=10^{8}$}
	
	\addplot[green,thick,smooth,dashdotted] coordinates {
	(1.*10^-10, 0.400565) (0.00584, 0.350245) (0.01168, 0.305868) (0.01752, 0.264116) (0.02336, 0.224246) (0.0292, 0.186061) (0.03504, 0.149213) (0.04088, 0.113476) (0.04672, 0.078735) (0.05256, 0.0448971) (0.0584, 0.011889) (0.06424, -0.02035) 
	};
	\addlegendentry{$\bar{n}=10^{7}$}

	\end{axis}  
\end{tikzpicture}

\caption{The expected key rate $r=\ell/\bar{n}$ as a function of the quantum bit error rate $Q$ for several values of the expected number of rounds $\bar{n}$. For $\bar{n}=10^{15}$ the curve essentially coincides with the curve for the i.i.d.\@ asymptotic case~\cite[Equation (12)]{pironio2009device}. 
	The following values for the error terms were chosen: $\varepsilon_{\mathrm{EC}}=10^{-10}, \; \varepsilon_{\mathrm{QKD}}^s=10^{-5}$ and $\varepsilon_{\mathrm{QKD}}^c=10^{-2}$.}
\label{fig:qkd_rates_Q_mod}
\end{figure}
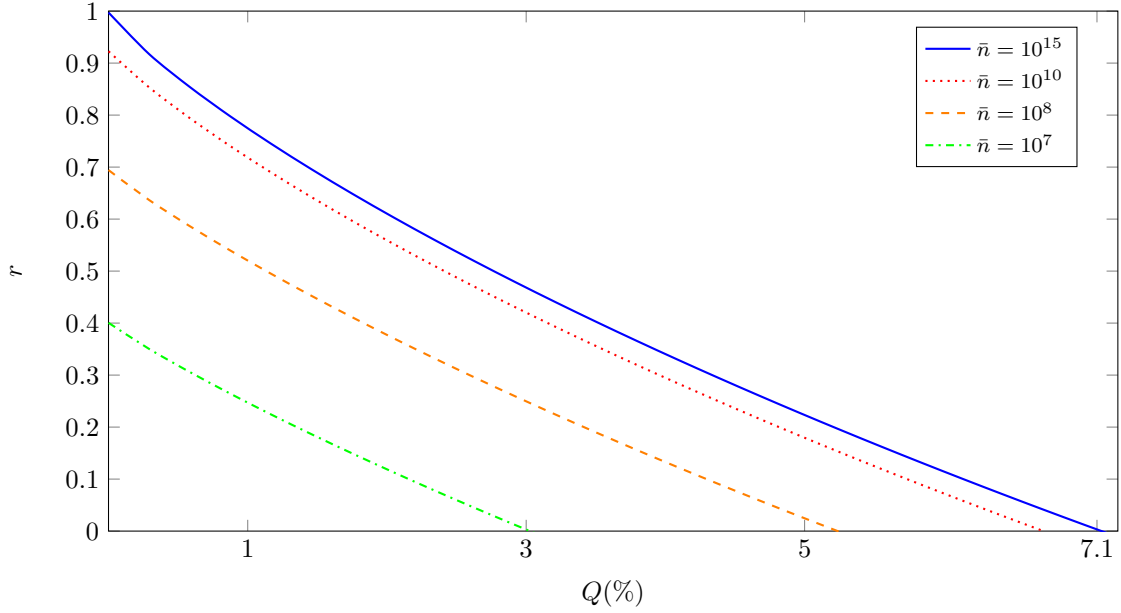

\begin{figure}
\centering
\begin{tikzpicture}
	\begin{axis}[
		height=8.5cm,
		width=15cm,
		xlabel=$\bar{n}$,
		ylabel=$r$,
		xmin=6,
		xmax=15,
		ymax=1,
		ymin=0,
	     xtick={6,7,8,9,10,11,15},
	     xticklabels={ $10^6$, $10^7$, $10^8$, $10^9$,$10^{10}$,$10^{11}$,$10^{15}$},
          ytick={0,0.1,0.2,0.3,0.4,0.5,0.6,0.7,0.8,0.9,1},
		legend style={at={(0.88,0.97)},anchor=north,legend cell align=left,font=\footnotesize} 
	]
	

	\addplot[blue,thick,smooth] coordinates {
	(6, -0.2442) (13/2, 0.0866932) (7, 0.318636) (15/2, 0.483919) (8, 0.601281) (17/2, 0.683763) (9, 0.741528) (19/2, 0.782662) (10, 0.810984) (21/2, 0.830336) (11, 0.841544) (23/2, 0.850964) (12, 0.856629) (25/2, 0.863243) (13, 0.867006) (27/2, 0.869136) (14, 0.870339) (29/2, 0.871017) (15, 0.871399)
	};
	\addlegendentry{$Q = 0.5\%$}
	
	\addplot[red,thick,smooth,dotted] coordinates {
	(6, -0.459711) (13/2, -0.152595) (7, 0.059624) (15/2, 0.208344) (8, 0.311673) (17/2, 0.383259) (9, 0.432467) (19/2, 0.465061) (10, 0.489571) (21/2, 0.504051) (11, 0.512275) (23/2, 0.520235) (12, 0.526662) (25/2, 0.531429) (13, 0.53412) (27/2, 0.535636) (14, 0.536489) (29/2, 0.536969) (15, 0.53724)
	};
	\addlegendentry{$Q = 2.5\%$}

	\addplot[orange,thick,smooth,dashed] coordinates {
	(6, -0.683142) (13/2, -0.398978) (7, -0.203553) (15/2, -0.0685348) (8, 0.02432) (17/2, 0.0884077) (9, 0.131977) (19/2, 0.160513) (10, 0.181837) (21/2, 0.193957) (11, 0.202814) (23/2, 0.207513) (12, 0.214511) (25/2, 0.218461) (13, 0.220688) (27/2, 0.221941) (14, 0.222646) (29/2, 0.223043) (15, 0.223266)
	};
	\addlegendentry{$Q = 5\%$}

	\end{axis}  
\end{tikzpicture}
\caption{The expected key rate $r=\ell/\bar{n}$ as a function of the expected number of rounds $\bar{n}$ for several values of the quantum bit error rate $Q$. For $Q=0.5\%,\;2.5\%$, and $5\%$ the achieved key rates are approximatly $r=87\%,\;53\%$, and $22\%$ respectively.
	The following values for the error terms were chosen: $\varepsilon_{\mathrm{EC}}=10^{-10}, \; \varepsilon_{\mathrm{QKD}}^s=10^{-5}$ and $\varepsilon_{\mathrm{QKD}}^c=10^{-2}$.}
\label{fig:qkd_rates_n_mod}
\end{figure}
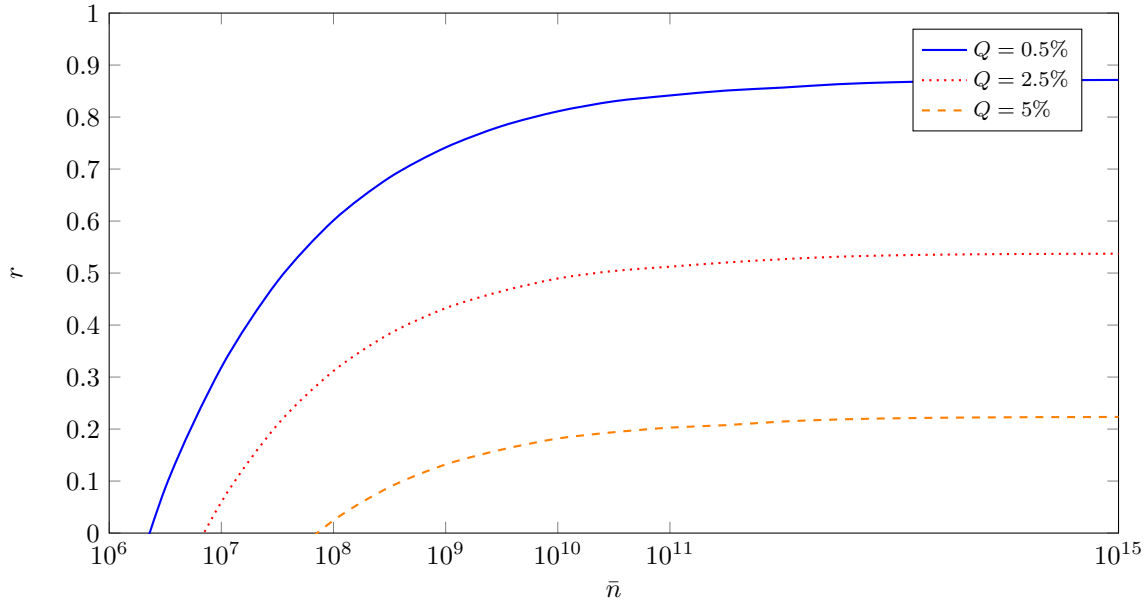

In Figure~\ref{fig:qkd_rates_Q_mod} the expected key rate $r=\ell/\bar{n}$ is plotted as a function of the quantum bit error rate $Q$ for several values of  the expected number of rounds $\bar{n}$. For $\bar{n}=10^{15}$ the curve already essentially coincides with the key rate achieved in the \emph{asymptotic i.i.d.\@} case, that is, when restricting the adversary to collective attacks~\cite[Equation~(12)]{pironio2009device} (see also Figure~2 therein). As the key rate for the asymptotic i.i.d.\@ case was shown to be optimal in~\cite{pironio2009device} (for practically the same protocol) it acts as an upper bound on the key rate and the amount of tolerable noise for the general case considered in this work. Hence, for large enough number of rounds our key rate becomes optimal and the protocol can tolerate up to the maximal error rate $Q=7.1\%$.

In an asymptotic analysis (i.e., with infinite resources $\bar{n}\rightarrow\infty$) it is well understood that the soundness and completeness errors $\varepsilon_{\mathrm{QKD}}^s,\varepsilon_{\mathrm{QKD}}^c$ should tend to zero as $\bar{n}$ increases. However, in the non-asymptotic scenario considered here these errors are always finite. We therefore fix some values for them which are considered to be realistic and relevant for actual applications. We choose the parameters such that the security parameters are at least as good (and in general even better) as in~\cite{scarani2008quantum}, such that a fair comparison can be made. All other parameters are chosen in a consistent way while (roughly) optimising the key rate.

In Figure~\ref{fig:qkd_rates_n_mod} $r$ is plotted as a function of $\bar{n}$ for several values of $Q$. As can be seen from the figure, the achieved rates are significantly higher than those achieved in previous works. Moreover, they are practically comparable to the key rates achieved in device-\emph{dependent} QKD (see Figure~1 in~\cite{scarani2008quantum}). The main difference between the curves for the device-dependent case  and the independent one is the minimal value of $\bar{n}$ which is required for a positive key rate. (That is, for the protocols considered in~\cite{scarani2008quantum} one can get a positive key rate with less rounds.) It is possible that by further optimising the parameters a positive key rate can also be achieved in our setting in the regime $\bar{n}=10^4-10^6$ for the different error rates.

\section{Randomness expansion}\label{sec:expansion}
We show how the entropy accumulation protocol can be used to perform randomness expansion. This can be achieved based on any non-local game for which one is able to prove a good bound in Equation~\eqref{eq:eat_f_min_bound}. For concreteness we focus on the CHSH game, for which an explicit bound is provided by Corollary~\ref{cor:main_generation}. Although the protocol can be used to achieve larger expansion factors, we give specific bounds that optimise the linear output rate, under the assumption that a small linear number of uniformly random bits is available to the experimenter for the execution of the protocol. 

In order to minimise the amount of randomness required to execute the protocol we adapt the main entropy accumulation protocol, Protocol~\ref{pro:randomness_generation}, by deterministically choosing inputs in the generation rounds from $\mathcal{X}_g = \{0\}$ and $\mathcal{Y}_g = \{0\}$. In particular there is no use for the input $2$ to the $B$ device, and no randomness is required for the generation rounds.~\footnote{This requires both users to know which rounds are selected as generation rounds, i.e. to share the random variable $T_i$. For the purposes of randomness expansion this does not even require communication as we may assume the parties are co-located.} Aside from the last step of randomness amplification the remainder of the protocol is essentially the same as Protocol~\ref{pro:randomness_generation} (in its instantiation with the CHSH game considered in Section~\ref{sec:entropy-chsh}). The complete protocol is described as Protocol~\ref{pro:randomness_expansion}.

\begin{algorithm}
\caption{Randomness expansion protocol}
\label{pro:randomness_expansion}
\begin{algorithmic}[1]
	\STATEx \textbf{Arguments:} 
		\STATEx\hspace{\algorithmicindent} $G$ -- CHSH game restricted to $\mathcal{X}_t=\mathcal{Y}_t=\{0,1\}$.
		\STATEx\hspace{\algorithmicindent} $D$ -- untrusted device of two components that can play $G$ repeatedly
		\STATEx\hspace{\algorithmicindent} $n \in \mathbb{N}_+$ -- number of rounds
		\STATEx\hspace{\algorithmicindent} $\gamma \in (0,1]$ -- expected fraction of test rounds 
		\STATEx\hspace{\algorithmicindent} $\omega_{\mathrm{exp}}$ -- expected winning probability in $G$ for an honest (perhaps noisy) implementation    
		\STATEx\hspace{\algorithmicindent} $\delta_{\mathrm{est}} \in (0,1)$ -- width of the statistical confidence interval for the estimation test
		
	\STATEx
	
	\STATE For every round $i\in[n]$ do Steps~\ref{prostep:choose_Ti}-\ref{prostep:calculate_Ci_RE}:
		\STATE\hspace{\algorithmicindent}Bob chooses a random bit $T_i\in\{0,1\}$ such that $\Pr(T_i=1)=\gamma$. \label{prostep:choose_Ti}
		\STATE\hspace{\algorithmicindent}If $T_i=0$ Alice and Bob choose $(X_i,Y_i)=(0,0)$. If $T_i=1$ they choose uniformly random inputs $(X_i,Y_i)\in\mathcal{X}_t\times\mathcal{Y}_t$.
		\STATE\hspace{\algorithmicindent}Alice and Bob use $D$ with $X_i,Y_i$ and record their outputs as $A_i$ and $B_i$ respectively. \label{prostep:measurement}
		\STATE\hspace{\algorithmicindent}If $T_i=1$ they set $C_i = w(A_i,B_i,X_i,Y_i)$.\label{prostep:calculate_Ci_RE}
	\STATE Alice and Bob abort if $\sum_j C_j < \left(\omega_{\mathrm{exp}}\gamma - \delta_{\mathrm{est}}\right) \cdot n$. \label{prostep:abort_RE}
	\STATE They return $\mathrm{Ext}(\mathbf{AB},\mathbf{S})$ where $\mathrm{Ext}$ is the extractor from Lemma~\ref{lem:ext_tre} and $\mathbf{S}$ is a uniformly random seed. \label{prostep:ext_RE}
\end{algorithmic}
\end{algorithm}

 Corollary~\ref{cor:main_generation} provides a lower bound on the min-entropy generated by the protocol. Given we are concerned here not only with \emph{generating} randomness, but also with \emph{expanding} the amount of randomness initially available to users of the protocol, we now evaluate the total number of random bits that is needed to execute Protocol~\ref{pro:randomness_expansion}.

\paragraph{Input randomness.}
Random bits are required to select which rounds are generation rounds, i.e. the random variable $\mathbf{T}$, to select inputs to the devices in the testing rounds, i.e. those for which $T_i=0$, and to select the seed for the extractor in Step~\ref{prostep:ext_RE}.

The random variables $T_i$ are chosen independently according to a biased Bernoulli$(\gamma)$ distribution. The following lemma shows that approximately $8h(\gamma) n$ uniformly random bits are sufficient to generate the $T_i$, provided one allows for the possibility of a small deviation error.  

\begin{lemma}\label{lem:generate-inputs}
	Let $\gamma>0$. There is an efficient procedure such that for any integer $n$, given $r=8h(\gamma) n$ uniformly random bits as inputs the procedure either aborts, with probability at most $\varepsilon_{\mathrm{SA}}=\exp(-\Omega(\gamma^3 \log^{-2}\gamma n))$, or outputs $n$ bits $T_1,\ldots,T_n$ whose distribution is within statistical distance at most $\varepsilon_{\mathrm{SA}}$ of $n$ i.i.d.\@ Bernoulli$(\gamma)$ random variables. 
\end{lemma}

\begin{proof}
	It is well-known that using the interval algorithm~\cite{hoshi1997interval} it is possible to sample exactly from $m$ i.i.d. Bernoulli($\gamma$) random variables using an expected number of random bits at most $h(\gamma)m+2$; furthermore the maximum number of random bits needed is at most $C m\log \gamma^{-1}$ for some constant $C$. 

	In order to obtain a bound on the maximum number of random bits used that holds with high probability, let $\alpha = h(\gamma)$ and partition $\{1,\ldots,n\}$ into at most $t=\lceil\alpha n\rceil$  chunks of $m=\lceil 1/\alpha\rceil$ consecutive integers each.
	Suppose we repeat the interval algorithm for each chunk. Let $N_i$ be the number of uniform bits used to generate the $T_j$ associated with the $i$-th chunk. Then by the above $\mathrm{E}[N_i]\leq h(\gamma)m+2$ and $N_i \leq C m\log \gamma^{-1}$. Applying the Hoeffding inequality,
	\begin{align*}
		\Pr\Big( \sum_{i=1}^t N_i > 2 (h(\gamma)m+2)t \Big) \leq e^{-C' \frac{(h(\gamma)m+2)^2}{m^2 \log^2 \gamma^{-1}} t}
		&\leq e^{-C'' \frac{\gamma^3}{\log^2 \gamma^{-1}} n}
	\end{align*}
	for some constants $C',C''>0$ and given our choice of $\alpha$. Using $mt\leq 2n$ and $t\leq n$ gives the claimed bound. 
\end{proof}

\begin{remark}
If one is willing to settle for a bound on the number of uniform bits used \emph{in expectation} then using the procedure from~\cite{hoshi1997interval} it is possible to  \emph{exactly} sample $n$ i.i.d.\@ Bernoulli$(\gamma)$ random variables using an expected number of random bits at most $h(\gamma)n+2$.
\end{remark}

It remains to account for the random bits required to generate inputs in the testing rounds, for which $T_i=0$. By Hoeffding's inequality there are at most $2\gamma n$ such rounds except with probability $\exp(-\Omega(\gamma^2 n))\leq \varepsilon_{\mathrm{SA}}$ for large enough $n$. Together with Lemma~\ref{lem:generate-inputs} we conclude that $10\gamma n$ uniformly random bits are sufficient to execute the protocol with a probability of success (up to but not including step~\ref{prostep:abort_RE}) at least $1-e^{-\tilde{\Omega}(\gamma^3)n}$. We also note that if one is only concerned with the expected number of random bits used then $(h(\gamma) + \gamma)n+2$ bits are sufficient.

\paragraph{Extraction.}
In the last step of the protocol, Step~\ref{prostep:ext_RE}, the user applies a quantum-proof extractor to  $\mathbf{AB}$ in order to produce a random string that is close to being uniformly distributed. This step requires the use of an additional seed $\mathbf{S}$ of uniformly random bits. We use the following construction based on Trevisan's extractor, designed to maximise the output length while not using too much seed. 

\begin{lemma}\label{lem:ext_tre}
	For any $\delta>0$ there is a $c=c(\delta)>0$ such that the following holds. For all large enough integer $n$ and any $k\geq \delta n$ there is an efficient procedure $\mathrm{Ext}:\{0,1\}^{2n} \times \{0,1\}^d \to \{0,1\}^m$ such that $d = \lceil \delta n\rceil$ and $m = \lceil k -9\log k \rceil$, and is such that for $\varepsilon_{\mathrm{EX}}=\exp(-c (n/\log n)^{1/2})$ and any classical-quantum state $\rho_{\mathbf{A}E}$ such that $H_{\min}^{\varepsilon_{\mathrm{EX}}}(\mathbf{A}|E)_\rho \geq k$ it holds that 
	\[
		\| \rho_{\mathrm{Ext}(\mathbf{A},\mathbf{S})\mathbf{S}E} - \rho_{U_m}\otimes \rho_{U_d} \otimes \rho_E \|_1\leq 2\varepsilon_{\mathrm{EX}}\;,
	\]
	where $\mathbf{S}\in\{0,1\}^d$ is a uniformly distributed random seed and $\rho_{U_m},\rho_{U_d}$ are totally mixed states on $m$ and $d$ bits respectively. 
\end{lemma}

\begin{proof}
	We use the construction given in~\cite[Corollary~5.1]{de2012trevisan}. To get the parameters stated here we note that provided $c$ is chosen small enough with respect to $\delta$ our choice of $\varepsilon_{\mathrm{EX}}$ ensures that the seed length $d = O(\log^2(n/\varepsilon_{\mathrm{EX}})\log m)$ can be made smaller than $\delta n$. The conclusion on the trace distance follows from the guarantee of strong extractor given by~\cite[Corollary~5.1]{de2012trevisan} using an argument similar to the proof of~\cite[Lemma 17]{arnon2015quantum}.
\end{proof}

We state the results of the above discussions as the following theorem stating the guarantees of the randomness expansion protocol.

\begin{thm}\label{thm:expansion}
	Let $\gamma,\delta>0$. Let $\varepsilon_{\mathrm{EX}}$ be as in Lemma~\ref{lem:ext_tre}. Then for all large enough $n$, $\varepsilon_{\mathrm{EA}} \in (0,1)$, and $\varepsilon_s$ such that $\varepsilon_{\mathrm{SA}} < \varepsilon_s +\varepsilon_{\mathrm{SA}} \leq \varepsilon_{\mathrm{EX}}$, Protocol~\ref{pro:randomness_generation} is an $(\varepsilon_{EA}^c + \varepsilon_{\mathrm{SA}},2\varepsilon_{\mathrm{EX}})$-secure $[(8h(\gamma)+\delta)n]\to [n\cdot \eta_{\mathrm{opt}}(\varepsilon_s-\varepsilon_{\mathrm{SA}} ,\varepsilon_{\mathrm{EA}})-9\log n]$ randomness expansion protocol. That is, either Protocol~\ref{pro:randomness_generation} aborts with probability greater than $\varepsilon_{\mathrm{EA}}$ or it generates a string of length $m \geq n\cdot \eta_{\mathrm{opt}}(\varepsilon_s,\varepsilon_{\mathrm{EA}})-9\log n$ (where  $\eta_{\mathrm{opt}}$ is defined in Equation~\eqref{eq:eta_opt}) such that
	\[
		\| \rho_{ZRE} - \rho_{U_m}\otimes \rho_{RE} \|_1\leq 2\varepsilon_{\mathrm{EX}} + \varepsilon_{\mathrm{SA}}\;,
	\]
	where $R$ is a register holding all the initial random bits used in the protocol (including the seed $\mathbf{S}$). 
\end{thm}

\begin{proof}
	Let $D$ be any device and $\rho$ the state (as defined in Equation~\eqref{eq:final_state_before_abort}) generated right before Step~\ref{prostep:abort_RE} of Protocol~\ref{pro:randomness_expansion}. Let $\Omega$ (as defined in Equation~\eqref{eq:good_event_def}) be the event that the protocol does not abort, and $\rho_{|\Omega}$ the state conditioned on $\Omega$. Then, applying Corollary~\ref{cor:main_generation} we obtain that for any $\varepsilon_{\mathrm{EA}},\varepsilon'_{\text{s}}\in(0,1)$, either the protocol aborts with probability greater than $1-\varepsilon_{\mathrm{EA}}$ or
	\begin{equation}\label{eq:main_thmb}
		 H^{\varepsilon'_{\text{s}}+\varepsilon_{\mathrm{SA}}}_{\min} \left( \mathbf{A B} | \mathbf{X Y T} E \right)_{\rho_{|\Omega}} > n\cdot \eta_{\mathrm{opt}}(\varepsilon'_{\text{s}},\varepsilon_{\mathrm{EA}}) \;, 
	\end{equation}
	where  $\eta_{\mathrm{opt}}$ is defined in Equation~\eqref{eq:eta_opt} and the additional smoothness parameter $\varepsilon_{\mathrm{SA}}$ accounts for the error in the verifier's input sampling procedure, as described in Lemma~\ref{lem:generate-inputs}.\footnote{The $\log(13)$ term in the definition of $\eta$ in Equation~\eqref{eq:eta_opt} could be replaced by a $\log(9)$ to account for the fact that here $d_{B_i}=2$, instead of $d_{B_i}=3$ in Section~\ref{sec:entropy-chsh}.} Given the bound Equation~\eqref{eq:main_thmb}, the guarantee on $m$ claimed in the theorem follows from Lemma~\ref{lem:ext_tre}. 
	
	Finally, completeness of the protocol follows directly from completeness of the entropy accumulation protocol, Protocol~\ref{pro:randomness_generation}, as stated in Lemma~\ref{lem:ea_completeness}, and the verifier's input sampling procedure, described in Lemma~\ref{lem:generate-inputs}. 
\end{proof}

Assuming a choice $\delta=\gamma$, the number of random bits required in the protocol scales linearly, roughly as~$\sim 9\gamma n$. 

For $\gamma\to 1$, which corresponds to randomness generation (also called randomness certification, i.e. the guarantee of ``fresh'' randomness independent of the inputs) the values of $\eta_{opt}$ plotted in Figure~\ref{fig:eta_rates} give a good idea of the rate of randomness generation that can be achieved from Protocol~\ref{pro:randomness_expansion} for different choices of the security parameters. The rate is to be compared to the rate that was shown achievable for randomness expansion in the case of \emph{classical} adversaries only in~\cite{pironio2010random} (Figure~2; see also~\cite{pironio2013security}). Our result, in contrast, holds against \emph{quantum} adversaries. The rate is much better than the ones obtained in~\cite{vazirani2012certifiable,miller2014robust}. 

For randomness expansion one would select a small value of $\gamma$. If $\gamma$ is a small constant then Theorem~\ref{thm:expansion} guarantees a constant expansion factor. One may wish to go further, by selecting a $\gamma$ that scales sub-linearly, polynomially or even poly-logarithmically,  in the number of rounds (e.g., to achieve exponential expansion).
It is possible to adapt our results to guarantee a linear production of  randomness even for such parameters by suitably adapting Protocol 3 so  that rounds are grouped in blocks, as in the modification of the entropy accumulation protocol described in Appendix~\ref{sec:better_rate}.

\section{Open questions}\label{sec:open_questions}

Several questions are left open. 
\begin{enumerate}
	\item Our results yield essentially optimal values for the leading and second-order constants, $c_1$ and $c_2$, that govern the achievable rate curves. 
	As loophole-free Bell tests (a necessity for DI cryptography) are finally being realized~\cite{hensen2015loophole,shalm2015strong,giustina2015significant}, it becomes increasingly relevant to achieve the best possible dependence of the rate curves on the number of rounds $n$, even for very small values of $n$. As can be seen from Figures~\ref{fig:eta_rates} and~\ref{fig:qkd_rates_Q_mod} our rate curves approach (and essentially coincide) with the optimal curves as the number of rounds increases. One thing that can perhaps still be further optimized is the dependency on the number of rounds, or in other words, how fast the curves approach the optimal curve. The explicit dependency on $n$ given in Equation~\eqref{eq:key_length_def} is already close to optimal, but the numerical analysis used to plot the curves can be made somewhat better for the range of $n=10^4-10^6$. Although this seems like a minor issue, it can make actual implementations more feasible. 
	\item Are there similar protocols, based on a different Bell inequality, that can lead to better entropy rates?  To apply our proof to other Bell inequalities one should find a good bound on the min-tradeoff function, as done in Equation~\eqref{eq:one_box_entropy_final} for the CHSH inequality. For many Bell inequalities such bounds are known, but for the min-entropy instead of the von Neumann entropy. In most cases using a bound on the min-entropy will result in far from optimal rate curves. Therefore, to adapt our protocol to other Bell inequalities one should probably bound the min-tradeoff function using the von Neumann entropy directly. Unfortunately, we do not know of any general technique to achieve such tight bounds. 
	\item Are there other protocols, e.g., with two-way classical post-processing, which achieve better key rates? The optimality of our key rates is only with respect to the structure of the considered protocol.
\end{enumerate}

\paragraph{Acknowledgements.} 
The authors thank Fr\'{e}d\'{e}ric Dupuis and Omar Fawzi for discussion about the Entropy Accumulation Theorem, as well as Max Kessler and Christopher Portman for helpful comments and J\'{e}r\'{e}my Ribeiro and Gl\'{a}ucia Murta for spotting a mistake in a previous version of this work.
RAF and RR were supported by the Stellenbosch Institute for Advanced Study (STIAS), by the European Commission via the project “RAQUEL”, by the Swiss National Science Foundation via the  National Center for Competence in Research, QSIT, and by the Air Force Office of Scientific Research (AFOSR) via grant~FA9550-16-1-0245. TV was partially supported by NSF CAREER Grant CCF-1553477, an AFOSR YIP award, the IQIM, and NSF Physics Frontiers Center (NFS Grant PHY-1125565) with support of the Gordon and Betty Moore Foundation (GBMF-12500028).

\appendix
\appendixpage

\section{Summary of Parameters, constants, and variables}

\begin{table}[H]
	\begin{center}
	\begin{tabular}{| l | l | l |} 
	\hline
		Symbol & Meaning & Relation to other parameters \\ \hhline{|=|=|=|}
		$n\in\mathbb{N}_+$ & Number of rounds  & \\ \hline
		$\gamma\in(0,1]$ & Expected fraction of Bell violation estimation rounds & \\ \hline
		$\omega_{\mathrm{exp}}\in[0,1]$ & \specialcell{Expected winning probability in an honest  (perhaps  \\ noisy) implementation} & \\ \hline
		$\delta_{\mathrm{est}} \in (0,1)$ & \specialcell{Width of the statistical confidence interval \\ for the Bell violation estimation test} & \\ \hline
		$\varepsilon_{\text{s}}$ & Smoothing parameter & \\ \hline
		$\varepsilon^c_{EA} $ & Completeness error of the entropy accumulation protocol & Given in Eq.~\eqref{eq:completeness_error_EA}\\ \hline
		$\varepsilon_{\mathrm{EA}}$ & The error probability of the entropy accumulation protocol & \\ \hline
		$\mathrm{leak_{EC}}$ & The leakage of the error correction protocol & Given in Eq.~\eqref{eq:ec_leakage} \\ \hline
		$\varepsilon_{\mathrm{EC}},\varepsilon'_{\mathrm{EC}}$ & Error probabilities of the error correction protocol & \\ \hline
		$\varepsilon_{\mathrm{EC}}^c$ & Completeness error of the error correction protocol & $\varepsilon_{\mathrm{EC}}^c=\varepsilon_{\mathrm{EC}}'+\varepsilon_{\mathrm{EC}}$\\ \hline
		$\varepsilon_{\mathrm{PE}}^c$ & Completeness error of the parameter estimation step & Given in Eq.~\eqref{eq:pe_completeness} \\ \hline
		$\varepsilon_{\mathrm{PA}}$ & Error probability of the privacy amplification protocol & Given in Eq.~\eqref{eq:universal_hashing_length}\\ \hline
		$\ell$ & Final key length in the DIQKD protocol & Given in Eq.~\eqref{eq:key_length_def} \\ \hline
		$\varepsilon^c_{QKD}$ & Completeness error of the DIQKD protocol & $\varepsilon^c_{QKD}\leq\varepsilon^c_{EC} + \varepsilon^c_{EA}+\varepsilon_{EC} $ \\ \hline
		$\varepsilon^s_{QKD}$ & Soundness error of the DIQKD protocol & $\varepsilon^s_{\mathrm{QKD}} \leq \varepsilon_{\mathrm{EC}} + \varepsilon_{\mathrm{PA}} + \varepsilon_{\text{s}}$ \\ \hline
	$\varepsilon_{\mathrm{SA}}$ & Error probability of the input sampling procedure used  & Given in Lemma~\ref{lem:generate-inputs}\\ 
	& in the randomness expansion protocol& \\ \hline
		$\varepsilon_{\mathrm{EX}}$ & Error probability of the extractor used in the & Given in Lemma~\ref{lem:ext_tre} \\
		& randomness expansion protocol& \\ \hline
	 \end{tabular}
	\end{center}\caption{Parameters and constants used throughout the paper}\label{tb:parameters}
\end{table}

\begin{table}[H]
	\begin{center}
	\begin{tabular}{| l | l |} 
	\hline
		Random variables and systems & Meaning \\ \hhline{|=|=|}
		$X_i\in\mathcal{X}$ & Alice's input in round $i\in[n]$ \\ \hline
		$Y_i\in\mathcal{Y}$ & Bob's input in round $i\in[n]$ \\ \hline
		$A_i\in\mathcal{A}$ & Alice's output in round $i\in[n]$ \\ \hline
		$B_i\in\mathcal{B}$ & Bob's output in round $i\in[n]$ \\ \hline
		$T_i\in\{0,1\}$ & \specialcell{Indicator of the estimation test in round $i$: \\ 
					$T_i = \begin{cases}
					0 &\text{$i$'th round is not a test round} \\
					1 &\text{$i$'th round is a test round} 
					\end{cases}$} \\ \hline
		$F_i\in\{0,1\}$ & A random uniform bit for the symmetrisation step in round $i\in[n]$ \\ \hline
		$C_i \in \{\perp,0,1\}$ & \specialcell{Indicator of the correlation in the test rounds: \\ 
					$C_i = \begin{cases}
					\perp &T_i =0 \\
					0 &\text{$T_i=1$ and the test fails} \\
					1 &\text{$T_i=1$ and the test succeeds} \;.
					\end{cases}$} \\ \hhline{|=|=|}
		$E$ & Register of Eve's quantum state \\ \hline
		$R_i$ & \specialcell{Register of the (unknown) quantum state $\rho_{Q_AQ_B}^i$ of Alice and \\ Bob's devices after step $i$ of the protocol, for $i\in\{0\}\cup[n]$.} \\ \hline
	 \end{tabular}
	\end{center}\caption{Random variables and quantum systems used throughout the paper}\label{tb:RV_table}
\end{table}

\section{An improved dependency on the test probability $\gamma$}\label{sec:better_rate}

In this section we show how the EAT can be used in a slightly different way than what was done in the main text. This results in an entropy rate which has a better dependency on the probability of a test round $\gamma$, compared to the entropy rate given in Equation~\eqref{eq:eta_opt} in Section~\ref{sec:entropy-chsh}. The improved entropy rate derived here is the one used for calculating the  key rates of the DIQKD protocol is Section~\ref{sec:qkd_curves}.

\subsection{Modified entropy accumulation protocol}

We use a different entropy accumulation protocol, given as Protocol~\ref{pro:randomness_generation_mod}.
In this modified protocol instead of considering each round separately we consider blocks of rounds. A block is defined by a sequence of rounds: in each round a test is carried out with probability $\gamma$ (and otherwise the round is a generation round). The block ends when a test round is being performed and then the next block begins.  If for $s_{\max}$ rounds there was no test, the block ends without performing a test and the next begins.  Thus, the blocks can be of different length, but they all consist at most $s_{\max}$ rounds.

In this setting, instead of fixing the number of rounds $n$ in the beginning of the protocol, we fix the number of blocks $m$. The expected length of  block is 
\begin{align}
	\bar{s} = \sum_{s\in [s_{\max}]}\left[ s (1-\gamma)^{(s-1)}\gamma \right] +s_{\max}(1-\gamma)^{s_{\max}} &=  \frac{1- (1-\gamma)^{s_{\max}}}{\gamma} \nonumber \\
	&= \sum_{s\in [s_{\max}]}\left[ (1-\gamma)^{(s-1)} \right] \;.\label{eq:exp_size}
\end{align}
The expected number of rounds is denoted by $\bar{n} = m\cdot \bar{s}$.

Compared to the main text, we now have a RV $\tilde{C}_j\in\{0,1,\perp\}$ for each block, instead of each round. Alice and Bob set $\tilde{C}_j$ to be $0$ or $1$ depending on the result of the game in the block's test round (i.e., the last round of the block), or $\tilde{C}_j=\perp$ if a test round was not carried out in the block. By the definition of the blocks we have $\Pr[\tilde{C}_j=\perp]=(1-\gamma)^{s_{\max}}$. 

\begin{algorithm}
\caption{Modified entropy accumulation protocol}
\label{pro:randomness_generation_mod}
\begin{algorithmic}[1]
	\STATEx \textbf{Arguments:} 
		\STATEx\hspace{\algorithmicindent} $G$ -- two-player non-local game
		\STATEx\hspace{\algorithmicindent} $\mathcal{X}_g,\mathcal{X}_t \subset \mathcal{X}$ -- generation and test inputs for Alice
		\STATEx\hspace{\algorithmicindent} $\mathcal{Y}_g,\mathcal{Y}_t \subset \mathcal{Y}$ -- generation and test inputs for Bob
		\STATEx\hspace{\algorithmicindent} $D$ -- untrusted device of (at least) two components that can play $G$ repeatedly
		\STATEx\hspace{\algorithmicindent} $m\in \mathbb{N}_+$ -- number of blocks
		\STATEx\hspace{\algorithmicindent} $s_{\max} \in \mathbb{N}_+$ --maximal length of a block
		\STATEx\hspace{\algorithmicindent} $\gamma \in (0,1]$ -- probability of a test round 
		
		\STATEx\hspace{\algorithmicindent} $\omega_{\mathrm{exp}}$ -- expected winning probability in $G$ for an honest (perhaps noisy) implementation    
		\STATEx\hspace{\algorithmicindent} $\delta_{\mathrm{est}} \in (0,1)$ -- width of the statistical confidence interval for the estimation test
		
	\STATEx
	
	\STATE For every block $j\in[m]$ do Steps~\ref{prostep:ini_block}-\ref{prostep:calculate_Ci_EA_mod}:
		\STATE\hspace{\algorithmicindent} Set $i=0$ and $C_j=\perp$.\label{prostep:ini_block}
		\STATE\hspace{\algorithmicindent} If $i \leq s_{\max}$:
		
		\STATE\hspace{\algorithmicindent}\hspace{\algorithmicindent} Set $i=i+1$.
		
		\STATE\hspace{\algorithmicindent}\hspace{\algorithmicindent}Alice and Bob choose $T_i\in\{0,1\}$ at random such that $\Pr(T_i=1)=\gamma$. 
		\STATE\hspace{\algorithmicindent}\hspace{\algorithmicindent}If $T_i=0$ Alice and Bob choose inputs $X_i\in\mathcal{X}_g$ and $Y_i\in \mathcal{Y}_g$ respectively. If $T_i=1$ they choose inputs $X_i\in\mathcal{X}_t$ and $Y_i\in \mathcal{Y}_t$. 
		\STATE\hspace{\algorithmicindent}\hspace{\algorithmicindent}Alice and Bob use $D$ with $X_i,Y_i$ and record their outputs as $A_i$ and $B_i$ respectively. \label{prostep:measurementb}
		\STATE\hspace{\algorithmicindent}\hspace{\algorithmicindent}If $T_i=0$ Bob updates $B_i$ to $B_i = \perp$.
		\STATE\hspace{\algorithmicindent}\hspace{\algorithmicindent}If $T_i=1$ they set $\tilde{C}_j =w\left(A_i,B_i,X_i,Y_i\right)$ and $i=s_{\max}+1$.\label{prostep:calculate_Ci_EA_mod}

	\STATE Alice and Bob abort if $\sum_{j\in[m]} \tilde{C}_j < \left[\omega_{\mathrm{exp}} \left(1-(1-\gamma)^{s_{\max}}\right) - \delta_{\mathrm{est}}\right] \cdot m$. \label{prostep:abort_modified_EA}
\end{algorithmic}
\end{algorithm}

Note that the symmetrisation step was dropped in Protocol~\ref{pro:randomness_generation_mod} just for simplicity, as it plays no role in the considered modification; it can (and should) be handled exactly as done in Section~\ref{sec:entropy-chsh}.

\subsection{Modified min-tradeoff function}

Below, we apply the EAT on blocks of outputs instead of single rounds directly. Let $\mathcal{M}_j$ denote the EAT channels defined by the actions of Steps~\ref{prostep:ini_block}-\ref{prostep:calculate_Ci_EA_mod} in Protocol~\ref{pro:randomness_generation_mod}, combined with the quantum channels that model the device's actions in those steps. 
It is easy to verify that $\mathcal{M}_j$ fulfil the necessary conditions given in Definition~\ref{def:eat_channels}. 

We now construct a min-tradeoff function for $\mathcal{M}_j$. Let $\tilde{p}$ be a probability distribution over $\{0,1,\perp\}$. Our goal is to find $F_{\min}$ such that 
\begin{equation}\label{eq:eat_f_min_bound_mod}
	\forall j\in[m]\qquad	F_{\min}(\tilde{p}) \leq \inf_{\sigma_{R_{j-1}R'}:\mathcal{M}_j(\sigma)_{\tilde{C}_j}=\tilde{p}} H\left( \vec{A_j} \vec{B_j} | \vec{X_j} \vec{Y_j} \vec{T_j}  R' \right)_{\mathcal{M}_j(\sigma)} \;,
\end{equation}
where $\vec{A_j}$ is a vector of varying length (but at most $s_{\max}$). We use $A_{j,i}$ to denote the $i$'th entry of $\vec{A_j}$ and $A_{j,1}^{j,i-1} = A_{j,1}\dotsc A_{j,i-1}$. Since we will only be interested in the entropy of $\vec{A_j}$ we can also describe it as a vector of length $s_{\max}$ which is initialised to be all $\perp$. For every actual round being performed in the block the value of $A_{j,i}$ is updated. Thus, the entries of  $\vec{A_j}$ which correspond to rounds which were not performed do not contribute to the entropy. We use similar notation for the other vectors of RVs.

To lower-bound the right-hand side of Equation~\eqref{eq:eat_f_min_bound_mod} we first use the chain rule
\begin{align}\label{eq:block_cr}
	H\left( \vec{A_j} \vec{B_j} | \vec{X_j} \vec{Y_j} \vec{T_j}  R' \right) = \sum_{i\in [s_{\max}]}  H(A_{j,i}B_{j,i}|\vec{X_j} \vec{Y_j} \vec{T_j}  R' A_{j,1}^{j,i-1}B_{j,1}^{j,i-1} )\;.
\end{align}
Next, for every $i\in[s_{\max}]$,
\begin{align*}
	H(A_{j,i}B_{j,i}|\vec{X_j} \vec{Y_j} \vec{T_j}  R' A_{j,1}^{j,i-1}B_{j,1}^{j,i-1} ) =&\Pr[T_{j,1}^{j,i-1} = \vec{0}] H(A_{j,i}B_{j,i}|\vec{X_j} \vec{Y_j}  R' A_{j,1}^{j,i-1}B_{j,1}^{j,i-1} T_{j,i}^{j,s_{\max}} T_{j,1}^{j,i-1} = \vec{0} ) \\
	&+ \Pr[T_{j,1}^{j,i-1} \neq \vec{0}] H(A_{j,i}B_{j,i}|\vec{X_j} \vec{Y_j}  R' A_{j,1}^{j,i-1}B_{j,1}^{j,i-1} T_{j,i}^{j,s_{\max}} T_{j,1}^{j,i-1} \neq \vec{0} ) \\
	=& (1-\gamma)^{(i-1)} H(A_{j,i}B_{j,i}|\vec{X_j} \vec{Y_j}  R' A_{j,1}^{j,i-1}B_{j,1}^{j,i-1} T_{j,i}^{j,s_{\max}} T_{j,1}^{j,i-1} = \vec{0} )
\end{align*}
since the entropy is not zero only if the $i$'th round is being performed in the block, i.e., if a test was not performed before that round. Plugging this into Eq.~\eqref{eq:block_cr} we get
\[
	H\left( \vec{A_j} \vec{B_j} | \vec{X_j} \vec{Y_j} \vec{T_j}  R' \right) = \sum_{i\in [s_{\max}]} (1-\gamma)^{(i-1)} H(A_{j,i}B_{j,i}|\vec{X_j} \vec{Y_j}  R' A_{j,1}^{j,i-1}B_{j,1}^{j,i-1} T_{j,i}^{j,s_{\max}} T_{j,1}^{j,i-1} = \vec{0} )\;.
\]

Each term in the sum can now be identified as the entropy of a single round. We can therefore use the bound derived in the main text, as given in Equation~\eqref{eq:entropy_bound_for_min_tradeoff}. For this we denote by $\omega_i$ the winning probability in the $i$'th round (given that a test was not performed before). Then it holds that
\begin{equation}\label{eq:to_mini}
	H\left( \vec{A_j} \vec{B_j} | \vec{X_j} \vec{Y_j} \vec{T_j}  R' \right) \geq \sum_{i\in [s_{\max}]} (1-\gamma)^{(i-1)} \left[1 - h\left( \frac{1}{2} + \frac{1}{2}\sqrt{16\omega_i \left(\omega_i-1\right) +3}  \right) \right] \;,
\end{equation}
where, by the actions of the EAT channel $\mathcal{M}_j$, the $\omega_i$'s must fulfil the constraint
\begin{equation}\label{eq:constraint}
	\tilde{p}(1) = \sum_{i\in[s_{\max}]} \gamma (1-\gamma)^{(i-1)}\omega_i \;.
\end{equation}

Note that, similarly to what was done in the main text,  we only need to consider $\tilde{p}$ for which $\tilde{p}(1) + \tilde{p}(0) = 1-(1-\gamma)^{s_{\max}}$ (otherwise the condition on the min-tradeoff function is trivial, as the infimum is over an empty set). 

To find the min-tradeoff function $F_{\min}(\tilde{p})$ we therefore need to minimise Equation~\eqref{eq:to_mini} under the constraint of Equation~\eqref{eq:constraint}. The following lemma shows that the minimum is achieved when all $\omega_i$ are equal.

\begin{lemma}
	The minimum of the function given on the righthand-side of Equation~\eqref{eq:to_mini} over $\omega_i$ constrained by Equation~\eqref{eq:constraint} is achieved for $\omega^*_i=\frac{\tilde{p}(1)}{1-(1-\gamma)^{s_{\max}}}$ for all $i\in[s_{\max}]$.
\end{lemma}
\begin{proof}
	Let $\vec{\omega} = \omega_1,\dotsc,\omega_{s_{\max}}$ and 
	\begin{align*}
		&f(\vec{\omega}) \equiv \sum_{i\in [s_{\max}]} (1-\gamma)^{(i-1)} \left[1 - h\left( \frac{1}{2} + \frac{1}{2}\sqrt{16\omega_i \left(\omega_i-1\right) +3}  \right) \right] \; ; \\
		&g(\vec{\omega}) \equiv \sum_{i\in[s_{\max}]} \gamma (1-\gamma)^{(i-1)}\omega_i  - \tilde{p}(1) \;. 
	\end{align*}
	Using the method of Lagrange multipliers, we should look for $\vec{\omega}^{*}$ such that $g(\vec{\omega}^{*})=0$ and $\nabla f (\vec{\omega}^{*})= - \lambda \nabla g (\vec{\omega}^{*})$ for some constant $\lambda$. $\nabla f (\vec{\omega}^{*})= - \lambda \nabla g (\vec{\omega}^{*})$ implies that for any $i$, 
	\[
		(1-\gamma)^{(i-1)} \frac{\mathrm{d}}{\mathrm{d}\omega_i}\left[1 - h\left( \frac{1}{2} + \frac{1}{2}\sqrt{16\omega_i \left(\omega_i-1\right) +3}  \right) \right] \Big|_{\omega^*_i}= - \lambda \gamma (1-\gamma)^{(i-1)} 
	\]
	and therefore
	\[
		 \frac{\mathrm{d}}{\mathrm{d}\omega_i}\left[1 - h\left( \frac{1}{2} + \frac{1}{2}\sqrt{16\omega_i \left(\omega_i-1\right) +3}  \right) \right] \Big|_{\omega^*_i}= -\lambda \gamma \;.
	\]
	The function on the left-hand side of the above equation is strictly increasing. Hence, it must be that all $\omega^*_i$ are equal to some constant $\omega^*$.
	
	Lastly,  we must have  $g(\vec{\omega}^{*})=0$. Thus,
	\[
		\sum_{i\in[s_{\max}]} \gamma (1-\gamma)^{(i-1)}\omega^*  - \tilde{p}(1) =0 
	\]
	which means
	\[
		\omega^* = \frac{\tilde{p}(1)}{\sum_{i\in[s_{\max}]} \gamma (1-\gamma)^{(i-1)}} =\frac{\tilde{p}(1)}{1-(1-\gamma)^{s_{\max}}} \;. \qedhere
	\]
\end{proof}

Plugging the minimal values of $\omega_i$ into Equation~\eqref{eq:to_mini} we get that
\begin{align*}
	H\left( \vec{A_j} \vec{B_j} | \vec{X_j} \vec{Y_j} \vec{T_j}  R' \right) &\geq \sum_{i\in [s_{\max}]} (1-\gamma)^{(i-1)} \left[1 - h\left( \frac{1}{2} + \frac{1}{2}\sqrt{16\frac{\tilde{p}(1)}{1-(1-\gamma)^{s_{\max}}} \left(\frac{\tilde{p}(1)}{1-(1-\gamma)^{s_{\max}}}-1\right) +3}  \right) \right] \\
	&= \bar{s} \left[1 - h\left( \frac{1}{2} + \frac{1}{2}\sqrt{16\frac{\tilde{p}(1)}{1-(1-\gamma)^{s_{\max}}} \left(\frac{\tilde{p}(1)}{1-(1-\gamma)^{s_{\max}}}-1\right) +3}  \right) \right] \;,
\end{align*}
where we used Equation~\eqref{eq:exp_size} to get the last equality. 

From this point we can follow the same steps as in Section~\ref{sec:entropy-chsh} (cutting and gluing the function etc.\@). The resulting min-tradeoff function is given by

\begin{align}
	&g(\tilde{p}) =  \begin{cases} 
			 \bar{s} \left[1 - h\left( \frac{1}{2} + \frac{1}{2}\sqrt{16\frac{\tilde{p}(1)}{1-(1-\gamma)^{s_{\max}}} \left(\frac{\tilde{p}(1)}{1-(1-\gamma)^{s_{\max}}}-1\right) +3}  \right) \right] &  \frac{\tilde{p}(1)}{1-(1-\gamma)^{s_{\max}}}\in\left[\frac{3}{4},\frac{2+\sqrt{2}}{4}\right] \\
			\bar{s}& \frac{\tilde{p}(1)}{1-(1-\gamma)^{s_{\max}}}\in\left[\frac{2+\sqrt{2}}{4},1\right]\;,
			\end{cases}\notag\\
	&F_{\min}\left(\tilde{p},\tilde{p}_t\right) = \begin{cases}
	g\left(\tilde{p}\right)&  \tilde{p}(1) \leq \tilde{p}_t(1) \;  \\
	\frac{\mathrm{d}}{\mathrm{d}\tilde{p}(1)} g(\tilde{p})\big|_{\tilde{p}_t}  \cdot \tilde{p}(1)+ \Big( g(\tilde{p}_t) -	\frac{\mathrm{d}}{\mathrm{d}\tilde{p}(1)} g(\tilde{p})\big|_{\tilde{p}_t} \cdot \tilde{p}_t(1) \Big)& \tilde{p}(1)> \tilde{p}_t(1)\;.
	\end{cases} \nonumber
\end{align}
The min-tradeoff function given above is effectively identical to the one derived in the main text; although it gives us a bound on the von Neumann entropy in a block, instead of a single round, this bound is exactly the expected length of a block, $\bar{s}$, times the entropy in one round.  For $s_{\max}=1$ the min-tradeoff function constructed in the main text is retrieved. 

\subsection{Modified entropy rate}

Since we apply the EAT on the blocks, the entropy rate is now defined to be the entropy \emph{per block}. We therefore get
\begin{align*}
	&\eta(\tilde{p},\tilde{p}_t,\varepsilon_{\text{s}},\varepsilon_{\text{e}}) =  F_{\min}\left(\tilde{p}, \tilde{p}_t\right) - \frac{1}{\sqrt{m}}2\left( \log (1+2\cdot 2^{s_{\max}} 3^{s_{\max}}) + \Big\lceil \frac{\mathrm{d}}{\mathrm{d}\tilde{p}(1)} g(\tilde{p})\Big|_{\tilde{p}_t}\Big\rceil  \right)\sqrt{1-2 \log (\varepsilon_{\text{s}} \cdot \varepsilon_{\text{e}})}\;, \nonumber\\
	&\eta_{\mathrm{opt}}(\varepsilon_{\text{s}}, \varepsilon_{\text{e}}) = \max_{\frac{3}{4} < \tilde{p}_t(1) < \frac{2+\sqrt{2}}{4}} \; \eta(\omega_{\mathrm{exp}}\left[ 1- (1-\gamma)^{s_{\max}} \right] - \delta_{\mathrm{est}},\tilde{p}_t,\varepsilon_{\text{s}},\varepsilon_{\text{e}}) \;,
\end{align*}
and the total amount of entropy is given by
\begin{equation}\label{eq:full_entropy_mod}
		 H^{\varepsilon_{\text{s}}}_{\min} \left( \mathbf{A B} | \mathbf{X Y T F} E \right)_{\rho_{|\Omega}} > m\cdot \eta_{\mathrm{opt}}(\varepsilon_{\text{s}},\varepsilon_{\mathrm{EA}}) = \frac{\bar{n}}{\bar{s}}\cdot \eta_{\mathrm{opt}}(\varepsilon_{\text{s}},\varepsilon_{\mathrm{EA}})  \;.
\end{equation}

By choosing  $s_{\max} = \lceil\frac{1}{\gamma}\rceil$ the scaling of the entropy rate with $\gamma$ is better than the rate derived in the main text.
In particular, a short calculation reveals that the second order term scales, roughly, as $\sqrt{\bar{n}/\gamma}$ instead of $\sqrt{n}/\gamma$.

\subsection{Modified key rate}

To get the final key rate we need to repeat the same steps from the main text, but this time applied to random variables of varying length. 

For this we first observe that, with high probability, the actual number of rounds, $n$, cannot be much larger than the expected number of rounds $\bar{n}$. Let $S_i$ be the RV describing the length of block $i$, for $i\in[m]$, and $N$ the RV describing the total number of rounds. Then $N=S_1+\dots +S_m$. Since all the $S_i$ are independent, identical, and have values in $\left[1,\frac{1}{\gamma}\right]$ we have
\[
	\Pr[N\geq \bar{n}+t] \leq \exp\left[ - \frac{2t^2\gamma^2}{m(1-\gamma)^2}\right] \;.
\]
Let $\varepsilon_t=\exp\left[ - \frac{2t^2\gamma^2}{m(1-\gamma)^2}\right]$ then
\[
	t=\sqrt{-\frac{m(1-\gamma)^2\log\varepsilon_t}{2\gamma^2}} \;.
\]

The first step in the derivation of the key rate which needs to be changed is the one given in Equation~\eqref{eq:bound_max_entropy}. The quantity that needs to be upper bounded is $H^{\frac{\varepsilon_{\text{s}}}{4}}_{\max} \left( \mathbf{B} | \mathbf{ T}  E N \right)_{\rho_{|\hat{\Omega}}} $; $N$ can be included in the entropy since its value is fixed by $\mathbf{T}$. 
By the definition of the smooth max-entropy we have
\[
	H^{\frac{\varepsilon_{\text{s}}}{4}}_{\max} \left( \mathbf{B} | \mathbf{ T}  E N \right) \leq H^{\frac{\varepsilon_{\text{s}}}{4} - \sqrt{\varepsilon_t}}_{\max} \left( \mathbf{B} | \mathbf{ T}  E N ,N\leq \bar{n}+t \right) \;.
\]
Following the same steps as in the proof of Lemma~\ref{lem:smooth_bound_qkd} we have
\[
	H^{\frac{\varepsilon_{\text{s}}}{4} - \sqrt{\varepsilon_t}}_{\max} \left( \mathbf{B} | \mathbf{ T}  E N ,N\leq \bar{n}+t \right) _{\rho_{|\hat{\Omega}}} < \gamma (\bar{n} + t) + 2\log7\sqrt{\bar{n} + t} \sqrt{1-2\log \left(( \varepsilon_{\text{s}}/4 -\sqrt{\varepsilon_t} )\cdot \left(\varepsilon_{\mathrm{EA}} + \varepsilon_{\mathrm{EC}}\right) \right)} \;.
\]
With this modification and the modified entropy rate given in Equation~\eqref{eq:full_entropy_mod} we get
\begin{equation*}
	\begin{split}
	H^{\varepsilon_{\text{s}}}_{\min} \left( \mathbf{A} | \mathbf{X Y T} O E \right)_{\tilde{\rho}_{|\tilde{\Omega}}} \geq \frac{\bar{n}}{\bar{s}} \cdot \eta_{\mathrm{opt}}\left(\varepsilon_{\text{s}}/4,\varepsilon_{\mathrm{EA}} + \varepsilon_{\mathrm{EC}}\right) - \mathrm{leak_{EC}}  -  3 \log\left(1-\sqrt{1-(\varepsilon_{\text{s}}/4)^2}\right)  \\
	- \gamma (\bar{n} + t) -2\log7 \sqrt{\bar{n} + t} \sqrt{1-2\log \left(( \varepsilon_{\text{s}}/4 -\sqrt{\varepsilon_t})\cdot \left(\varepsilon_{\mathrm{EA}} + \varepsilon_{\mathrm{EC}}\right) \right)}  \;.\qedhere
\end{split}
\end{equation*}

Similarly, the amount of leakage due to the error correction step $ \mathrm{leak_{EC}}$ should be modified as well. Following the steps in Section~\ref{sec:leakage_ec_calc}, the quantity to be upper bounded is $H_{\max}^{\frac{\varepsilon'_{\mathrm{EC}}}{2}}\left(\mathbf{A}|\tilde{\mathbf{B}}\mathbf{XYT}N\right)$. Here as well we have
\[
	H_{\max}^{\frac{\varepsilon'_{\mathrm{EC}}}{2}}\left(\mathbf{A}|\tilde{\mathbf{B}}\mathbf{XYT}N\right) \leq H_{\max}^{\frac{\varepsilon'_{\mathrm{EC}}}{2}-\sqrt{\varepsilon_t}}\left(\mathbf{A}|\tilde{\mathbf{B}}\mathbf{XYT}N,N\leq \bar{n}+t\right) \;.
\]
The asymptotic equipartition property can be used with the maximal length $\bar{n}+t$ to get
\[
	H_{\max}^{\frac{\varepsilon'_{\mathrm{EC}}}{2}-\sqrt{\varepsilon_t}}\left(\mathbf{A}|\tilde{\mathbf{B}}\mathbf{XYT}N,N\leq \bar{n}+t\right) \leq (\bar{n}+t) \cdot H(A_i|\tilde{B}_iX_i Y_i T_i) + \sqrt{\bar{n}+t}\; \delta (\varepsilon'_{\mathrm{EC}}-2\sqrt{\varepsilon_t}, \tau) \;,
\]
for $\tau =2 \sqrt{2^{H_{\max}(A_i|\tilde{B}_iX_i Y_i T_i)}}+1$ and $\delta (\varepsilon'_{\mathrm{EC}}-2\sqrt{\varepsilon_t},\tau) = 4\log \tau \sqrt{2 \log \left(8/ (\varepsilon'_{\mathrm{EC}}-2\sqrt{\varepsilon_t})^2\right)}$.
Continuing exactly as in Section~\ref{sec:leakage_ec_calc} we get
\begin{equation*}
	\begin{split}
		\mathrm{leak_{EC}} \leq (\bar{n}+t) \cdot  \left[\left( 1-\gamma \right)  h(Q) + \gamma h(\omega_{\mathrm{exp}}) \right] +  4\log \left(2\sqrt{2} +1\right)\sqrt{\bar{n}+t} \, \sqrt{2 \log \left(8/ (\varepsilon'_{\mathrm{EC}}-2\sqrt{\varepsilon_t})^2\right)}\\
		+ \log \left( 8/\varepsilon'^2_{\mathrm{EC}} + 2/\left(2-\varepsilon'_{\mathrm{EC}}\right)\right) + \log\left(\frac{1}{\varepsilon_{\mathrm{EC}}}\right) \;.
	\end{split}
\end{equation*}

The parameter $\varepsilon_t$  should be chosen such that the key rate is optimised. The resulting key rates are shown in Figures~\ref{fig:qkd_rates_Q_mod} and~\ref{fig:qkd_rates_n_mod} in the main text.

\bibliographystyle{alpha}
\bibliography{refs}

\newcommand{\etalchar}[1]{$^{#1}$}
\begin{thebibliography}{GMDLT{\etalchar{+}}13}

\bibitem[ABG{\etalchar{+}}07]{acin2007device}
Antonio Ac{\'\i}n, Nicolas Brunner, Nicolas Gisin, Serge Massar, Stefano
  Pironio, and Valerio Scarani.
\newblock Device-independent security of quantum cryptography against
  collective attacks.
\newblock {\em Physical Review Letters}, 98(23):230501, 2007.

\bibitem[AFDF{\etalchar{+}}18]{arnon2018practical}
Rotem Arnon-Friedman, Fr{\'e}d{\'e}ric Dupuis, Omar Fawzi, Renato Renner, and
  Thomas Vidick.
\newblock Practical device-independent quantum cryptography via entropy
  accumulation.
\newblock {\em Nature communications}, 9(1):459, 2018.

\bibitem[AFPS15]{arnon2015quantum}
Rotem Arnon-Friedman, Christopher Portmann, and Volkher~B Scholz.
\newblock Quantum-proof multi-source randomness extractors in the {M}arkov
  model.
\newblock {\em arXiv preprint arXiv:1510.06743}, 2015.

\bibitem[AGM06]{acin2006bell}
Antonio Ac{\'\i}n, Nicolas Gisin, and Llu{\'\i}s Masanes.
\newblock From {B}ell's theorem to secure quantum key distribution.
\newblock {\em Physical review letters}, 97(12):120405, 2006.

\bibitem[AMP06]{acin2006efficient}
Antonio Ac{\'\i}n, Serge Massar, and Stefano Pironio.
\newblock Efficient quantum key distribution secure against no-signalling
  eavesdroppers.
\newblock {\em New Journal of Physics}, 8(8):126, 2006.

\bibitem[AMP12]{acin2012randomness}
Antonio Ac{\'\i}n, Serge Massar, and Stefano Pironio.
\newblock Randomness versus nonlocality and entanglement.
\newblock {\em Physical review letters}, 108(10):100402, 2012.

\bibitem[AMPS15]{aharon2015device}
Nati Aharon, Serge Massar, Stefano Pironio, and Jonathan Silman.
\newblock Device-independent bit commitment based on the {CHSH} inequality.
\newblock {\em arXiv preprint arXiv:1511.06283}, 2015.

\bibitem[BB84]{bennett1984proceedings}
Charles~H Bennett and Gilles Brassard.
\newblock Quantum cryptography: public key distribution and coin tossing.
\newblock In {\em Proceedings of the IEEE International Conference on
  Computers, Systems, and Signal Processing, Bangalore, India, 1984}. IEEE New
  York, 1984.

\bibitem[BCK13]{barrett2013memory}
Jonathan Barrett, Roger Colbeck, and Adrian Kent.
\newblock Memory attacks on device-independent quantum cryptography.
\newblock {\em Physical review letters}, 110(1):010503, 2013.

\bibitem[BCP{\etalchar{+}}14]{brunner2014bell}
Nicolas Brunner, Daniel Cavalcanti, Stefano Pironio, Valerio Scarani, and
  Stephanie Wehner.
\newblock Bell nonlocality.
\newblock {\em Reviews of Modern Physics}, 86(2):419, 2014.

\bibitem[Bea15]{beaudry2015assumptions}
Normand~J Beaudry.
\newblock Assumptions in quantum cryptography.
\newblock {\em arXiv preprint arXiv:1505.02792}, 2015.

\bibitem[Bel64]{bell1964einstein}
John~S Bell.
\newblock On the {E}instein-{P}odolsky-{R}osen paradox.
\newblock {\em Physics}, 1(3):195--200, 1964.

\bibitem[BHK05]{barrett2005no}
Jonathan Barrett, Lucien Hardy, and Adrian Kent.
\newblock No signaling and quantum key distribution.
\newblock {\em Physical Review Letters}, 95(1):010503, 2005.

\bibitem[BMP17]{bamps2017device}
C{\'e}dric Bamps, Serge Massar, and Stefano Pironio.
\newblock Device-independent randomness generation with sublinear shared
  quantum resources.
\newblock {\em arXiv preprint arXiv:1704.02130}, 2017.

\bibitem[BOM04]{ben2004general}
Michael Ben-Or and Dominic Mayers.
\newblock General security definition and composability for quantum \&
  classical protocols.
\newblock {\em arXiv preprint quant-ph/0409062}, 2004.

\bibitem[BRG{\etalchar{+}}16]{brandao2016realistic}
Fernando~GSL Brand{\~a}o, Ravishankar Ramanathan, Andrzej Grudka, Karol
  Horodecki, Micha{\l} Horodecki, Pawe{\l} Horodecki, Tomasz Szarek, and Hanna
  Wojew{\'o}dka.
\newblock Realistic noise-tolerant randomness amplification using finite number
  of devices.
\newblock {\em Nature communications}, 7:11345, 2016.

\bibitem[BS93]{brassard1993secret}
Gilles Brassard and Louis Salvail.
\newblock Secret-key reconciliation by public discussion.
\newblock In {\em advances in Cryptology EUROCRYPT 93}, pages 410--423.
  Springer, 1993.

\bibitem[Can01]{canetti2001universally}
Ran Canetti.
\newblock Universally composable security: A new paradigm for cryptographic
  protocols 2005.
\newblock {\em Revision 3 of ECCC Report}, 2001.

\bibitem[CGJV17]{coladangelo2017verifier}
Andrea Coladangelo, Alex Grilo, Stacey Jeffery, and Thomas Vidick.
\newblock Verifier-on-a-leash: new schemes for verifiable delegated quantum
  computation, with quasilinear resources.
\newblock {\em arXiv preprint arXiv:1708.07359}, 2017.

\bibitem[CHSH69]{clauser1969proposed}
John~F Clauser, Michael~A Horne, Abner Shimony, and Richard~A Holt.
\newblock Proposed experiment to test local hidden-variable theories.
\newblock {\em Physical review letters}, 23(15):880, 1969.

\bibitem[Col06]{Colbeck09}
Roger Colbeck.
\newblock {\em Quantum And Relativistic Protocols For Secure Multi-Party
  Computation}.
\newblock PhD thesis, Trinity College, University of Cambridge, November 2006.

\bibitem[CR12]{colbeck2012free}
Roger Colbeck and Renato Renner.
\newblock Free randomness can be amplified.
\newblock {\em Nature Physics}, 8(6):450--453, 2012.

\bibitem[CSW14]{chung2014physical}
Kai-Min Chung, Yaoyun Shi, and Xiaodi Wu.
\newblock Physical randomness extractors: Generating random numbers with
  minimal assumptions.
\newblock {\em arXiv preprint arXiv:1402.4797}, 2014.

\bibitem[CY13]{coudron2013infinite}
Matthew Coudron and Henry Yuen.
\newblock Infinite randomness expansion and amplification with a constant
  number of devices.
\newblock {\em arXiv preprint arXiv:1310.6755}, 2013.

\bibitem[DF18]{dupuis2018entropy}
Fr{\'e}d{\'e}ric Dupuis and Omar Fawzi.
\newblock Entropy accumulation with improved second-order.
\newblock {\em arXiv preprint arXiv:1805.11652}, 2018.

\bibitem[DFR16]{dupuis2016entropy}
Frederic Dupuis, Omar Fawzi, and Renato Renner.
\newblock Entropy accumulation.
\newblock {\em arXiv preprint arXiv:1607.01796}, 2016.

\bibitem[DPA13]{dhara2013maximal}
Chirag Dhara, Giuseppe Prettico, and Antonio Ac{\'\i}n.
\newblock Maximal quantum randomness in {B}ell tests.
\newblock {\em Physical Review A}, 88(5):052116, 2013.

\bibitem[DPVR12]{de2012trevisan}
Anindya De, Christopher Portmann, Thomas Vidick, and Renato Renner.
\newblock Trevisan's extractor in the presence of quantum side information.
\newblock {\em SIAM Journal on Computing}, 41(4):915--940, 2012.

\bibitem[Eke91]{ekert1991quantum}
Artur~K Ekert.
\newblock Quantum cryptography based on {B}ell's theorem.
\newblock {\em Physical review letters}, 67(6):661, 1991.

\bibitem[EPR35]{einstein1935can}
Albert Einstein, Boris Podolsky, and Nathan Rosen.
\newblock Can quantum-mechanical description of physical reality be considered
  complete?
\newblock {\em Physical review}, 47(10):777, 1935.

\bibitem[ER14]{ekert2014ultimate}
Artur Ekert and Renato Renner.
\newblock The ultimate physical limits of privacy.
\newblock {\em Nature}, 507(7493):443--447, 2014.

\bibitem[FGS13]{fehr2013security}
Serge Fehr, Ran Gelles, and Christian Schaffner.
\newblock Security and composability of randomness expansion from {B}ell
  inequalities.
\newblock {\em Physical Review A}, 87(1):012335, 2013.

\bibitem[FQTL07]{fung2007phase}
Chi-Hang~Fred Fung, Bing Qi, Kiyoshi Tamaki, and Hoi-Kwong Lo.
\newblock Phase-remapping attack in practical quantum-key-distribution systems.
\newblock {\em Physical Review A}, 75(3):032314, 2007.

\bibitem[GKW15]{gheorghiu2015robustness}
Alexandru Gheorghiu, Elham Kashefi, and Petros Wallden.
\newblock Robustness and device independence of verifiable blind quantum
  computing.
\newblock {\em New Journal of Physics}, 17(8):083040, 2015.

\bibitem[GLLL{\etalchar{+}}11]{gerhardt2011full}
Ilja Gerhardt, Qin Liu, Antia Lamas-Linares, Johannes Skaar, Christian
  Kurtsiefer, and Vadim Makarov.
\newblock Full-field implementation of a perfect eavesdropper on a quantum
  cryptography system.
\newblock {\em Nature communications}, 2:349, 2011.

\bibitem[GMDLT{\etalchar{+}}13]{gallego2013full}
Rodrigo Gallego, Lluis Masanes, Gonzalo De~La~Torre, Chirag Dhara, Leandro
  Aolita, and Antonio Ac{\'\i}n.
\newblock Full randomness from arbitrarily deterministic events.
\newblock {\em Nature communications}, 4:2654, 2013.

\bibitem[GVW{\etalchar{+}}15]{giustina2015significant}
Marissa Giustina, Marijn~AM Versteegh, S{\"o}ren Wengerowsky, Johannes
  Handsteiner, Armin Hochrainer, Kevin Phelan, Fabian Steinlechner, Johannes
  Kofler, Jan-{\AA}ke Larsson, Carlos Abell{\'a}n, et~al.
\newblock Significant-loophole-free test of bell’s theorem with entangled
  photons.
\newblock {\em Physical review letters}, 115(25):250401, 2015.

\bibitem[H{\etalchar{+}}97]{hoshi1997interval}
Mamoru Hoshi et~al.
\newblock Interval algorithm for random number generation.
\newblock {\em Information Theory, IEEE Transactions on}, 43(2):599--611, 1997.

\bibitem[HBD{\etalchar{+}}15]{hensen2015loophole}
Bas Hensen, H~Bernien, AE~Dr{\'e}au, A~Reiserer, N~Kalb, MS~Blok, J~Ruitenberg,
  RFL Vermeulen, RN~Schouten, C~Abell{\'a}n, et~al.
\newblock Loophole-free {B}ell inequality violation using electron spins
  separated by 1.3 kilometres.
\newblock {\em Nature}, 526(7575):682--686, 2015.

\bibitem[HJPW04]{hayden2004structure}
Patrick Hayden, Richard Jozsa, Denes Petz, and Andreas Winter.
\newblock Structure of states which satisfy strong subadditivity of quantum
  entropy with equality.
\newblock {\em Communications in mathematical physics}, 246(2):359--374, 2004.

\bibitem[HPDF15]{hajduvsek2015device}
Michal Hajdu{\v{s}}ek, Carlos~A P{\'e}rez-Delgado, and Joseph~F Fitzsimons.
\newblock Device-independent verifiable blind quantum computation.
\newblock {\em arXiv preprint arXiv:1502.02563}, 2015.

\bibitem[HR10]{hanggi2010device}
Esther H{\"a}nggi and Renato Renner.
\newblock Device-independent quantum key distribution with commuting
  measurements.
\newblock {\em arXiv preprint arXiv:1009.1833}, 2010.

\bibitem[HRW10]{hanggi2010efficient}
Esther H{\"a}nggi, Renato Renner, and Stefan Wolf.
\newblock Efficient device-independent quantum key distribution.
\newblock In {\em Advances in Cryptology--EUROCRYPT 2010}, pages 216--234.
  Springer, 2010.

\bibitem[KAF17]{kessler2017device}
Max Kessler and Rotem Arnon-Friedman.
\newblock Device-independent randomness amplification and privatization.
\newblock {\em arXiv preprint arXiv:1705.04148}, 2017.

\bibitem[KW16]{kaniewski2016device}
Jedrzej Kaniewski and Stephanie Wehner.
\newblock Device-independent two-party cryptography secure against sequential
  attacks.
\newblock {\em New Journal of Physics}, 18(5):055004, 2016.

\bibitem[LBS{\etalchar{+}}14]{law2014quantum}
Yun~Zhi Law, Jean-Daniel Bancal, Valerio Scarani, et~al.
\newblock Quantum randomness extraction for various levels of characterization
  of the devices.
\newblock {\em Journal of Physics A: Mathematical and Theoretical},
  47(42):424028, 2014.

\bibitem[LWW{\etalchar{+}}10]{lydersen2010hacking}
Lars Lydersen, Carlos Wiechers, Christoffer Wittmann, Dominique Elser, Johannes
  Skaar, and Vadim Makarov.
\newblock Hacking commercial quantum cryptography systems by tailored bright
  illumination.
\newblock {\em Nature photonics}, 4(10):686--689, 2010.

\bibitem[LYL{\etalchar{+}}17]{liu2017high}
Yang Liu, Xiao Yuan, Ming-Han Li, Weijun Zhang, Qi~Zhao, Jiaqiang Zhong, Yuan
  Cao, Yu-Huai Li, Luo-Kan Chen, Hao Li, et~al.
\newblock High speed self-testing quantum random number generation without
  detection loophole.
\newblock In {\em Frontiers in Optics}, pages FTh2E--1. Optical Society of
  America, 2017.

\bibitem[Mas09]{masanes2009universally}
Llu{\'\i}s Masanes.
\newblock Universally composable privacy amplification from causality
  constraints.
\newblock {\em Physical review letters}, 102(14):140501, 2009.

\bibitem[MPA11]{masanes2011secure}
Llu{\'\i}s Masanes, Stefano Pironio, and Antonio Ac{\'\i}n.
\newblock Secure device-independent quantum key distribution with causally
  independent measurement devices.
\newblock {\em Nature communications}, 2:238, 2011.

\bibitem[MRC{\etalchar{+}}14]{masanes2014full}
Llu{\'\i}s Masanes, Renato Renner, Matthias Christandl, Andreas Winter, and
  John Barrett.
\newblock Full security of quantum key distribution from no-signaling
  constraints.
\newblock {\em Information Theory, IEEE Transactions on}, 60(8):4973--4986,
  2014.

\bibitem[MS14a]{miller2014robust}
Carl~A Miller and Yaoyun Shi.
\newblock Robust protocols for securely expanding randomness and distributing
  keys using untrusted quantum devices.
\newblock In {\em Proceedings of the 46th Annual ACM Symposium on Theory of
  Computing}, pages 417--426. ACM, 2014.

\bibitem[MS14b]{miller2014universal}
Carl~A Miller and Yaoyun Shi.
\newblock Universal security for randomness expansion from the spot-checking
  protocol.
\newblock {\em arXiv preprint arXiv:1411.6608}, 2014.

\bibitem[MY98]{mayers1998quantum}
Dominic Mayers and Angela Yao.
\newblock Quantum cryptography with imperfect apparatus.
\newblock In {\em Foundations of Computer Science, 1998. Proceedings. 39th
  Annual Symposium on}, pages 503--509. IEEE, 1998.

\bibitem[NC02]{nielsen2002quantum}
Michael~A Nielsen and Isaac Chuang.
\newblock Quantum computation and quantum information, 2002.

\bibitem[NPA08]{navascues2008convergent}
Miguel Navascu{\'e}s, Stefano Pironio, and Antonio Ac{\'\i}n.
\newblock A convergent hierarchy of semidefinite programs characterizing the
  set of quantum correlations.
\newblock {\em New Journal of Physics}, 10(7):073013, 2008.

\bibitem[PAB{\etalchar{+}}09]{pironio2009device}
Stefano Pironio, Antonio Ac{\'\i}n, Nicolas Brunner, Nicolas Gisin, Serge
  Massar, and Valerio Scarani.
\newblock Device-independent quantum key distribution secure against collective
  attacks.
\newblock {\em New Journal of Physics}, 11(4):045021, 2009.

\bibitem[PAM{\etalchar{+}}10]{pironio2010random}
Stefano Pironio, Antonio Ac{\'\i}n, Serge Massar, A~Boyer de~La~Giroday,
  Dzimitry~N Matsukevich, Peter Maunz, Steven Olmschenk, David Hayes, Le~Luo,
  T~Andrew Manning, et~al.
\newblock Random numbers certified by {B}ell's theorem.
\newblock {\em Nature}, 464(7291):1021--1024, 2010.

\bibitem[PM13]{pironio2013security}
Stefano Pironio and Serge Massar.
\newblock Security of practical private randomness generation.
\newblock {\em Physical Review A}, 87(1):012336, 2013.

\bibitem[PR14]{portmann2014cryptographic}
Christopher Portmann and Renato Renner.
\newblock Cryptographic security of quantum key distribution.
\newblock {\em arXiv preprint arXiv:1409.3525}, 2014.

\bibitem[Ren05]{Ren05}
Renato Renner.
\newblock {\em Security of Quantum Key Distribution}.
\newblock PhD thesis, Swiss Federal Institute of Technology Zurich, September
  2005.

\bibitem[RK05]{renner2005universally}
Renato Renner and Robert K{\"o}nig.
\newblock Universally composable privacy amplification against quantum
  adversaries.
\newblock In {\em Theory of Cryptography}, pages 407--425. Springer, 2005.

\bibitem[RMW17]{ribeiro2017fully}
J{\'e}r{\'e}my Ribeiro, Gl{\'a}ucia Murta, and Stephanie Wehner.
\newblock Fully device independent conference key agreement.
\newblock {\em arXiv preprint arXiv:1708.00798}, 2017.

\bibitem[RUV13]{reichardt2013classical}
Ben~W Reichardt, Falk Unger, and Umesh Vazirani.
\newblock Classical command of quantum systems.
\newblock {\em Nature}, 496(7446):456--460, 2013.

\bibitem[RW05]{renner2005simple}
Renato Renner and Stefan Wolf.
\newblock Simple and tight bounds for information reconciliation and privacy
  amplification.
\newblock In {\em Advances in cryptology-ASIACRYPT 2005}, pages 199--216.
  Springer, 2005.

\bibitem[Sca13]{scarani2013device}
Valerio Scarani.
\newblock The device-independent outlook on quantum physics (lecture notes on
  the power of {B}ell's theorem).
\newblock {\em arXiv preprint arXiv:1303.3081}, 2013.

\bibitem[SGB{\etalchar{+}}06]{scarani2006secrecy}
Valerio Scarani, Nicolas Gisin, Nicolas Brunner, Llu{\'\i}s Masanes, Sergi
  Pino, and Antonio Ac{\'\i}n.
\newblock Secrecy extraction from no-signaling correlations.
\newblock {\em Physical Review A}, 74(4):042339, 2006.

\bibitem[SLT{\etalchar{+}}18]{shen2018randomness}
Lijiong Shen, Jianwei Lee, Le~Phuc Tinh, Jean-Daniel Bancal, Alessandro
  Cer{\`e}, Antia Lamas-Linares, Adriana Lita, Thomas Gerrits, Sae~Woo Nam,
  Valerio Scarani, et~al.
\newblock Randomness extraction from bell violation with continuous parametric
  down conversion.
\newblock {\em arXiv preprint arXiv:1805.02828}, 2018.

\bibitem[SMSC{\etalchar{+}}15]{shalm2015strong}
Lynden~K Shalm, Evan Meyer-Scott, Bradley~G Christensen, Peter Bierhorst,
  Michael~A Wayne, Martin~J Stevens, Thomas Gerrits, Scott Glancy, Deny~R
  Hamel, Michael~S Allman, et~al.
\newblock Strong loophole-free test of local realism.
\newblock {\em Physical review letters}, 115(25):250402, 2015.

\bibitem[SR08a]{scarani2008quantum}
Valerio Scarani and Renato Renner.
\newblock Quantum cryptography with finite resources: Unconditional security
  bound for discrete-variable protocols with one-way postprocessing.
\newblock {\em Physical review letters}, 100(20):200501, 2008.

\bibitem[SR08b]{scarani2008security}
Valerio Scarani and Renato Renner.
\newblock Security bounds for quantum cryptography with finite resources.
\newblock In {\em Theory of Quantum Computation, Communication, and
  Cryptography}, pages 83--95. Springer, 2008.

\bibitem[TCR09]{tomamichel2009fully}
Marco Tomamichel, Roger Colbeck, and Renato Renner.
\newblock A fully quantum asymptotic equipartition property.
\newblock {\em Information Theory, IEEE Transactions on}, 55(12):5840--5847,
  2009.

\bibitem[TCR10]{tomamichel2010entropyduality}
Marco Tomamichel, Roger Colbeck, and Renato Renner.
\newblock Duality between smooth min- and max-entropies.
\newblock {\em IEEE Transactions on Information Theory}, 56(9):4674--4681,
  2010.

\bibitem[Tom12]{tomamichel2012framework}
Marco Tomamichel.
\newblock A framework for non-asymptotic quantum information theory.
\newblock {\em arXiv preprint arXiv:1203.2142}, 2012.

\bibitem[Tom15]{tomamichel2015quantum}
Marco Tomamichel.
\newblock Quantum information processing with finite resources-mathematical
  foundations.
\newblock {\em arXiv preprint arXiv:1504.00233}, 2015.

\bibitem[TSSR11]{tomamichel2011leftover}
Marco Tomamichel, Christian Schaffner, Adam Smith, and Renato Renner.
\newblock Leftover hashing against quantum side information.
\newblock {\em Information Theory, IEEE Transactions on}, 57(8):5524--5535,
  2011.

\bibitem[VV12]{vazirani2012certifiable}
Umesh Vazirani and Thomas Vidick.
\newblock Certifiable quantum dice: or, true random number generation secure
  against quantum adversaries.
\newblock In {\em Proceedings of the forty-fourth annual ACM symposium on
  Theory of computing}, pages 61--76. ACM, 2012.

\bibitem[VV14]{vazirani2014fully}
Umesh Vazirani and Thomas Vidick.
\newblock Fully device-independent quantum key distribution.
\newblock {\em Physical review letters}, 113(14):140501, 2014.

\bibitem[WKR{\etalchar{+}}11]{weier2011quantum}
Henning Weier, Harald Krauss, Markus Rau, Martin F{\"u}rst, Sebastian Nauerth,
  and Harald Weinfurter.
\newblock Quantum eavesdropping without interception: an attack exploiting the
  dead time of single-photon detectors.
\newblock {\em New Journal of Physics}, 13(7):073024, 2011.

\end{thebibliography}

\end{document}